%% file: main.tex
\documentclass[pra,aps,twocolumn,groupedaddress,nofootinbib,superscriptaddress,preprintnumbers]{revtex4-2}
\include{ethans_preamble.tex}

\usepackage[caption=false]{subfig}
\usepackage{enumitem}

\usepackage{amsthm}
\usepackage{physics}
\usepackage{booktabs}
\usepackage{bm}

\DeclarePairedDelimiter\ceil{\lceil}{\rceil}

\DeclarePairedDelimiterX{\infdivx}[2]{(}{)}{%
	#1\;\delimsize\|\;#2%
}
\newtheorem*{theorem*}{Theorem}

\newtheorem{definition}{Definition}

\renewcommand\cp{\Phi}

\usepackage{accents}

\newcommand{\oEE}{\mathop{\mathbb{E}}}

\newcommand{\dyn}{{\sf Dyn}}

\renewcommand{\tth}{t_{\rm th}}
\newcommand{\poly}{{\rm poly}}
\begin{document}

        \title{Exponentially slow thermalization in 1D fragmented dynamics}
	
	\author{Cheng Wang}
	\affiliation{School of Physics, Peking University, Beijing 100871, China}

    \author{Shankar Balasubramanian}
    \email{sbalasub@mit.edu}
        \affiliation{Center for Theoretical Physics, Massachusetts Institute of Technology, Cambridge, Massachusetts 02139, USA}

        \author{Yiqiu Han}
	\affiliation{Department of Physics, Boston College, Chestnut Hill, MA 02467, USA}
        \affiliation{Department of Physics and Center for Theory of Quantum Matter, University of Colorado Boulder, Boulder, Colorado 80309 USA}

        \author{Ethan Lake}
        \email{elake@berkeley.edu}
        \affiliation{Department of Physics, University of California Berkeley, Berkeley, CA 94720, USA}

            \author{Xiao Chen}
	\affiliation{Department of Physics, Boston College, Chestnut Hill, MA 02467, USA}

        \author{Zhi-Cheng Yang}
        \email{zcyang19@pku.edu.cn}
        \affiliation{School of Physics, Peking University, Beijing 100871, China}
        \affiliation{Center for High Energy Physics, Peking University, Beijing 100871, China}

	\begin{abstract}
         We investigate the thermalization dynamics of 1D systems with local constraints coupled to an infinite temperature bath at one boundary. The coupling to the bath eventually erases the effects of the constraints, causing the system to tend towards a maximally mixed state at long times. We show that for a large class of local constraints, the time at which thermalization occurs can be extremely long. In particular, we present evidence for the following conjecture: when the constrained dynamics displays strong Hilbert space fragmentation, the thermalization time diverges exponentially with system size. We show that this conjecture holds for a wide range of dynamical constraints, including dipole-conserving dynamics, the $tJ_z$ model,  and a large class of group-based dynamics, and relate a general proof of our conjecture to a different conjecture about the existence of certain expander graphs. 
	\end{abstract}
	
	\maketitle
	\tableofcontents

	\section{Introduction and overview}
	Most physical systems thermalize: when prepared in a generic initial state, they relax to a universal equilibrium state determined by a small number of thermodynamic variables. It is of great interest to characterize systems where thermalization takes a ``long'' time, or even fails to occur altogether \cite{deutsch1991quantum, srednicki1994chaos, rigol2008thermalization, nandkishore2015many-body, dalessio2016from}. Two well-studied examples of such systems are strongly disordered systems with complex energy landscapes (MBL, spin glasses, etc.~\cite{alet2018many-body, abanin2019colloquium}), and integrable systems.  A third class of systems are those with {\it dynamical constraints}, where hard restrictions are placed on the allowed local transition rules governing the dynamics. In these systems the presence of constraints shatters the configuration space into disconnected regions, a phenomenon known as {\it Hilbert space fragmentation} (HSF) \cite{khemani2020localization,sala2020ergodicity}.  By now, a diverse array of systems displaying HSF are known, both for classical and quantum systems \cite{motrunich22hilbert,pancotti2020quantum,van2015dynamics,mukherjee2021minimal,yang2020hilbert,brighi2023hilbert,lan2018quantum,garrahan2018aspects, detomasi2019dynamics,moudgalya2022thermalization}. These systems remain non-ergodic for all times, and---like in integrable systems---require an extensive number of quantities to label the equilibrium state. 
    
    More formally, Hilbert space fragmentation is identified when the many-body Hilbert space $\mch$ can be decomposed as 
	\be \label{splitting} \mch = \bigoplus_\a \mck_\a,\ee 
	where the different sectors $\mck_\a$ are subspaces which admit a basis of weakly entangled states, and the constraints prevent any states from different sectors from mixing for all times. One example where this occurs is when a system has a global symmetry; however, the more interesting cases are when not all of the sectors can be enumerated by the quantum numbers of standard global symmetries\footnote{We add the qualifier ``standard'' here because there exist certain kinds of modulated symmetries \cite{sala2022dynamics} which split $\mch$ in a similar way to Hilbert space fragmented systems (see Sec.~\ref{ss:breakdown}). } 
	, as otherwise such systems typically fall under conventional thermalization paradigms.
	
	While constraints provide an interesting way to arrest thermalization, they are undoubtedly fine-tuned in the strictest sense, as violating the constraints even weakly will generically render the system's dynamics ergodic at infinite times.  It is thus an interesting question to ask whether certain signatures of fragmentation remain even in these perturbed systems, as manifested e.g. by anomalously long thermalization times. 
    % -- one example is if such a system eventually thermalizes but requires exponential time to explore the full state space.  
    
    In this paper, we will study a model where a 1D fragmented system of length $L$ is subject to maximally depolarizing noise on an $O(1)$ number of sites (which for a large portion of the paper will be located at one end of an open chain), which can be regarded as a way of studying the system's thermalization dynamics in the presence of a local coupling to an infinite temperature bath, or as a way of studying the thermalization of subsystems in unperturbed fragmented models.\footnote{We note that similar types of setups have been used to study the stability of MBL against a thermal bubble embedded in the system (see e.g.~\cite{nandkishore2017many}) as well as to investigate how spatially isolated noise influences entanglement dynamics in random unitary circuits \cite{lovas2023quantum} .}
    
    Adding depolarizing noise  to {\it all} the sites in the system would cause it to decohere to a maximally mixed state after $O(\log L)$ time by standard quantum information theory arguments \cite{aharonov1996limitations} (keeping the noise local but removing the constraints on the dynamics also results in rapid thermalization \cite{lovas2023quantum}).  However, as we will see, only subjecting an $O(1)$ number of sites to noise can dramatically increase the thermalization time, due to the presence of dynamical constraints that arrest how the effects of the noise can spread across the system.

    For fragmented systems, this setup was first studied in Ref.~\cite{hsf_impurities} in the context of 1D random unitary dynamics exhibiting a ``pair flip'' constraint \cite{caha2018pair}. It was shown that, while the coupling to the bath eventually heats the system to a maximally mixed state, this process takes an {\it exponentially} long time (in $L$) to occur. This was identified as being due to the connectivity of the configuration space, which in the pair flip model features structural bottlenecks that result in in the slow diffusion of initial product states across Hilbert space, yielding a way of arresting thermalization qualitatively distinct from other
    % .  This phenomenon is rather distinct from other thermalization-arresting 
    mechanisms such as integrability and disorder-induced localization.
	
	\begin{table}\label{table:results}
		\begin{center}
			\begin{tabular}{c c c c c}
				\toprule
				Constraint & Class & Fragmentation & Group? & $L_{\rm erg}$ \\ [0.5ex]
				\midrule
				PXP & 0 & Exp & N & $\infty$  \\
				Spin-1 breakdown (\ref{ss:breakdown}) & I & Sym & N  & $\infty$ \\
				Pair flip (\cite{hsf_impurities}) & II & Exp & Y & $O(L)$ \\
%				Colored Motzkin & C & Exp & N & Exp \\
				$tJ_z$ (\ref{ss:tjz}) & II & Exp & N & $O(2^L)$ \\
				Dipole (\ref{ss:dipole}) & II & Exp & N & $\infty$ \\
				Hyperbolic groups (\ref{ss:hyperbolic}) & II & Exp & Y & $O(L)$ \\
			\end{tabular}
		\end{center}
		\caption{The main examples of HSF systems studied in this paper. For each type of constrained dynamics we indicate the model's ``class'' (class 0 systems are non-ergodic even after adding boundary noise, while classes I and II display slow thermalization for distinct reasons), the type of fragmentation (`Exp' denotes exponentially strong fragmentation and `Sym' denotes a global symmetry), whether the constraints arise from a group structure, and the ergodicity length.}
	\end{table}

	In this paper, we make progress towards obtaining a general understanding of which kinds of constrained systems exhibit 
    exponentially  slow thermalization dynamics when coupled to boundary noise.  
    We conjecture a simple sufficient condition for exponentially slow thermalization to occur, based only on the size $|\mck_{\max}|$ of the largest Krylov sector. 
    %This means that the size $|\mck_{\max}|$ of the largest Krylov sector. 
    We will say that the dynamics is {\it exponentially fragmented} if $|\mck_{\max}| / |\mch| = O(\exp(-L))$, while we say it is {\it polynomially fragmented} if $|\mck_{\max}| / |\mch| = 1/\poly(L)$.  Our main contribution is to formulate and provide evidence for the following conjecture: {\it for exponentially fragmented dynamics with depolarizing noise acting on an $O(1)$ sized subregion near the boundary, the thermalization time of typical initial computational-basis product states is either i) infinite, or ii) exponential in $L$}.\footnote{In this definition, the thermalization time is defined as the smallest time at which the system's density matrix is $\ep$-close (in 1-norm distance) to the steady state, for a fixed constant $\ep <1$. }  A similar conjecture can be made for thermalization times of subsystems, with system size replaced by the subsystem size.  If true, this conjecture would imply that dynamical phenomena like thermalization times can be universally deduced purely from structural properties of the constraints.
	
	There are three distinct mechanisms behind slow thermalization in exponentially fragmented systems, and we accordingly divide such systems into three classes. In class 0, fragmentation persists even under the presence of boundary depolarizing noise, which is unable to completely restore ergodicity. In class I, depolarizing noise renders the dynamics ergodic, but exponentially many (in $L$) steps of the dynamics are required to move between different Krylov sectors. Consider forming a graph---referred to hereafter as the ``connectivity graph''---whose vertices are product states and whose edges represent allowed transitions between states induced by the dynamics. Models in class I produce connectivity graphs with exponentially large diameters. Finally, systems in class II have ergodic dynamics and a connectivity graph with a sub-exponential diameter; the slowness of thermalization in these systems is instead due to {\it bottlenecks} which occur in the connectivity graph.  
    
    We show that this conjecture is true in many exponentially fragmented systems,\footnote{Our conjecture does {\it not} necessarily hold when the fragmented system is subjected to noise at {\it both} ends of the system, or is subjected to periodic boundary conditions. In these cases, we find exponentially fragmented examples where $\tth = \ct(\poly(L))$.}
	such as dipole conserving models \cite{khemani2020localization,rakovszky2020statistical}, the $tJ_z$ model \cite{zhang1997tjz,Batista2000tjz,sala2020ergodicity,moudgalya2022quantum}, the pair-flip model~\cite{caha2018pair, hsf_impurities}, and the colored Motzkin chain \cite{movassagh}.  Table~\ref{table:results} illustrates the particular models we study in this work, which class they belong to, and their thermalization times.  
	
	A major step that we take towards proving our conjecture is to relate it to a different conjecture in the theory of expander graphs. In Sec.~\ref{sec:groups}, we show how our conjecture can be formulated in terms of the conductance of expander graphs weighted by heat kernels (i.e. distributions of random walks on these graphs).  Exponentially long thermalization times of systems with dynamical constraints can be translated into showing that the conductance of these graphs scales like $O(\exp(-L))$.  A recent result by Fraczyk and van Limbeek~\cite{fraczyk2019heat} proves a version of a conjecture of Benjamini~\cite{benjamini1998expanders} and shows that the conductance must vanish in the thermodynamic limit, therefore implying long (but not necessarily {\it exponentially} long) thermalization times; our conjecture therefore amounts to a stronger version of Benjamini's conjecture.  We prove that the conductance vanishes like $O(\exp(-L))$ for a large class of dynamical constraints arising from multiplication laws in hyperbolic groups (see Ref.~\cite{balasubramanian2023glassy}).  The groups used to construct these models are quite generic in that {\it any} randomly chosen constraint which derives from a group multiplication law (of which the pair flip model is an example) is either trivial (polynomially fragmented) or hyperbolic with high probability.  We leave a full proof or disproof of our conjecture to future work.

  Another motivation for the boundary depolarizing model studied in this paper is that it could serve as a crude model for subsystem dynamics. 
   Indeed, one might be tempted to regard depolarizing boundary noise as a way of mimicking the dynamics experienced by a subsystem $A$ when coupled to a sufficiently large reservoir $A^c$, with the entire system $A\cup A^c$ undergoing constraint-preserving unitary dynamics. Nonetheless, in systems displaying HSF, these two scenarios can be quite different, with the subsystem dynamics sometimes never exhibiting thermalization, or requiring a bath of anomalously large size in order for thermalization to occur.  This scenario was first pointed out in Ref.~\cite{balasubramanian2023glassy}.  In particular, we define the \textit{ergodicity length} $L_{\rm erg}(|A|)$ as the minimal size system in which the subsystem $A$ must be embedded so that the dynamics on $A$ is ergodic and qualitatively similar to the dynamics induced by maximally depolarizing boundary noise. In generic chaotic systems, one typically expects $L_{\rm erg}/|A| \sim O(1)$. We compute the ergodicity length for the various models studied in this paper, finding examples where $L_{\rm erg}$ scales as either ${\rm poly}(L)$ or $\exp(L)$, as well as ones where $L_{\rm erg}=\infty$ (Table~\ref{table:results}).
    
	An outline of the remainder of the paper is as follows. We begin in Sec.~\ref{sec:setup} by establishing basic definitions, introducing the concept of a Krylov graph, and precisely describing the class of dynamics we will study in the remainder of the work. Our main conjecture regarding exponential fragmentation and slow dynamics is then formulated in Sec.~\ref{sec:graphs}.
    In Sec.~\ref{sec:ex} we prove the conjecture in a variety of models across all different classes, and with different scaling behaviors of the ergodicity length $L_{\rm erg}$. 
    In Sec.~\ref{sec:groups} we prove of the conjecture for a large class of dynamics based on group multiplication laws, and provide a discussion of the mathematical results needed to prove the conjecture in full generality.
    Sec.~\ref{sec:fragtrans} contains a discussion of models that exhibit a strong to weak fragmentation transition as a function of the expectation value of a global symmetry charge, for which we show thermalization is fast. Sec.~\ref{sec:disc} concludes with a discussion of open problems and future research directions.

	\section{General setup}\label{sec:setup}
	
	\ss{Polynomial versus exponential fragmentation}
	
	Throughout this work, we will refer to a particular model of dynamics using the symbol $\dyn$, with $\dyn(t)$ denoting the quantum channel implementing evolution under $\dyn$ for time $t$. Unless stated otherwise, $\dyn$ will be taken to act on a Hilbert space $\mch$ associated with an $L$-site qudit chain with open boundary conditions. We will focus throughout on discrete-time dynamics generated by random unitary circuits subjected to a particular type of constraint, since it allows us to readily make analytic progress; we expect many of our results to also hold for constrained Hamiltonian dynamics, Floquet dynamics, and classical reversible Markov chain dynamics. 
	
	The constraints present in $\dyn$ ``fragment'' $\mch$ via \eqref{splitting}, where the Krylov sectors $\mck_\a$ denote the irreducible subspaces preserved by $\dyn(t)$. The dimension of these spaces will be denoted by $|\mck_\a|$, the number of Krylov sectors by $N_\mck$, and the sector with the largest dimension by $\mck_{\rm max}$. When each of the $\mck_\a$ admit an orthonormal basis of product states, $\dyn$ is said to be {\it classically fragmented}; when this is not the case $\dyn$ is said to be {\it quantum fragmented}, following~\cite{moudgalya2022quantum}. 

	In this work we will find it useful to distinguish between the cases when $\mck_{\rm max}$ constitutes a polynomially small fraction of $\mch$, or an exponentially small fraction.\footnote{If $\mck_{\rm max}$ is an exponentially small fraction of $\mch$ then the number of Krylov sectors is exponentially large. Having $N_\mck = \O(\exp(L))$ is of course however still possible even if $\mck_{\rm max}$ is a polynomially large fraction of $\mch$. }
    We will say that (following terminology first appearing in \cite{balasubramanian2023glassy}) $\dyn$ is
	\begin{enumerate}
		\item {\it polynomially fragmented} if 
		\be |\mck_{\max}| / |\mch| = \O(1/\poly(L)),\ee 
		and 
		\item {\it exponentially fragmented} if 
		\be |\mck_{\max}| / |\mch| = O(\exp(-L)).\ee 
	\end{enumerate}
	There may be examples which are neither polynomially nor exponentially fragmented according to this definition (with $|\mck_{\rm max}| / |\mch|$ scaling as $\exp(-L^{\beta < 1})$) Ref.~\cite{balasubramanian2023glassy}, but we will not explicitly address such examples in this work.   
	
	Note that when the dynamics possesses global symmetries we do not restrict ourselves to a symmetry sector to define fragmentation, as in Refs.~\cite{moudgalya2022quantum,balasubramanian2023glassy} . This is because when $\dyn$ has symmetries, the ratio of $|\mck_{\max}|$ to $\mch$ (and not to the size of the symmetry sector to which $\mck_{\max}$ belongs) is what determines thermalization timescales for the dynamics considered in most this paper.  This distinction is particularly important in the exponentially-fragmented ``breakdown model'' of Sec.~\ref{ss:breakdown}, in which each fragment is uniquely identified with a global symmetry sector. However, in some settings focusing on a single symmetry sector is meaningful; examples will be discussed in Sec.~\ref{sec:fragtrans}. Until then, we will stick with the above definition. 
	
	The second remark is that our definition above does not distinguish between systems where $|\mck_{\max}| / |\mch|$ is asymptotically constant and those where it vanishes as $1/\poly(L)$---both are treated as ``polynomially fragmented''. In the literature, situations where $|\mck_{\max}| / |\mch|$ is asymptotically constant (or, more often, where the size of $\mck_{\max}$ is a constant fraction of the global symmetry sector to which it belongs) are referred to as ``weakly fragmented'', while situations where $|\mck_{\max}| / |\mch|$ vanishes as $L\ra\infty$ (regardless of how quickly) are referred to as ``strongly fragmented''. For us, the distinction between $1/\poly(L)$ and $\exp(-L)$ scaling will be more important.

	\ss{Model of dynamics} \label{ss:model}

	In this section we introduce and motivate the particular models of dynamics we will study in subsequent sections (see Fig.~\ref{fig:schematic} for a schematic). We aim to derive lower bounds on thermalization times (defined below) which are as general as possible, taking as input only the structure of the constraints, and not (for our purposes extraneous) details such as the exact form of some Hamiltonian, a particular representation of a set of Kraus operators, etc. For this reason we will take $\dyn$ to be generated by an appropriate form of discrete-time constrained random unitary dynamics (whose gates will be averaged over in the channel implemented by $\dyn$) and an appropriate type of depolarizing boundary noise. Since we are interested in placing lower bounds on thermalization times, we will choose the structure of the unitary circuits to be as rapidly thermalizing as possible, given the constraints.

	For each choice of dynamics we will be interested in studying the thermalization time $\tth$, defined as 
	\be \label{tth} \tth = \max_{\r} \min\{ t \, : \, || \r(t) - \unit / |\mch|||_1 < 1/2\},\ee 
	where $\r(t) \equiv \dyn^t(\r)$, and the maximum is over all initial density matrices contained in the support of a {\it single} Krylov sector $\mck_\a$ (i.e. $\r$ such that $\Tr[\r \Pi_\b] = \d_{\b\a}$ for some $\a$, where $\Pi_\b$ is the projector onto $\mck_\b$). 
    % and $\ep$ is a small $O(1)$ constant whose exact value is unimportant. 
    This definition assumes that thermalization always occurs to the maximally mixed state, i.e. that $\dyn$ thermalizes if $\dyn^\infty(\r) = \unit / |\mch|$ for all $\r$. This assumption is correct in the minimally-structured models of dynamics considered below but would need to be modified if $\dyn$ possesses additional conserved quantities (e.g. energy). In general however we expect a bound on $\tth$ in this setting to also provide bounds on $\tth$ in situations where $\dyn$ has a more structured fixed point (see below). 
	
	We note in passing that the definition \eqref{tth} means that if $t<\tth$, there exists {\it some} observable $\mco$ that distinguishes $\r(t)$ from $\unit/|\mch|$ with probability greater than $\ep$. While a priori there need not exist a {\it local} observable which does this, a local $\mco$ can in fact be found in all of the examples studied in Sec.~\ref{sec:ex}. 
	
	\sss{Local constraint breaking} \label{sss:local_constr_breaking}
	
	As mentioned above, our primary focus will be on situations where the constraints are violated in a local connected region of space. When studying this situation we will mostly work with open boundary conditions, and for simplicity will take the region in which the constraints are broken to be located at one boundary of the system. The choice of open boundary conditions is important in general, as there exist some models for which the imposition of periodic boundary conditions speeds up thermalization by an amount exponential in $L$ (such as the $tJ_z$ model, see Sec.~\ref{ss:tjz}). 
    
	We will take each time step of $\dyn$ to act as a constraint-preserving quantum channel $\mcc_{\rm bulk}$ acting on the whole system, followed by a constraint-breaking channel $\mcc_{\rm bound}$ acting nontrivially only near the boundary:
	\be \dyn = \mcc_{\rm boundary} \circ \mcc_{\rm bulk}.\ee 
	On physical grounds, we expect that the quickest way to thermalize by breaking the constraints is to apply maximally depolarizing noise in the constraint-breaking region, which effectively replaces the constraint-breaking region by an infinite-temperature thermal bath.\footnote{Indeed, one can show show that among all possible choices of $\mcc_{\rm boundary}$, this choice brings the system's density matrix closer (in trace distance) to the maximally mixed state than any other channel.} We will thus set $\mcc_{\rm boundary} = \mcd_1 \tp \unit^{L-1}$, where $\unit$ denotes the identity channel on a single site, and $\mcd_1$ applies maximally depolarizing noise to the first site (under this choice of $\mcc_{\rm boundary}$, we may take the noise to act only on the first site without loss of generality). 	
	
	Similarly, we expect that thermalization will be fastest when $\mcc_{\rm bulk}$ globally scrambles all states within a given Krylov sector, i.e. when 
	\be \label{mccbulk} \mcc_{\rm bulk} = \bigoplus_\a \mcd_\a,\ee 
	where $\mcd_\a$ depolarizes states in $\mck_\a$; explicitly, $\mcd_\a(\proj\psi) = \Pi_\a / |\mck_\a|$ for all $\k\psi\in \mck_\a$. Physically, this choice of $\mcc_{\rm bulk}$ corresponds to acting with a deep constrainted RU circuit at each time step, and then averaging over circuit realizations. We will employ this choice of $\mcc_{\rm bulk}$ when making analytic statements, but for the numerics to follow we will instead take $\mcc_{\rm bulk}$ to be a local channel, each gate of which is drawn from an appropriately constrained Haar ensemble (as in Refs.~\cite{Singh,balasubramanian2023glassy,hsf_impurities}). Intuitively, we expect that imposing locality will only slow the dynamics more because the system will not instantaneously thermalize within a sector~\cite{hsf_impurities}.  In Appendix~\ref{app:local}, we prove this in generality; thus, proving slow thermalization in the model with instantaneous intrasector thermalization also implies slow thermalization under generic local dynamics.
	
	The choice \eqref{mccbulk} means that the internal structure of each $\mck_\a$ is not important---at the end of every time step, $\r(t)$ is always decohered and spread out uniformly across any given sector, and all of the information in $\r(t)$ is contained in the quantities 
	\be \label{pongk} p_\a(t) \equiv \Tr[\r(t) \Pi_\a], \ee 
	which define a probability distribution over the set of Krylov sectors. As we will see in Sec.~\ref{sec:graphs}, this means that the thermalization dynamics of $\dyn$ can be computed by studying the thermalization of a certain Markov process (see also Refs.~\cite{Singh,balasubramanian2023glassy,hsf_impurities}), whose mixing time can be bounded using standard graph theory techniques. 
	
	\begin{figure}
		\includegraphics[width=.43\tw]{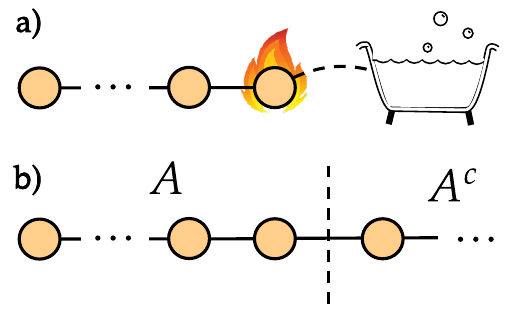}
		\caption{\label{fig:schematic} A schematic of the general setup. In {\bf a)}, an open 1d chain is coupled to a thermal bath at one of its boundaries, which induces depolarizing noise on one of its boundary sites. We conjecture that in all exponentially fragmented dynamics, the thermalization time $\tth$ is either infinite, or scales exponentially with system size. In {\bf b)}, a finite chain is bipartitioned into $A$ and $A^c$, with the state on $A^c$ now playing the role of the bath. This dynamics generically thermalizes at least as slowly as the former type, although thermalization may only be possible when $|A^c| / |A|$ diverges sufficiently quickly as $|A|\ra\infty$. }
	\end{figure}

	\sss{Subsystem dynamics and ergodicity lengths}
	
	The depolarizing noise model above corresponds to locally coupling the system to an infinitely-large infinite temperature heat bath. We can also consider the dynamics of the reduced density matrix $\r_A$ on a subsystem $A$ induced by constraint-{\it preserving} dynamics applied to the full system $A \cup A^c$, where $\r_A(t) = \Tr_{A^c}[\dyn^t(\r)]$ with $\dyn$ now given by fully constraint-preserving random unitary dynamics (with the circuit average incorporated into $\dyn$ as above, and with $A$ taken to be e.g. the left half of the system). For generic chaotic dynamics, the degrees of freedom in $A^c$ act as a thermal bath for the degrees of freedom in $A$ as long as $A^c$ is large enough, and so we expect the reduced density matrix $\r_A$---or more precisely, the diagonal matrix elements thereof \footnote{Off-diagonal matrix elements can behave differently; for example in the present situation $\Pi_\a \r_A(t) \Pi_\b = 0$ when $\r_A(t)$ is obtained by tracing out $A^c$ in constraint-preserving dynamics, i.e., $\rho_{A}(t)$ is block-diagonal. For the choices of $\dyn$ we consider the off-diagonal elements will however rapidly dephase.}---to evolve at long times in qualitatively the same way as the density matrix $\r$ in the model where the system undergoes depolarizing noise at its boundary. 
	
	One natural question concerns the amount of spatial resources required for subsystems to thermalize, i.e. how large $L_{\rm tot} \equiv |A| + |A^c|$ needs to be before the dynamics of $\r_A$ is qualitatively similar to the dynamics in the maximally depolarizing model (this comparison being made only for initial states contained within a single Krylov sector, or else confined to a small number of nearby sectors). 
    To quantify this, we define the {\it ergodicity length} $L_{\rm erg}(|A|)$ as the minimial size system in which the subsystem $A$ must be embedded in order that the dynamics on $A$ is ergodic, meaning that at long times, the reduced density matrix $\r_A$ has support on all Krylov sectors associated to a system of size $|A|$. 
    
    In generic chaotic dynamics, $L_{\rm erg}/|A|$ is usually a (perhaps large) $O(1)$ number, which is independent of $|A|$ in the $|A|\ra\infty$ limit. 
    In exponentially fragmented models, the story is rather different: in Sec.~\ref{sec:ex} we will see that the scaling of $L_{\rm erg}(|A|)$ with $|A|$ can vary quite dramatically, ranging from $L_{\rm erg} = \ct(|A|)$ to $L_{\rm erg} = \infty$.
    For models with infinite ergodicity length, subsystem entanglement entropies can {\it never} reach the Page value after undergoing a quench from a state in a definite Krylov sector, {\it even when they are embedded in an infinite system}. In these models, even the most generic possible constraint-preserving dynamics is unable to fully entangle subsystems with their complements. See Ref.~\cite{rakovszky2020statistical} and the ``fragile fragmentation'' phenomenon identified in Ref.~\cite{balasubramanian2023glassy} for further discussion.

	\section{The exponential fragmentation conjecture} \label{sec:graphs}
	
	\begin{figure}
		\includegraphics[width=.45\tw]{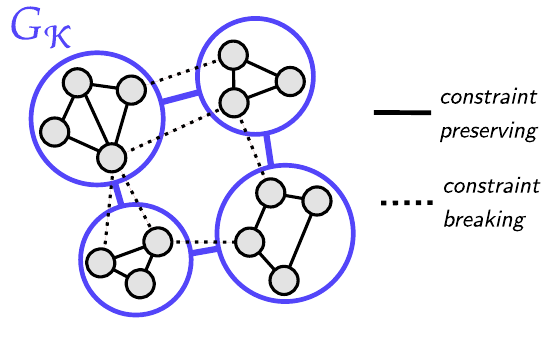} 
		\caption{\label{fig:GK} Coarse-graining and the Krylov graph. Each vertex represents a basis state in $\mch$, with a solid line drawn between two vertices if the corresponding states are connected by the constraint-preserving part of the dynamics, and a dashed line drawn if they are connected by the constraint-breaking part; the resulting graph is $\mcg_\mch$. Each Krylov sector (purple circles) defines a vertex of the Krylov graph $\mcg_\mch$, with an edge drawn between two sectors if states in each sector are connected to one another under the constraint-breaking part. }
	\end{figure}
	
	\ss{Krylov graphs and expansion}\label{sec: graph expansion}
	Most of the results in this paper are derived from understanding how the constraint-breaking part of $\dyn$ connects different states in Hilbert space. To this end, we will define a graph $\mathcal{G}_\mch$ by associating to each basis state $\k{\psi}$ of $\mch$ a vertex $v_\psi$ of $\mathcal{G}_\mch$, and drawing an edge $(v_\psi,v_{\psi'})$ between $v_\psi, v_{\psi'}$ if $\k\psi$, $\k{\psi'}$ are connected under a single step\footnote{In the case where $\dyn$ is generated by continuous time Hamiltonian evolution, an edge is present if $\lan \psi' | H | \psi\ran \neq 0$ (and likewise for Lindbladian evolution).} of the dynamics, i.e. if $\langle \psi'|\dyn(\proj\psi)|\psi'\rangle\neq0$ (we will always choose a basis of $\mch$ where each basis state is a product state belonging to a single $\mck_\a$). For dynamics where the constraints are everywhere unbroken, $\mcg_\mch$ contains one disconnected piece for each $\mck_\a$. When the constraints are broken, vertices in different sectors become connected, and this additional connectivity determines how fast the resulting dynamics can thermalize. 
	
    From here on, we will restrict ourselves to a particular kind of dynamics corresponding to random unitary circuit dynamics described in the Sec.~\ref{sss:local_constr_breaking}. In this model, once we take the Haar average (which will always be done at each time step), we can map the dynamics onto $\dyn = \mathcal{C}_{\mathrm{boundary}} \circ \mathcal{C}_{\mathrm{bulk}}$ for $\mathcal{C}_{\mathrm{boundary}}$ maximally depolarizing on the first site and $\mathcal{C}_{\mathrm{bulk}}$ maximally depolarizing within each Krylov sector, as described above. For the maximally depolarizing noise model we focus on, letting $p_\psi(t) \equiv \lan \psi | \r(t) | \psi\ran$ denote a natural probability distribution over $\mathcal{G}_\mch$ (with $
    \sum_\psi p_\psi(t) = 1$), $\dyn$ induces a classical stochastic process describing the time evolution of this distribution. In particular, the vector of probabilities $p(t)$ evolves according to $p(t+t') = \mcm^{t'} p(t)$, where $\mcm$ is the transition matrix of a Markov chain with matrix elements 
	\be \mcm_{\psi,\psi'} = \lan \psi'| \dyn(\proj \psi) | \psi'\ran.\ee 
    We will denote the stationary distribution of this chain by $\mu(\cdot)$, with $\mu(R) = \sum_{\psi \in R} \mu(\psi)$ for a subspace $R\subset \mch$; for all the examples we will be interested in, $\mu$ will simply be the uniform distribution over $\mch$.  
	
	To link the connectivity of $\mathcal{G}_\mch$ to the thermalization time of $\dyn$ (or equivalently, to the mixing time of the Markov process $\mcm$), we define the {\it expansion} or {\it conductance} of a subset of states $R \subset \mathcal{G}_\mch$ as 
	\be \label{cprdef} \cp(R) \triangleq \frac{\sum_{\psi \in R, \psi' \in R^c } \mu(\psi) \mathcal{M}_{\psi, \psi'}}{\mu(R)}.\ee 
	If $\mu$ is the uniform distribution over all nodes, we can replace the numerator with $\sum_{\psi \in R, \psi' \in R^c } \mathcal{M}_{\psi, \psi'}$ and the denominator with $|R|$.  The graph expansion (or graph conductance) is defined as
    \begin{equation}
     \cp(\mathcal{G}_\mch) = \min_{R\subset \mathcal{G}_\mch \, : \, |R| \leq |\mch|/2} \cp(R).
    \end{equation}
	$\cp(\mcg_\mch)$ measures how ``well-connected'' the dynamics is, and directly places a lower bound bound on the thermalization time due to (one side of) Cheeger's inequality \cite{lyons2017probability, hsf_impurities}:\footnote{The side of Cheeger's inequality which lower bounds $t_{\rm th}$ is conventionally formulated in terms of the second-largest eigenvalue $\l_2$ of the Markov process as $\frac1{1-\l_2} \geq (2\Phi(\mcg_\mch))\inv$; the form written here follows from this and a standard bound between $\l_2$ and $t_{\rm th}$ \cite{lyons2017probability}.}
	\be 
    \label{cheeger} 
    % \frac1{2\cp(\mathcal{G}_\mch)} \leq  \tth \leq \frac2{ \cp(\mathcal{G}_\mch)^2}.
    t_{\rm th} \geq \frac{C}{\Phi(\mcg_\mch)}-1
    \ee 
    with the constant $C = \ln(2)/2$.
	When $\mathcal{G}_\mch$ is disconnected, $\cp(\mathcal{G}_\mch) = 0$, and the system never thermalizes. On the other extreme, when $\cp(\mathcal{G}_\mch)$ is $O(1)$, then $\mathcal{G}_\mch$ is an expander graph, and the system thermalizes rapidly. 
	One of the central messages of this work will be to link the severity of fragmentation to the scaling of $\Phi(\mcg_\mch)$, by showing that systems with strong fragmentation have small $\cp(\mathcal{G}_\mch)$ and long thermalization times. 
	
	Unfortunately the expansion of $\mcg_\mch$ is usually difficult to calculate exactly, and we will thus mostly be interested finding upper bounds for its scaling with $L$ (which by \eqref{cheeger} then give lower bounds on $\tth$). For this purpose we will find it convenient to define a ``coarse-grained'' version $\mcg_\mck$ of $\mcg_\mch$ called the {\it Krylov graph}, which ignores the intra-sector aspects of $\dyn$ and focuses only on how $\dyn$ connects the different $\mck_\a$.  This is natural as under our model of dynamics, we assume that intra-sector dynamics in each $\mck_\a$ thermalizes instantaneously to the maximally mixed distribution over $\mck_\a$. To define the Krylov graph, we group all nodes in $\mcg_\mch$ belonging to a given sector into a single ``supervertex'', and all edges connecting two sectors into a single ``superedge''. We can define an ``effective'' expansion or conductance corresponding to this coarse-grained graph.  In particular, 
    upon coarse-graining to form the Krylov graph, the steady state distribution has a probability weight of $|\mck_\alpha|/|\mch|$ on the node labeled by $\alpha$, and the effective expansion of a region $R_\mck \subset \mcg_\mck$ is given by
    \begin{equation}
        \Phi(\mcg_{\mck}) = \min_{\substack{R_{\mck} \subset \mcg_{\mck} \\  \mu(R_{\mck})\leq 1/2}} \Phi(R_{\mck}),
    \end{equation}
    which by definition satisfies $\Phi(\mcg_{\mch}) \leq \Phi(\mcg_{\mck})$.

    We chose to work with the maximally depolarizing model as we expect the instantaneous intra-sector depolarization to thermalize the system in a strictly faster way than any local dynamics. To formalize this, we can instead consider local random unitary circuit dynamics, which can be mapped to a Markov process where the transition matrix has a local tensor product structure.  Granted that the stationary distribution is uniform over all nodes in  $\mathcal{G}_{\mathcal{H}}$, we thus expect the expansion of this Markov chain to be upper bounded by $\cp(\mathcal{G}_\mck)$. We prove this in App.~\ref{app:local}, and thus in what follows we will thus mostly focus on upper bounding $\cp(\mathcal{G}_\mck)$. At this point, in order to make further progress, we will need to make additional assumptions on the structure of $\mathcal{G}_\mck$. 
	In particular, configuration graphs obtained from exponentially fragmented dynamics appear to always yield exponentially small $\cp(\mcg_\mck)$, as we now discuss.

	\ss{Main statement } \label{ss:conj}
	
	\begin{figure}
		\includegraphics[width=.5\tw]{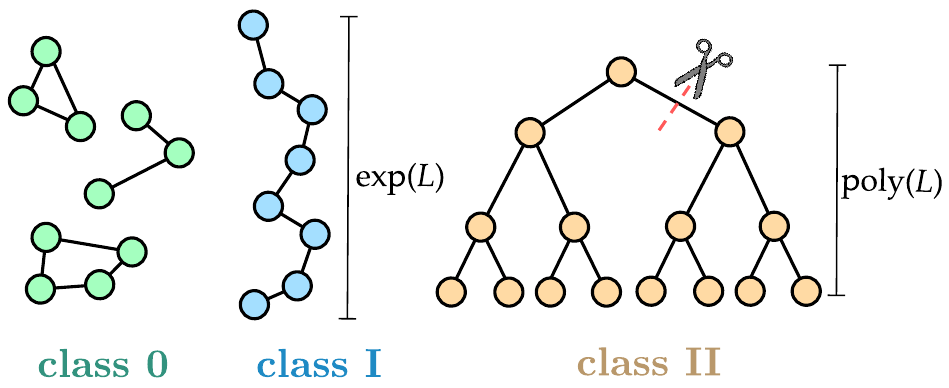}
		\caption{\label{fig:classes} Different types of Hilbert space connectivity that give rise to exponentially slow thermalization. In class 0, Hilbert space remains disconnected even in the presence of the bath, and $\tth = \infty$. In class I, the diameter of the Krylov graph $\mcg_\mck$ is $\O(\exp(L))$, meaning that it takes the bath at least $\sim\exp(L)$ time steps to move the system across Hilbert space. In class II, the bath can connect any two states in only $\poly(L)$ time steps, but Hilbert space possesses strong bottlenecks that render the thermalization dynamics exponentially slow. In this class, $\mcg_\mck$ typically has a tree-like structure as in the figure, where $\mcg_\mck$ can become disconnected into two thermodynamically large pieces after only a small number of edges are cut.   }
	\end{figure}
	
	Much of the remainder of this work will be devoted to addressing the following conjecture: if $\dyn$ is exponentially fragmented, then $\cp(\mcg_\mck) = O(\exp(-L))$. 
	This can be restated as: 
	
	\vspace{0.25cm}
	\fbox{\begin{minipage}{22.5em}
			\textbf{Conjecture: } If $\dyn$ is exponentially fragmented, the thermalization time is exponentially long in system size: $\tth = \O(\exp(L))$. 
	\end{minipage}}
	\vspace{0.25cm}
	
	This conjecture suggest that one can bound $\tth$, a (usually) hard-to-compute quantity determined by the global structure of $\mcg_\mck$, based on the size $|\mck_{\max}|$ of a single Krylov sector.
    The motivation for proposing this relation comes from the fact that for a large subclass of constraints---those obtainable within the group dynamics framework of Ref.~\cite{balasubramanian2023glassy}---it can be reformulated as a statement about the expansion of heat kernels of random walks on Cayley graphs; in this setting, our conjecture can be mapped to a refinement of a different conjecture by Benjamini related to the existence of certain kinds of ``robust'' expander graphs. This perspective is discussed further in Sec.~\ref{sec:groups}, where we prove our conjecture for dynamics whose constraints derive from the action of any hyperbolic group.  Though this may seem to be a restricted class of systems, a randomly generated group-based constraint corresponds to hyperbolic group dynamics with high probability; thus, our results are quite generic.  While a general proof for {\it all} exponentially fragmented dynamics (including those that do not exhibit a group structure) will seem to require new mathematical ideas and techniques (see Sec.~\ref{ss:math}), it holds for all examples known to the authors, some of which will be studied in detail in Sec.~\ref{sec:ex}.  
	
	To develop intuition for this conjecture, it is helpful to group models of exponentially fragmented dynamics into three classes, corresponding to the three qualitatively distinct ways in which a graph can have small expansion. We refer the reader to Fig.~\ref{fig:classes} for an illustration: 
	\begin{itemize}
		\item Class 0 ({\it persistent ergodicity breaking}): $\mcg_\mck$ is disconnected. Here $\dyn$ remains non-ergodic even after the constraint-breaking terms are added, and the system fails to thermalize even at infinite time, $\tth = \infty$. 
		\item Class I ({\it large diameters}): $\mcg_\mck$ is connected, but has an exponentially large diameter, ${\sf diam}(\mcg_\mck) = \O(\exp(L))$. $\tth = \O(\exp(L))$ simply because it takes exponentially long time to traverse Hilbert space. 
		\item Class II ({\it bottlenecks}): ${\sf diam}(\mcg_\mck) = O(\poly(L))$, but $\mcg_\mck$ has exponentially small expansion, meaning that it possesses severe bottlenecks that produce an exponentially long thermalization time (for example, some models in this class have ``tree-like'' Krylov graphs). Sometimes these bottlenecks manifest as localized motifs in real space which restrict the dynamics, while other times they are associated with non-local degrees of freedom.
	\end{itemize}
	The conjecture thus claims that if $\dyn$ is exponentially fragmented, then $\mcg_\mck$ is either disconnected, or becomes disconnected into several pieces, each of which contains $\O(\exp(L))$ states, after only a small $(\sim o(\poly(L))$) number of edges are cut. 
	
	A well-known example of a model in class 0 is the PXP model \cite{turner2018quantum} (a less well-known example is provided by a colored version of the Fredkin chain \cite{salberger2018fredkin}). Examples of systems in class I that the authors are aware of have exponentially modulated symmetries \cite{sala2022dynamics}, and are variants on the quantum breakdown models of Refs.~\cite{lian2023quantum,liu20232d,hu2024bosonic,chen2024quantum} (an example of which is treated in Sec.~\ref{ss:breakdown}).
    Ref.~\cite{hsf_impurities} showed that the pair-flip model belongs to class II; other well-known models belonging to this class include the $tJ_z$ model \cite{rakovszky2020statistical} (Sec.~\ref{ss:tjz}) and the exponentially fragmented $S^z=1$ dipole-conserving model \cite{khemani2020localization,sala2020ergodicity} (Sec.~\ref{ss:dipole}). Only the last of these examples has bottlenecks visible as localized motifs in real-space. 

    Our conjecture applies equally well to systems with classical and quantum HSF. However, the exponentially- and quantum-fragmented models the authors are aware of all have the property that they reduce to classically- and exponentially-fragmented models upon adding certain operators to the dynamics. Adding additional operators to the dynamics in this way should not parametrically increase thermalization times, and hence in what follows we will simplify the discussion by focusing on classically fragmented examples. 
	
	\section{Examples} \label{sec:ex}
	
	In this section we verify the conjecture for several explicit examples, which illustrate the range of mechanisms by which slow thermalization can occur. We will focus on examples in classes I and II only, as those in class 0 have $\tth=\infty$ for trivial reasons.

	\ss{Class I: The spin-1 breakdown model} \label{ss:breakdown}
	
	\sss{Boundary depolarizing noise}
	
	An illustrative model in class I is what we will refer to as the {\it spin-1 breakdown model}, following the fermionic and bosonic breakdown models studied in Refs.~\cite{lian2023quantum,liu20232d,hu2024bosonic,chen2024quantum}. The constraint in breakdown models comes from an {\it exponentially modulated} $U(1)$ symmetry \cite{sala2022dynamics}, which in the spin-1 context is defined as 
	\be Q = \sum_{i=1}^L 2^{i-1} n_i,\label{eq:charge_breakdown} \ee 
	where $n_i \equiv S^z_i + 1$ is to be thought of as counting the number of ``particles'' on site $i$---we will accordingly write the basis for the onsite Hilbert space as $\{\k0,\k1,\k2\}$. If we were to specialize to Hamiltonian dynamics, $\dyn$ could be generated by a Hamiltonian of the form 
	\be H = J \sum_i ((S^+_i)^2 S^-_{i+1} + h.c.) + H_z,\ee 
	where $H_z$ contains only $S^z$ operators. As usual however, we will find $Q$-preserving RU dynamics to be more convenient. 
	
	The maximum value of $Q$ is $Q_{\max} = 2\sum_{i=1}^L 2^{i-1} = 2^{L+1}-2$, immediately implying an exponentially large number of sectors. It is easy to convince oneself that $\dyn$ is fully ergodic within each charge sector,\footnote{The easiest way to see this is by noting the close resemblance of the definition of conserved charge in Eq.~(\ref{eq:charge_breakdown}) to the binary representation of real numbers in reverse order, with the least significant bit on the left. Hence different configurations in $\mathcal{K}_Q$ correspond to all possible ways of representing $Q$, with the carries either transferred to the more significant bits or not.} giving exactly $N_\mck = 2^{L+1}-2$ Krylov sectors in total; we will accordingly write a given Krylov sector with charge $Q$ as $\mck_Q$. 
	We remark that although this model is not fragmented according to a previously adopted definition of HSF in literature---since the dynamics is fully ergodic within each charge sector---it {\it is} exponentially fragmented according to our definition in Sec.~\ref{sec:setup} (see below), which is the more relevant notion when considering the dynamics induced by boundary depolarizing noise. 
	
	%---------------------------------------------------------------------------------
	\begin{figure}[!t]
		\centering
		\includegraphics[width=0.5\textwidth]{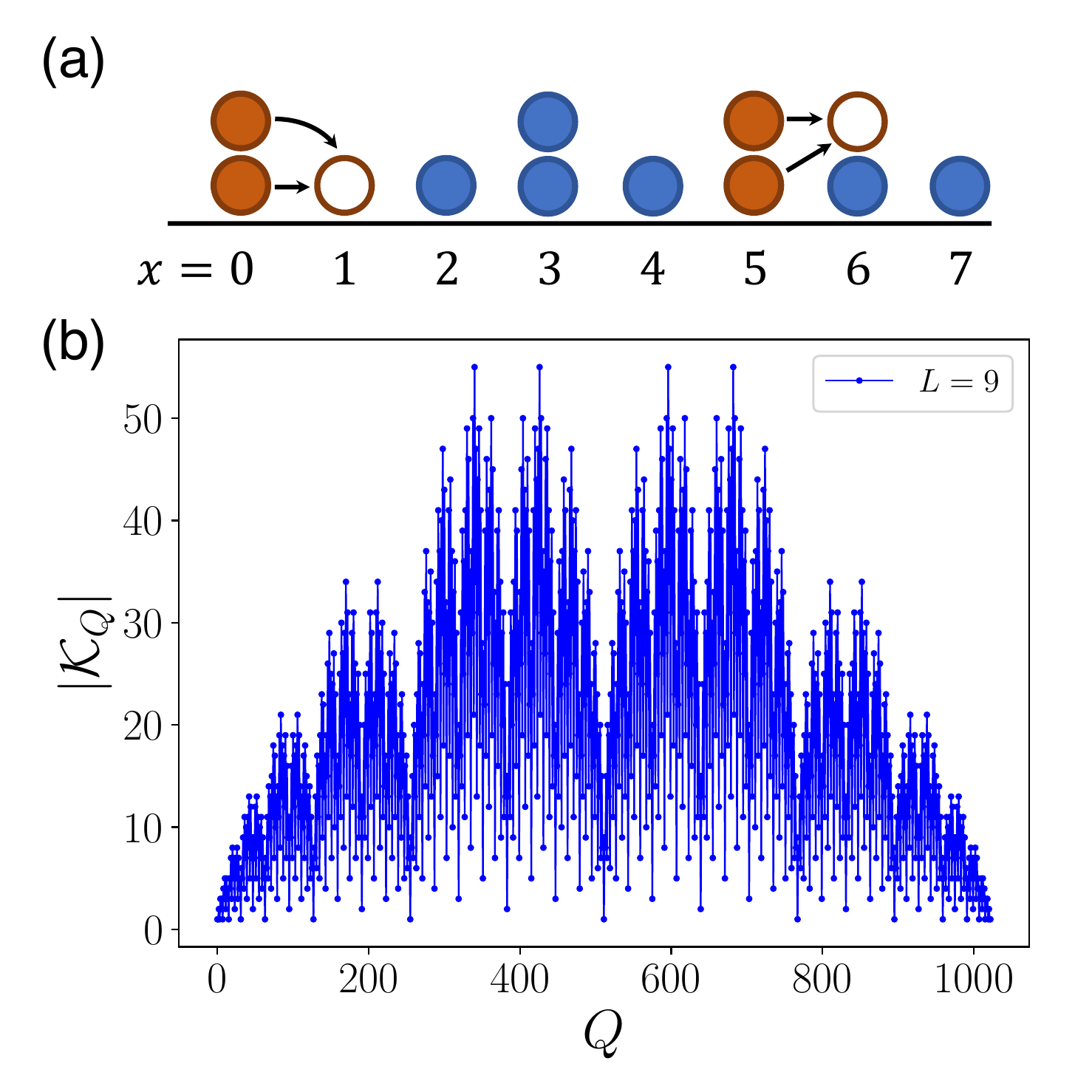}
		\caption{(a)  Allowed dynamical moves of the spin-1 breakdown model. (b) Sizes of each Krylov sector $\mathcal{K}_Q$ with charge $Q$. Notice that the sector sizes are symmetric about $Q_{\rm max}/2$ (particle-hole symmetric). The sizes also exhibit a self-similar structure, which can be understood from an underlying recursion relation of $|\mathcal{K}_Q|$ (see App.~\ref{app:breakdown})}
		\label{fig:breakdown_Krylov_dimension}
	\end{figure}
	%---------------------------------------------------------------------------------
	In Fig.~\ref{fig:breakdown_Krylov_dimension}, we plot the sizes of each Krylov sector $\mathcal{K}_Q$ with charge $Q$ ranging from 0 to $Q_{\rm max}$. Notice that the Krylov sector sizes exhibit a remarkable self-similar structure, and is symmetric about the middle point $Q_{\rm max}/2$.
	For now, we simply list the following useful facts regarding the structure of $\mck_Q$ in this model, while leaving a detailed justification in App.~\ref{app:breakdown}:
	\begin{enumerate}
		\item The sector sizes are particle-hole symmetric:\footnote{The particle-hole symmetry amounts to taking $n_i\rightarrow 2-n_i$ on each site, while the total charge transforms as $Q \rightarrow Q_{\rm max}-Q$. One can check that the two types of allowed moves in this model are exactly related by a particle-hole transformation.} $|\mck_Q| = |\mck_{Q_{\max}-Q}|$;
		\item The sector of charge $Q_{\max}/2$ has dimension $|\mck_{Q_{\max}/2}| = 1$, being spanned by the state $\k{1}^{\tp L}$;
		\item The sector sizes $|\mathcal{K}_Q|$ satisfy a simple recurrence relation, which allows one to obtain the sizes of sectors with larger $Q$ (and larger system sizes) from those of smaller $Q$ (and system sizes);
		\item The sectors with largest dimensions are those with charge $Q_{\max}/3, 5Q_{\max}/12$, and the particle-hole conjugates thereof. These sectors have dimension that grows as
		\be |\mck_{\rm max}| \propto \phi^L \ee \label{golden}
		in the thermodynamic limit, where $\phi = (1+\sqrt5)/2$ is the golden ratio.
	\end{enumerate}
	Fact 4 shows that the spin-1 breakdown model is exponentially fragmented, with the Krylov graph $\mcg_\mck$ being a line segment of length $N_\mck = Q_{\max}+1$.\footnote{The bath located at the left boundary can change the charge by $\pm1$ when acting on a state with odd charge, and by $\pm1$ or $\pm2$ when acting on a state with even charge. Thus $\mcg_\mck$ is more accurately described as a line segment where each odd-numbered vertex is connected to its nearest neighbors, and each even-numbered vertex is connected to both its nearest and next-nearest neighbors. } Thus
	\be {\sf diam}(\mcg_\mck) = \ct(2^L),\ee 
	immediately implying an exponentially long thermalization time.

	%-----------------------------------------------------------------------------------------------------------------------------------------------
	\begin{figure}[!t]
		\centering
		\includegraphics[width=0.5\textwidth]{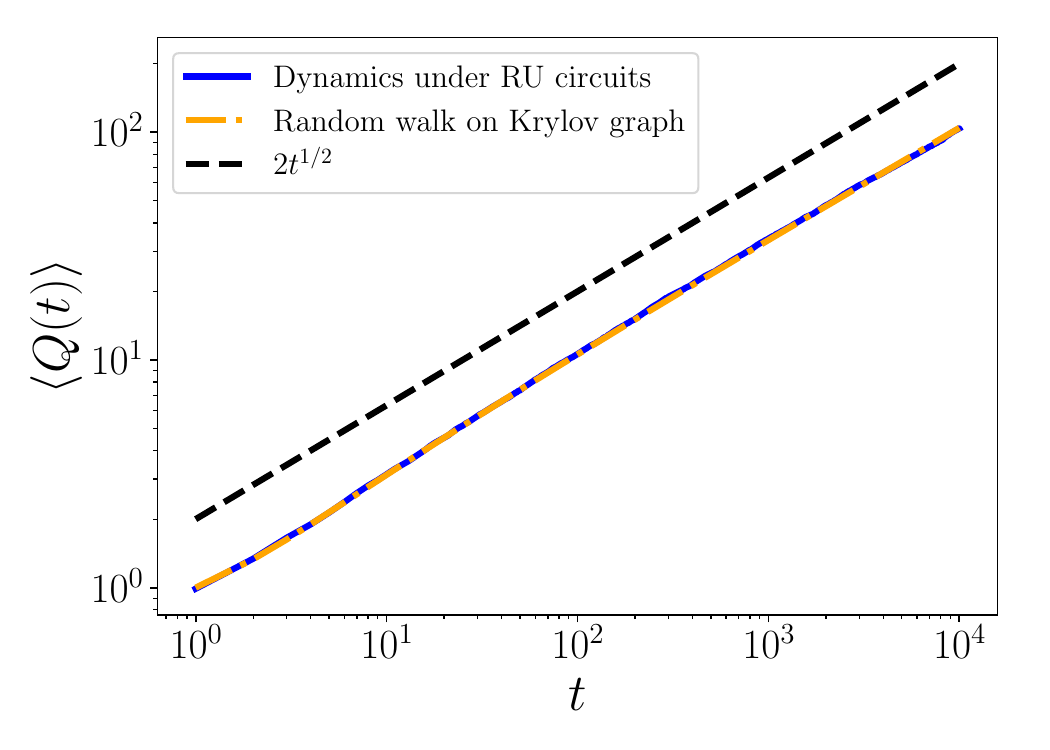}
		\caption{Numerical results for the total charge relaxation dynamics in the spin-1 breakdown model with boundary depolarizing noise. The blue line shows results obtained from a direct simulation of the stochastic dynamics of the spin-1 chain under deep RU circuits in the bulk + depolarizing channel at the endpoint. The orange line is obtained by simulating the effective random walk process on the corresponding Krylov graph. Both results are consistent with a diffusive charge relaxation.}
		\label{fig:Breakdown_Dynamics_numerics}
	\end{figure}
	%-----------------------------------------------------------------------------------------------------------------------------------------------
	
	For example, consider what happens when one starts from the $Q=0$ state $\k{0}^{\tp L}$, and tracks the charge expectation value $\lan Q(t)\ran$ as a function of time $t$. In this setting we may define the charge thermalization time $t_{\rm th}^Q$ as the first time for which $\lan Q(t)\ran$ approaches within $\epsilon$ of its value in the $T=\infty$ late time steady state: 
	\be \tth^Q = \min\{ t \, : \, |\lan Q(t)\ran - \lan Q(\infty)\ran | < \ep \lan Q(\infty)\ran\},\ee  
	where fact 1 above implies $\lan Q(\infty)\ran = Q_{\max}/2 = 2^L-1$. 
 %The scaling of $\tth^Q$ with $L$ can be estimated by looking at the random walk on $\mcg_\mck$ induced by the dynamics.
 %\ethan{the following argument is only a proof in the limit where the bulk dynamics is "rapidly mixing", viz. in the limit where the system relaxes to a uniform sum over the different states in a given sector in between applications of the bath. However we really only require that the sum look uniform near the boundary, so I think the argument is morally still correct. } It is easy to show that this random walk has symmetric transition rates away from the boundary of $\mcg_\mck$.\footnote{When the bath acts on states in sectors with odd charge---all of which have $n_1=1$---the charge is clearly equally likely to increase by 1 as it is to decrease by 1, and hence the transition rates on the odd sectors of $\mcg_\mck$ are symmetric. For sectors of even charge, the transition rates will be symmetric if states with $n_1=0$ are equally common as states with $n_1=2$. That this is indeed the case follows by noting that any state with $n_1=0$ can be mapped uniquely onto a state with $n_1=2$ (and vice versa) by moving charges near the boundary to the left.\ethan{cheng will need to correct this\dots} } 
 %The motion of $\mcg_\mck$ induced by the dynamics is therefore diffusive, yielding 
 %	\be \tth^Q \sim (Q_{\max}/2)^2 \sim 4^L.\ee 	
	%We corroborate the above argument with numerical simulations, as shown in Fig.~\ref{fig:Breakdown_Dynamics_numerics}. 
        We numerically simulate the charge relaxation dynamics, as shown in Fig.~\ref{fig:Breakdown_Dynamics_numerics}. We perform two different simulations: (1) a direct simulation of the stochastic dynamics of the spin-1 chain under deep RU circuits in the bulk and depolarizing channel at the boundary; (2) the effective random walk process on the corresponding Krylov graph (a chain in this case). In (1), we evolve the system for sufficiently long in between two consecutive actions of the depolarizing channel at each time step, so that the intra-sector dynamics is fully mixing instantaneously. For (2), the transition probability from vertex/sector $\mathcal{K}_{q}$ to $\mathcal{K}_{q'}$ is chosen according to 
	\begin{equation}
		p(\mathcal{K}_q \rightarrow \mathcal{K}_{q'}) = \frac{1}{|\mathcal{K}_q|} \sum_{\psi \in \mathcal{K}_q, \psi' \in \mathcal{K}_{q'}} \mathcal{M}_{\psi, \psi'},
	\end{equation}
	where $q'=q\pm 1$ for $q$ odd and $q'=q\pm 1$ or $\pm 2$ for $q$ even. We find that, although for a particular sector on the Krylov graph, the transition rates to its left and right are not symmetric, the total charge relaxes in a diffusive manner, suggesting that on average the dynamics correspond to a random walk on the Krylov graph with no bias. Nonetheless, a rigorous proof of this statement so far has remained elusive.
        %We find that both results are in good agreement with a diffusive charge relaxation.
	
	\sss{Subsystem dynamics: infinite ergodicity length}
        \label{sec:breakdown_ergodicity}
	
	While adding depolarizing noise at a single site renders $\dyn$ ergodic, this model has $L_{\rm erg} = \infty$, meaning that even in an {\it infinite} system,  the reduced density matrices $\r_A$ of any contiguous subsystem with size $|A|>1$ will never have full rank. This in turn implies that the constraint-preserving dynamics is very poor at generating entanglement, and is unable to maximally entangle any finite region with its complement during evolution from a product state, even if given infinite temporal and spatial resources. From the perspective of thermalization in isolated quantum systems, this means that under unitary dynamics, the system cannot act as its own bath and bring its subsystems to thermal equilibrium, however large the size of the reservoir compared to the subsystem.
	
	Consider a subsystem $A$ of a larger system. The total conserved charge can be split into 
	\begin{equation}
		Q = Q_{A^c_L} + 2^{i_{l,A}-1} Q_A + Q_{A^c_R},
	\end{equation}
	where
	\be Q_A = \sum_{i \in A}2^{i-i_{l,A}} n_i,\ee 
	$i_{l,A}$ is the site at the leftmost end of $A$, and $Q_{A^c_{L/R}}$ denotes charge in the region to the left/right of $A$.
	%with this definition $Q_A$ can take on any integer between $0$ and $Q_{A,\max} = 2^{|A|+1}-2$. 
	Suppose at $t=0$ the system is initialized in a product state with a particular initial value $Q_A(0)$ of $Q_A$. Since the total charge of the full system $A\cup A^c$ is conserved, change of $Q_A$ must come from charge transferred in and out of $A^c$. Let us first consider region $A^c_L$ to the left of $A$. A particle entering region $A$ from its left end will increase $Q_A$ by $2^{i_{l,A}-1}$. However, since the maximal amount of charge in region $A^c_L$ is $2^{i_{l,A}}-2$, it can only pump or absorb at most {\it one} particle into or from region $A$, no matter how big its size is. A similar reasoning holds for $A^c_{R}$ to the right of $A$. We find that, quite remarkably, the imposition of a global symmetry---albeit a rather unconventional one---renders the system being an extremely poor particle reservoir for its subsystems.
	
	Thus, after evolving the full system $A\cup A^c$ under constraint-preserving dynamics, the value $Q_A(t)$ of $Q_A$ at time $t$ must be expressible as 
	\be Q_A(t) = Q_A(0) + a + b2^{|A|},\ee 
	where $a,b \in \{-1,0,1\}$ express the distinct ways that particles can be transferred between $A$ and $A^c$. For large $|A|$, we may use $|\mck_{\max}| = \phi^L + \cdots$ to conclude that for all $t$, 
	\be {\rm rank}(\r_A(t)) \leq 9 \phi^{|A|},\ee 
	where the factor of $9$ comes from the number of ways of choosing $a,b$. The entanglement entropy of $\r_A$ is accordingly upper bounded as 
	\be \label{breakdown_S_bound} S(\r_A(t)) \leq |A| \ln(\phi) + {\rm const},\ee 
	which since $\ln(\phi) < \ln(3)$ means that the coefficient of the volume law can never be made to match the scaling of a random state, so that $\dyn$ is unable to fully entangle subsystems with their complements, even when given infinite spatial and temporal resources. 
	
	\ss{Class II: $tJ_z$ Model} \label{ss:tjz}
	\sss{Boundary depolarizing noise}
	Perhaps the simplest model in class II is the $tJ_z$ model, which was originally formulated as a hard-core Fermi-Hubbard chain with only $S^z_i S^z_{i+1}$ spin interactions~\cite{zhang1997tjz,Batista2000tjz}.\footnote{The thermalization dynamics of this model is somewhat similar to that of the colored Motzkin chain \cite{movassagh}, another well-known exponentially fragmented model; since the analysis of $tJ_z$ is simpler, we will not discuss the Motzkin chain explicitly. } Writing the onsite Hilbert space as $\k{0}, \k{\upa}, \k{\doa}$, the dynamics is such that the only allowed local matrix elements are of the form $\kb{0\s}{\s0}_{i,i+1}, \kb{\s0}{0\s}_{i,i+1}$ for $\s \in \{\upa,\doa\}$. For convenience, we will consider the corresponding RU dynamics from now on. Under such dynamics, the Hilbert space is fragmented into exponentially many Krylov sectors, with each sector characterized by a spin pattern of $\uparrow$ and $\downarrow$'s. For example, the spin pattern of the product state $|\uparrow 0\downarrow\uparrow0\uparrow\rangle$ is $\uparrow\downarrow\uparrow\uparrow$, which labels its Krylov sector that is dynamically disconnected from other sectors of different spin patterns. 
 
	\begin{figure}
		\centering
		\includegraphics[width=0.5\tw]{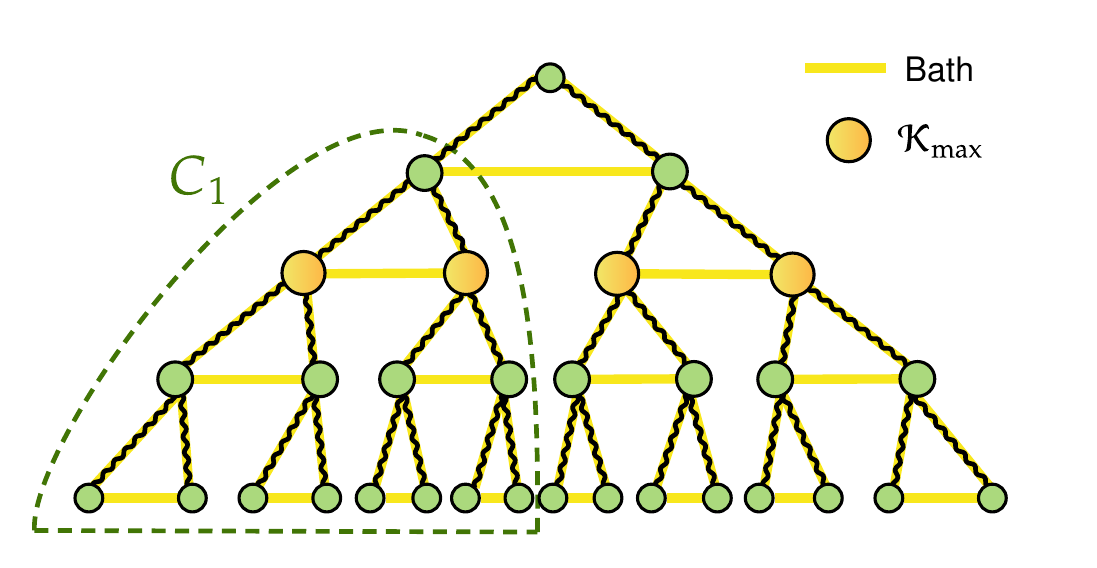}
		\caption{Krylov graph of the $tJ_z$ model for $L=4$ coupled to a heat bath at one end of the chain. Each product state defines a walk on a binary tree (highlighted in wavy lines), where the walk stops for an empty site and takes the direction determined by the spin for an occupied site. The endpoint of the walk is invariant under pure $tJ_z$ dynamics and labels the corresponding Krylov sector. A given sector of particle number $Q$ lies at depth $d=Q$ defined as the distance from the top of the tree. The largest Krylov sectors (yellow circles) lie at  $d=L/2$, while a typical sector lies at $d=2L/3$. Breaking the constraints at one end of the chain changes the last step of the walk and connects Krylov sectors in the way indicated by the yellow line. A random initial state in $C_1$ (the region demarcated by the green dashed line) takes an exponentially long time to escape due to the  bottleneck at the tree apex. }
		\label{fig:tJz}
	\end{figure}

	The corresponding Krylov graph $\mcg_\mck$ is simply a binary tree, as illustrated in Fig.~\ref{fig:tJz}. A product state $|a_1\dots a_L\rangle$ with $a_i\in\{0,\uparrow,\downarrow\}$ is associated with a length-$L$ lazy walk on $\mcg_\mck$ which starts from the top of the tree (the top being associated with the single state $\k{0}^L$). The sequence of the walk is read from left to right, with the $i$-th step of the walk taking the direction determined by $a_i$: the walk moves left if $a_i=\upa$, moves right if $a_i=\doa$, and remains where it is if $a_i=0$. Since the spin pattern is invariant under the dynamics, the endpoint of the walk is conserved, with the endpoint labelling the Krylov sector of the state. There are $2^{L+1}-1$ Krylov sectors in total. We define the depth $d$ of a sector as the distance from the top vertex of the tree. A given sector of particle number $Q = \sum_i (\proj{\upa}_i + \proj{\doa}_i)$ is located at depth $d=Q$, with dimension $|\mck_Q|={L\choose Q}$. The largest Krylov sectors are located at depth $d=L/2$ and have dimension $|\mck_{\max}|={L\choose L/2}=O(\frac{2^L}{\sqrt{L}})$, indicating that the $tJ_z$ model is exponentially fragmented. 
	
	Now we consider the RU dynamics $\dyn$ coupled to the depolarizing noise at site $L$. Breaking the constraints at the end of the chain restores ergodicity by mapping the last step of the walk $a_L$ to a random direction, and connecting the Krylov sectors as indicated in Fig.~\ref{fig:tJz}. The thermalization time is lower bounded by the inverse of the coarse-grained expansion $\Phi(\mcg_\mck)$ from Eq.~(\ref{cheeger}). $\Phi(\mcg_\mck)$, defined as the minimum expansion of any subset $R\in \mcg_\mck$, is found to be the expansion of $C_1$, a full branch of the binary tree cut from depth $d=1$ (shown in Fig.~\ref{fig:tJz}). While the number of edges connecting $C_1$ and $C_1^c$ is $O(1)$, the dimension of the branch is $|C_1|=(|\mch|-1)/2=(3^L-1)/2$, resulting in a severe bottleneck with exponentially small $\Phi(\mcg_\mck)$. Therefore, $\tth$ satisfies 
	\be \tth\geq\frac{1}{2\Phi(\mcg_\mch)}\geq\frac{1}{2\Phi(\mcg_\mck)}=\frac{3}{8}(3^L-1), \ee
	and is thus exponentially long. For more details, see App.\ref{app:tJz}. 
	
	Since the distribution of charges is anisotropic across $\mcg_\mck$, the expectation value of the magnetization 
	\be m\equiv\frac{1}{L}\sum_{i=1}^L(\kb{\uparrow}{\uparrow}_i-\kb{\downarrow}{\downarrow}_i) \ee
	serves as a reliable indicator of thermalization. We find that it takes an exponentially long time for a state initialized in a subset of states with nonzero expectation value of $m$ to leave that region, indicating an exponentially long thermalization time. More precisely, defining the magnetization relaxation time $t_m(\gamma)$ as the time needed for $\lan m(t)\ran$ to drop below $\gamma$, as computed for the initial state with maximum charge $|\psi_{\max}\ran = |\uparrow\ran^{\otimes L}$. Following an analysis similar to that of \cite{hsf_impurities}, we find 
	\be \label{eq:tJz_t_m} t_m(\gamma;\psi_{\max})\gtrapprox 3^{(1-2\gamma\log_3[e/\gamma])L}\times\frac{3e}{4\gamma(1+2\gamma)} \ee
	under the limit $L\to\infty$ and $\gamma\ll 1$ (See App.\ref{app:tJz}). This lower bound is numerically verified in the left panel of Fig.\ref{fig:tJz_t_Q}.
	
	\begin{figure}[!t]
		\centering
		\includegraphics[width=0.5\textwidth]{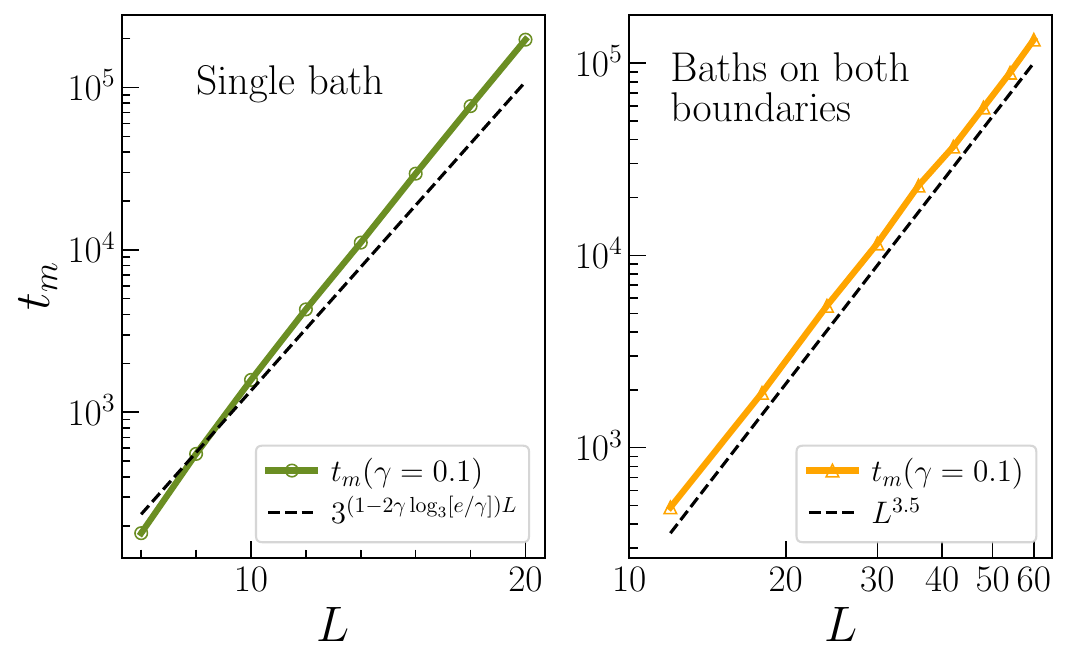}
		\caption{Magnetization relaxation time $t_m(\gamma,\psi_{\max})$ for the maximal-charge initial state $|\psi_{\max}\ran=|\uparrow\ran^{\otimes L}$. We take $\gamma=0.1$. {\it Left:} $t_m$ for the $tJ_z$ model coupled to a thermal bath at {\it one} boundary. The result shows good agreement with the lower bounded in Eq.~(\ref{eq:tJz_t_m}). {\it Right:} $t_m$ for the $tJ_z$ model coupled to thermal baths at {\it both} boundaries. Numerically $t_m\sim L^{3.5}$ and charge relaxation is exponentially faster than the case when only a single boundary is coupled to the bath.}
		\label{fig:tJz_t_Q}
	\end{figure}

    We remark that in this model (as well as the pair-flip model studied in Ref.~\cite{hsf_impurities}), thermalization becomes exponentially faster either when the boundary condition is changed from open to periodic, or when coupling both boundaries to a thermal bath. This is numerically verified for the $tJ_z$ model in Fig.~\ref{fig:tJz_t_Q}, from which we find
	%We note that the thermalization time significantly decreases to a polynomial scaling when the boundary condition is changed from open to periodic, or by adding additional depolarizing noise to the other end of the chain. The numerical simulation in the right panel of Fig.\ref{fig:tJz_t_Q} gives 
	\be t_m(\gamma=0.1)|_{\text{baths at both ends}}\sim L^{3.5}.\ee  Heuristically, with depolarizing noise applied to both ends, one can think of the original system as being embedded in a one-dimensional dictionary containing all possible configurations, and the stochastic process can be visualized as the original system sliding over the dictionary like a window, with the region enclosed in the window denoting the current configuration to which the initial state is mapped. This facilitates rapid thermalization compared to the case with only one end coupled to the noise. Another way of understanding this is that the additional thermal bath at site 1 generates nonlocal moves on the Krylov graph, mapping states within one branch of the tree $C_1$ to states in the other branch, resulting in a highly connected Krylov graph with an expansion that is no longer exponentially small.
	
	\sss{Ergodicity length}
	To find the ergodicity length of the $tJ_z$ model, let us consider a subsystem of size $|A|$ embedded in a larger system. Apparently, for the $tJ_z$ model, subsystem $A$ must be coupled to its complement at both ends for the subsystem dynamics to be ergodic: $A^c \cup A \cup A^c$. Since the total spin pattern cannot change for the full system, the only way for subsystem $A$ to explore all its possible spin patterns is to have all of them stored in $A^c$ and transported to $A$ under the constrained dynamics. We thus consider the following embedding: $\k{0}^{\tp L_{\rm erg}} \tp |\psi\rangle_A \tp \k{{\sf dict}}$, where $\k{{\sf dict}}$ is the de Bruijn sequence which is the shortest string that contains all possible orderings of $|A|$ spins \cite{de1946combinatorial}. Under open boundary condition, the length of $\k{{\sf dict}}$ is
    \be L_{\rm erg}(|A|)=2^{|A|}+|A|-1. \ee
	Thus, the ergodicity length of the $tJ_z$ model scales exponentially with $|A|$.

	\ss{Range-three dipole-conserving model (class II): real-space bottlenecks} 
        \label{ss:dipole}

    \sss{Boundary depolarizing noise}

	\begin{figure}
		\centering
		\includegraphics[width=0.5\textwidth]{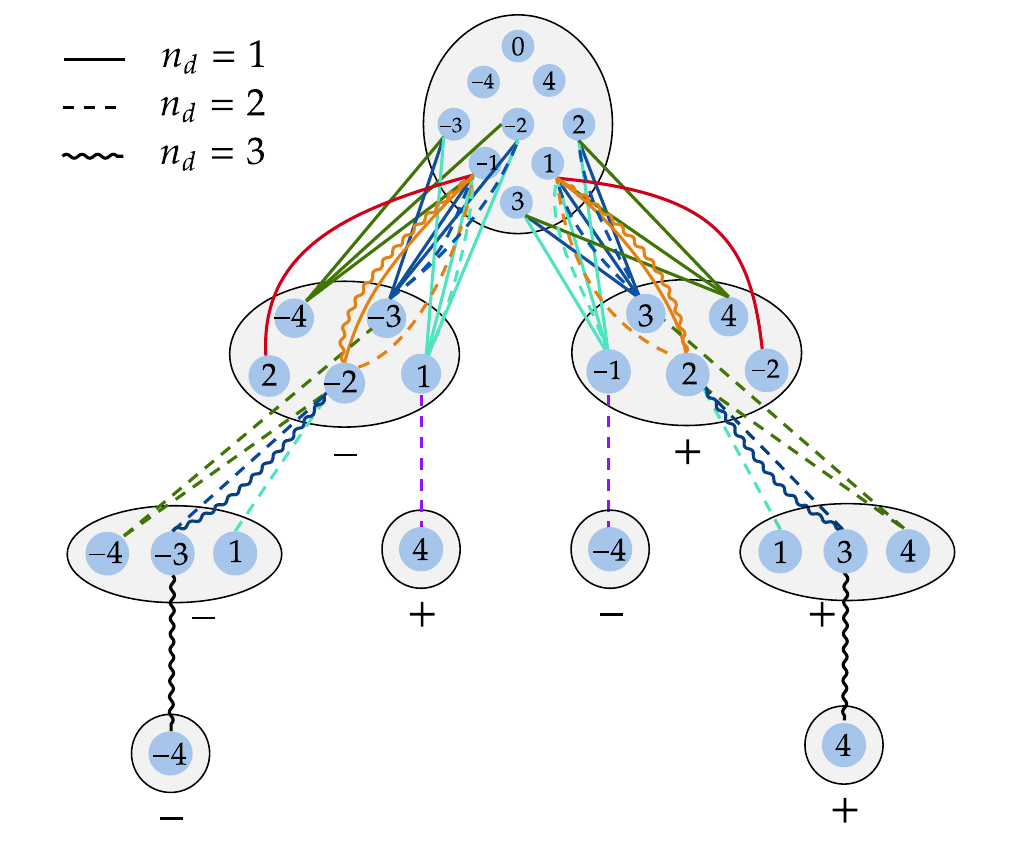}
		\caption{The Krylov graph of range-three dipole-conserving model for $L=4$. There are two layers of fragmentation in this model. The first layer is the conserved defect pattern represented by the nodes on the binary tree. The spin of the $k$-th defect is labeled by $+/-$ under each sector at depth $k$. Within each sector at depth $k$, the Hilbert space is further fragmented into subsectors (blue circles) labeled by the dipole moment $P_k$ between the $k$-th and $(k+1)$-th defects. The connected sectors give the allowed combinations of $\{P_k\}$ under the constraint, with different patterns of lines connecting them denoting different numbers of defects $n_d$ in the state.
		}
		\label{fig:dipole_krylov}
	\end{figure}
	
	We move on to a prototypical example of HSF, namely the dipole-conserving ``fracton'' model first studied in Refs. \cite{pai2019localization,khemani2020localization,sala2020ergodicity}. The model consists of a spin-1 degree of freedom on each site, which in the $S^z$ eigenbasis can be interpreted as positive/negative charge ($S^z=\pm$) and vacuum ($S^z=0$). Local dynamics are subject to two global conserved quantities: the total charge $Q=\sum_i S^z_i$ and dipole moment $P=\sum_i i S_i^z$. Alternatively, one can think of a one-dimensional system of particles with occupation number on each site restricted to $n_i=0, 1$ or $2$, and particle hoppings are constrained by total particle number and center-of-mass conservation. The addition of dipole moment or center of mass conservation has a nontrivial effect on the dynamics of the system. For example, isolated charges or particles are immobile due to this additional conservation law. Dynamical moves that are compatible with the above two conserved quantities must therefore involve $r \geq 3$ consecutive sites. In this subsection, we restrict ourselves to $r=3$ (hence the name range-three dipole-conserving model). It turns out that the dynamics both with and without a bath are drastically distinct for $r=3$ and $r>3$, and a detailed discussion of the latter case will be deferred till Sec.~\ref{ss:dipole_four}.
	
	The dynamics of the range-three dipole-conserving model can be generated by the following Hamiltonian using the spin-1 representation:
	\begin{equation}
		H_3 = \sum_j S^+_{j-1} (S_j^-)^2 S^+_{j+1} + {\rm H.c.},
	\end{equation}
	but we consider more generally $\dyn$ generated by symmetry-preserving three-site random unitary gates. It has been shown that the Hilbert space under the above $\dyn$ is strongly fragmented, and also exponentially fragmented~\cite{khemani2020localization, sala2020ergodicity}. We will show that this model belongs to Class II and possesses severe bottlenecks towards thermalization. However, unlike the $p$-flip and $tJ_z$ model, here the bottleneck manifests itself as localized ``motifs" in real space, and this special feature leads to a higher level of robustness against baths coupled to \textit{both} ends of the system.
	
	We start from the structure of the Krylov graph $G_{\mathcal{K}}$. The Krylov sectors of this model share certain similarities with the $tJ_z$ model, in that the dynamics preserve a certain pattern of ``defects" which can be used to label each Krylov sector. Following Ref.~\cite{rakovszky2020statistical}, such defects are defined as spins that are identical to their nearest-neighbor on the left, ignoring the empty sites. The defect pattern plays the same role as the conserved spin pattern in the $tJ_z$ model. However, configurations sharing the same defect pattern further fractures into even smaller subsectors, as the dipole moment $P_k$ in between the $k$-th and $(k+1)$-th defects must also be conserved~\cite{rakovszky2020statistical}. Thus, at a coarse-grained level, $G_{\mathcal{K}}$ of the range-three fracton model is a binary tree just like the $tJ_z$ model when only the conserved defect pattern is resolved. Zooming in, however, each node of the binary tree now exhibits a second layer of fragmentation, according to the collection of conserved dipole moments $\{ P_k \}$ in between the defects. Fig.~\ref{fig:dipole_krylov} illustrates the Krylov graph $\mcg_\mck$ for $L=4$. The internal structure within each node on the binary tree further reduces the connectivity of the Hilbert space. When coupled to constraint-breaking perturbations at the boundary, not all states sharing the same defect pattern can be connected. To fully scramble the subsectors within the same sector of a defect pattern close to the leaves of the tree, one needs to walk deeply into the bulk of the tree and then backtrack. 
	% \zhicheng{Perhaps there's more that can be said about the Krylov graph. Maybe Yiqiu can add.}
	
	% While an analytical bound on $\Phi(G_{\mathcal{K}})$ and hence $t_{\rm th}$ has eluded us so far, the similarity of the Krylov graph to that of the $tJ_z$ model suggests that thermalization in this model is also exponentially slow. 
	While an analytical bound on $\Phi(G_{\mathcal{K}})$ and hence $t_{\rm th}$ has eluded us so far, the above analysis suggests that thermalization in this model is slower than that of the $tJ_z$ model, and thus also exponentially slow. In fact, the slow dynamics in fracton systems can be directly understood from a real-space perspective. First of all, notice that since the dipole moment $P_k$ between adjacent defects must be conserved, the mobility of the individual defects is severely constrained.~\footnote{Consider, for example, two adjacent positively charged defects at positions $i$ and $j$, respectively. If the dipole moment in between the two defects is $p$, it is easy to see that because of the conservation of $P_k=i+p$, the first defect can move at most $p$ steps to the right, and the position of the second defect must satisfy $j>i+p$.} An initial configuration close to the leaves of the tree will contain a high density of defects, i.e. contiguous regions of $++\cdots +$ or $--\cdots -$. Such regions are completely frozen under the dynamics. Therefore, a typical such configuration will contain little active puddles separated by frozen regions in real space. This is in contrast to the $tJ_z$ or the pair-flip model where a single hole or flippable pair is able to move across the entire system, and hence there is no inert region in real space. In order to melt these frozen regions, the bath must provide dipoles from the boundary. Consider for definiteness a region of $++\cdots+$ at a distance $l$ from the bath at the left boundary. Dipoles of the type $+-$ are able to melt this region from the left. However, since the bath generates $+-$ and $-+$ dipoles at random, it must first absorb all $-+$ dipoles in order to transport a $+-$ dipole. As an estimate for the timescale of absorption, consider the rightmost $-+$ dipole, which experiences a biased random walk with a drift velocity to the right. The bias comes from the presence of a finite number of dipoles to its left, which prevents it from diffusing to the left. The timescale for this dipole to overcome the bias and diffuse to the left is $O({\rm exp}(l))$. Therefore, the bottleneck due to inert regions in real space makes thermalization an exponentially slow process in this model.
	
	The real-space picture also makes it clear that the same bottleneck remains even in the presence of baths at \textit{both} endpoints. The range-three dipole-conserving model is thus more robust against contraint-violating perturbations that try to restore full ergodicity. We numerically test the above picture by coupling the system to a bath either at one boundary or both boundaries. To allow for the addition of dipoles, the bath now couples to the leftmost and/or rightmost two sites, and creates a mixture of all possible charge configurations on the two sites with equal probability. In Fig. \ref{fig:Dipole_conserving_r=3_Dynamics}, we plot the average total charge $\langle Q(t) \rangle$ as a function of time, which clearly shows a slow logarithmic-in-time growth behavior. This slow dynamics remains robust when baths are coupled to both boundaries, as expected.

	%-----------------------------------------------------------------------------------------------------------------------------------------------
	\begin{figure}
		\centering
		\includegraphics[width=0.5\textwidth]{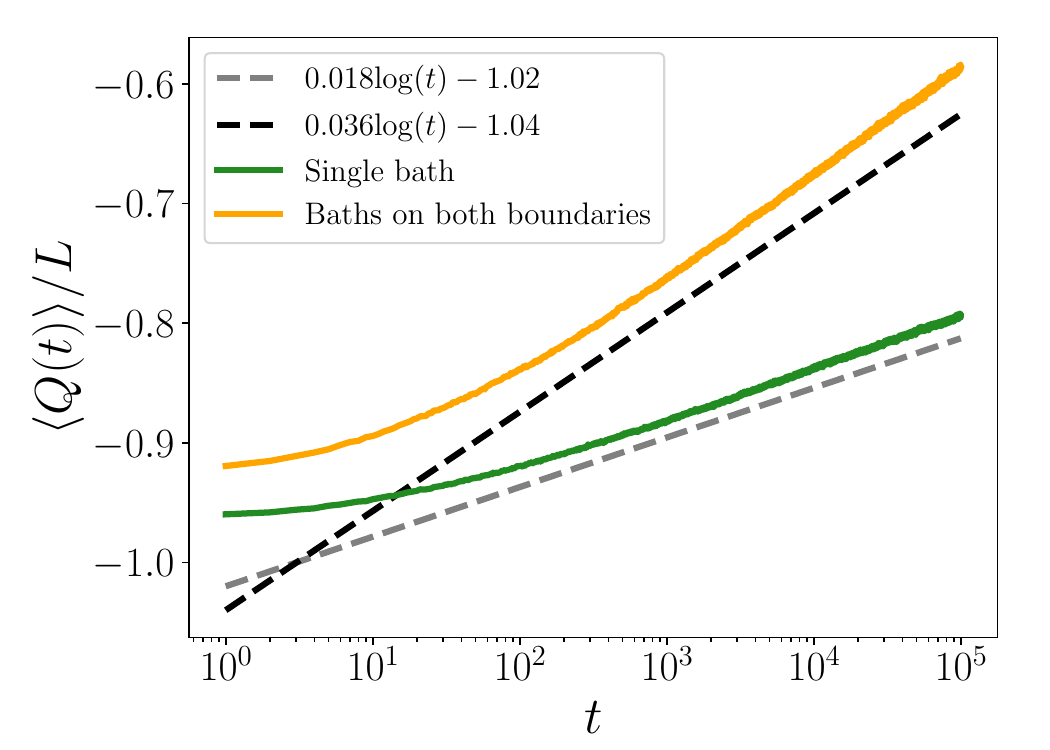}
		\caption{Numerical results for the range-three dipole-conserving model. The average total charge $\langle Q(t) \rangle$ relaxes logarithmically slowly under depolarizing bath coupled to one or both boundaries, which implies exponentially slow thermalization. Due to the real-space bottleneck, the slow dynamics remains robust even when baths are coupled to both boundaries, in contrast to the $p$-flip and $tJ_z$ model. }
		\label{fig:Dipole_conserving_r=3_Dynamics}
	\end{figure}
	%-----------------------------------------------------------------------------------------------------------------------------------------------

    \sss{Subsystem dynamics: infinite ergodicity length}

    While the range-three dipole-conserving model is able to restore ergodicity (albeit logarithmically slowly) when coupled to a stochastic bath, we will show that the subsystem dynamics in fact has an infinite ergodicity length, in a way similar to the breakdown model discussed in Sec.~\ref{sec:breakdown_ergodicity}. More specifically, we will show below that when a subsystem $A$ is embedded in a larger system, the total charge and dipole moment in $A$ can only change by an $O(1)$ amount under the constraint-preserving dynamics, irrespective of the length of the entire chain. 

    We start by pointing out the following simple fact about the range-three dipole-conserving dynamics: the boundary site of any finite system can only transition between $0\leftrightarrow +$ or $0 \leftrightarrow -$. In other words, the boundary site cannot explore all three possible charge configurations under \dyn. We establish this via induction. The statement is clearly true for $L=3$, as the only allowed moves are $0-0\leftrightarrow-+-$, $0+0 \leftrightarrow +-+$, $0-+ \leftrightarrow -+0$, and $0+-\leftrightarrow +-0$. Now suppose the statement is true for a system of size $L$, and consider taking $L \rightarrow L+1$ by adding one site to the right. Since the only way for the added site to transition between $-\leftrightarrow +$ is such that the $L$-th site is able to change from $-$ to $+$ under \dyn~restricted within system size $L$, which contradicts our assumption, we conclude that the statement is true for arbitrary system sizes.

    With the above fact established, we can show that the amount of charges transferred across a cut of the system is at most 2 under \dyn, regardless of the total system size. Consider the two sites on both sides of the cut; there are three possibilities for a charge to move from left to right: $\cdots + | 0 \cdots$, $\cdots 0 | - \cdots$, and $\cdots + | - \cdots$. For the first case, $\cdots + | 0 \cdots$, the configuration evolves to $\cdots - | + \cdots$. As previously discussed, the site on the right can only be $0$ or $+$, and the left site can be $0$ or $-$ under \dyn~restricted within the subsystem to the left and right of the cut, preventing further charge transport across the cut. In this case, the maximum amount of charge transferred from left to right is thus $\Delta q = 1$. The same argument applies to the $\cdots 0 | - \cdots$ case, which also evolves to $\cdots - | + \cdots$ with $\Delta q =1$. In the third case, $\cdots + | - \cdots$ can evolve to either $\cdots - | 0 \cdots$ or $\cdots 0 | + \cdots$. The amount of charge transferred is $\Delta q = 2$ and 1, respectively. In either case, the left region is unable to transfer more positive charge to the right. Hence, $\Delta q_{\rm max}=2$ across any cut of the system under range-three dipole-conserving dynamics, no matter how big the size of the reservoir is. We thus conclude that the subsystem dynamics cannot be ergodic, even when the size of the reservoir becomes infinite. That is $L_{\rm erg}= \infty$.

	\section{Towards a proof}\label{sec:groups}
	
	In the previous sections, we showed the conjecture to be true for a diverse array of dynamical constrains, such as those of the $tJ_z$ model, the breakdown model, and dipole conserving dynamics.  In this section we will present a more general framework for addressing the main conjecture, by relating it to a conjecture in the study of expander graphs called {\it Benjamini's conjecture}.  For the setting we study in this paper, a weaker form of the conjecture has been proven, which shows the existence of states with long (at least polynomial in $L$) relaxation times.  However, this result is not strong enough to prove exponential relaxation times.  We therefore propose a refinement of Benjamini's conjecture whose proof would imply the existence of initial states that relax in exponential time in system size, and show that this refinement holds for a large class of dynamics whose constraints are determined by multiplication laws of hyperbolic groups (in a sense to be made precise below).  While this condition may seem rather abstract, a randomly chosen constraint subject to a group property is hyperbolic with high probability, and hence our result implies that the conjecture is true for dynamics with {\it generic group constraints}.  We also discuss challenges with proving the refined conjecture in settings beyond imposing a group constraint. 
	
	\ss{Group dynamics and hyperbolic groups: review}
	
	    We now briefly review {\it group dynamics}, a class of constrained dynamics introduced in Ref.~\cite{balasubramanian2023glassy}.  These dynamics are obtained from some input group $G$, which we will always assume has a finite generating set and is finitely presented.  This means that the group can be presented as
	\begin{equation}
		G = \langle \ttg_{\a_1}, \ttg_{\a_2}, \cdots, \ttg_{\a_n} | R \rangle 
	\end{equation}
	Here the elements $\ttg_{\a_1}, \ttg_{\a_2}, \cdots, \ttg_{\a_n}$, together with their inverses, generate $G$, with the structure of the group determined by the relations $R$ relations satisfied by the generators.  We will only be interested in groups which are both finitely generated and finitely presented, with $n$ and $|R|$ both being $O(1)$ constants.  We will also abuse notation by using $\ttg_i$ to denote either an arbitrary group generator, the inverse of an arbitrary generator, or the identity element $\tte$. 
    
    We will need additional terminology before describing the dynamical constraint determined by $G$.  Given $G$, we define a {\it word} to be a sequence of generators, written as 
    % $w = \ttg_{\sigma(1)} \ttg_{\sigma(2)} \cdots \ttg_{\sigma(L)}$ where $\sigma: \mathbb{N} \to [n]$.  
     $w = \ttg_{1} \ttg_{2} \cdots \ttg_{L}$ (where again each $\ttg_i$ is either a generator, a generator's inverse, or the identity).
    Call $W_L$ the set of all words of length $L$. There exists a natural homomorphism $\varphi: W_L \to G$ whose action performs group multiplication of the generators corresponding to the characters in the word: $\varphi(w) = \ttg_1 \cdot \ttg_2 \cdot \ldots \cdot \ttg_L \in G$, where `$\cdot$' denotes group multiplication.  Two words $w_1$ and $w_2$ evaluate to the same group element under $\varphi$ if they are related by applying a sequence of relations of $G$.  By applying a relation, we mean that if a word has a length $k$ subword of the form $w_k = \ttg_1 \ttg_2 \cdots \ttg_{k}$, and a relation in $R$ sets $\ttg_1 \ttg_2 \cdots \ttg_{k} = \ttg_1' \ttg_2' \cdots \ttg_{k}'$, then we replace $w_k$ with the right hand side of the relation.  Under a sequence of these `rewriting' moves, any two words $w_1$ and $w_2$ can be connected to one another if and only if $\varphi(w_1) = \varphi(w_2)$.
	
	We can associate this kind of rewriting dynamics with a kind of quantum dynamics by associating each word with a basis element of a many-body Hilbert space of a 1D chain. The local Hilbert space at each site of the chain is taken to be $2n+1$-dimensional, with  basis vectors we label as $\{\k{\ttg_i}\}$ (with, as mentioned above, the $\ttg_i$ denoting either generators, their inverses, or the identity). A word $w = \ttg_1 \ttg_2 \cdots \ttg_L$ then corresponds to the product state $\ket{w} = \ket{\ttg_1} \otimes \ket{\ttg_2} \otimes \cdots \otimes \ket{\ttg_L}$.  The rewriting rules are represented as non-trivial matrix elements between basis states.  
  
	Let $\gamma \in R$ denote an arbitrary relation in $R$, which has the form $\gamma_{\ell} = \gamma_r$ for some words $\gamma_{\ell / r}$. We can write down a natural Hamiltonian capturing the rewriting dynamics of group $G$ as 
	\begin{equation}
		H_G = \sum_i c_\gamma \sum_{\gamma \in R} \ketbra{\gamma_\ell}{\gamma_r}_i + \text{h.c.}
	\end{equation}
	where the subscript $i$ denotes the location on the spin chain where the rewriting rule is performed and the $c_\gamma$ are arbitrary constants. Equivalently, we can also define random unitary circuit dynamics using the group constraint by constructing gates that impose the group constraint.  First define $l_R = \max_{\gamma_{\ell/r}\in R}|\gamma_{\ell/r}|$, which yields the maximum number of characters involved in the left or right hand side of any relation in $R$.  Also define $G_{l_R}$ to be the group restricted to group elements with word representations of length $\leq l_R$.  The elementary unitary gates entering the circuit can then be expressed as
	\begin{equation} \label{groupru}
		U = \bigoplus_{g \in G_{l_R}} U_g, \hspace{0.3cm} G_{l_R} = \{g \, | \, \exists \, \k w \, : \, \vp(w) = g, \,  |w| \leq l_R\},
	\end{equation}
	where $U_g$ is drawn from a Haar distribution of dimension equal to the number of words with $\vp(w) = g$ and $|w| \leq l_R$.
	
	We now describe the nature of fragmentation in these models.  First, $\varphi$ is defined such that two words have the same image if and only if they are related by a sequence of rewriting relations.  This means that if basis elements $\ket{w}$ and $\ket{w'}$ are in the same sector then they have the same image under $\varphi$, suggesting that the Krylov sectors are labelled by $\mck_g$ for $g \in G$.  The converse of the previous statement however is not quite true.  Suppose two words $w$ and $w'$ have the same image under $\varphi$.  If both of these words have length $L$ (here, equal to the system size), then the sequence of relations sending $w$ to $w'$ might require $w$'s length to increase beyond the system size before it can shrink back to $w'$.  As this is not possible, $\ket{w}$ and $\ket{w'}$ cannot belong to the same Krylov sector, a phenomenon dubbed {\it fragile fragmentation} in Ref.~\cite{balasubramanian2023glassy}.  So, the nature of the fragmentation in these group models is twofold: there is an {\it intrinsic fragmentation} into sectors $\mck_g$ and potentially further {\it fragile fragmentation} within these sectors depending on various geometric properties of the group.  We will not discuss the latter phenomenon in depth.
	
	In this section we will be interested in a large class of groups known as {\it Gromov hyperbolic groups}, which exhibit exponentially fragmented dynamics, and have certain geometric properties which allow us to place rigorous bounds on thermalization times. These groups are those whose Cayley graphs ``resemble'' hyperbolic space (but strictly speaking they cannot be isometrically embedded in such spaces).  In order to define this condition, note that the Cayley graph of any group admits a natural {\it word metric}, where the distance $d(\cdot, \cdot)$ between $g$ and $g'$ is the shortest graph distance between these nodes in the Cayley graph.  We also define the notion of a geodesic $[g,g']$ which corresponds to a path of shortest distance between $g$ and $g'$.  Finally, we define $d(h, [g, g'])$ to be the shortest distance from $h$ and the geodesic $[g, g']$.
	\begin{definition}[Hyperbolic groups]
		A group $G$ is hyperbolic if and only if its Cayley graph $\mathcal{C}_G$ satisfies the $\delta$-thin triangle property (also called the Rips thin-triangle property):  for all $u, v, w \in G$, there exists a $\delta > 0$ such that for any $h \in [u,v]$, either $d(h, [v,w]) \leq \delta$ or $d(h, [u,w]) \leq \delta$.
	\end{definition}
Intuitively, this property means that for any geodesic triangle in the Cayley graph, a ball of radius $\geq \delta$ cannot fit within it.  This is similar to what happens in hyperbolic space. 
    
	Examples of hyperbolic groups include:
	\begin{itemize}
		\item Free groups $F_n$ on $n$ generators and free products of discrete groups
		\item The modular group $SL(2;\mathbb{Z})$
		\item Random groups, with high probability
	\end{itemize}
	Regarding the third point, the following seminal result due to Gromov~\cite{gromov1987hyperbolic} provides a rigorous definition of what is meant by a random group:
    \begin{theorem}[Gromov]\label{thm:random_hyperbolic}
    Call $\mathcal{P}(n, k, \ell)$ the set of all group presentations with $n$ generators with $|R| = k$ relations, and with the property that each relation $r \in R$ satisfies $|r| \leq \ell$.  Then, with $m \geq 2$ and $k \geq 1$, 
    \begin{equation}
    \lim_{\ell \to \infty} \frac{\sum_{m \leq \ell}|\{P \in \mathcal{P}(n, k, m) \, | \, G_P \text{ is hyperbolic}\}|}{\sum_{m \leq \ell} |\mathcal{P}(n, k, m)|} = 1
    \end{equation}
    where $G_P$ denotes the group with presentation $P$.
    \end{theorem}
    A consequence of this statement is that with high probability, a randomly chosen group with fixed $\ell$ will be hyperbolic.\footnote{Gromov also introduced a more natural ``density model''.  For this setting, one fixes the number of generators to be $n$ and defines $\rho(\ell) = \log(k(\ell))/\log(N_{n, \ell})$.  Here, $k(\ell)$ is the number of cyclically reduced relations in the group presentation, with $\ell$ denoting the maximum length of a relation.  Furthermore, $N_{n, \ell}$ is the number of cyclically reduced words which have length at most $\ell$.  Intuitively, $\rho(\ell)$ indicates how constrained the group algebra is.  When $\rho(\ell) = 1$, the group is trivial, and when $\rho(\ell)$ is small, it is ``close'' to a free group.  Gromov proved that almost every group with $\lim_{\ell \to \infty} d(\ell) < 1/2$ is infinite and hyperbolic.}  This result indicates that hyperbolic groups are rather abundant, and proving thermalization properties for these groups indicates that our conjecture would be generically true.  However, when the underlying dynamics is not group dynamics,\footnote{Ref.~\cite{balasubramanian2023glassy} showed that any 1D constraint can be formulated in terms of a {\it semi}group, whose generators no longer necessarily have inverses and where the identity element need not exist. } additional subtleties arise, which we will mention later.

	\ss{Heat kernels and Benjamini's conjecture}\label{ss:math}

	In this subsection, we will show that the main conjecture can be related to a conjecture about the existence of certain kinds of expander graphs.  As in previous sections, our dynamics at each time step will consist of i) constrained circuit-averaged dynamics which maps any state in a fixed Krylov sector to the maximally mixed state in that sector, and ii) maximally depolarizing noise applied to the boundary site (though the analysis below can be readily extended to the case where noise is applied to $O(1)$ boundary sites).  We emphasize that by results in App.~\ref{app:local}, any model of dynamics which is more local will relax more slowly. 
    The thermalization time in this model is determined by the conductance of the Krylov graph $\mcg_{\mck}$, which as before we aim to show is exponentially small in the system size $L$.  Note that due to the phenomenon of fragile fragmentation, the Krylov sectors are no longer labeled by group elements $g \in G$.  However, in this situation, we will additionally cluster together Krylov sectors containing words representing the same group element, thus defining a graph $\mcg_G$.  In particular, defining $S$ to be a subset of group elements of $G$ and $R_S = \{w : \exists g \in S,\, w \sim g\}$, we write 
    \begin{equation}\label{eq:group_conductance}
    \Phi(\mcg_{G}) = \min_{S \subset G:\, \mu(R_S) \leq 1/2} \frac{\sum_{\psi \in R_S, \psi' \in R_S^c} \mu(\psi) \mel{\psi}{\mcm}{\psi'}}{\mu(R_S)}.
    \end{equation}   
    with $\mu$ the stationary distribution as usual.  The value of this ``doubly'' coarse-grained conductance indicates whether the dynamics is slow, which follows from $\Phi(\mcg_\mck) \leq \Phi(\mcg_G)$ (see App.~\ref{app:local}). Qualitatively, by clustering together all Krylov sectors corresponding to a given group element, these clusters are associated with nodes in the Cayley graph $\mathrm{Cay}(G)$\footnote{The usual notation is $\mathrm{Cay}(G, \mca)$ where $\mca$ is a generating set, but we will sometimes drop the second argument.}. Furthermore, edges between these clusters---drawn due to the boundary depolarizing term---connect nodes to their neighbors and a subset of next nearest neighbors in $\mathrm{Cay}(G)$, and therefore the connectivity of the state space inherits the topology of the Cayley graph. 
	
	We will now invoke a number of results from the group theory literature.  The reader does not need to be familiar with their proofs and thus we do not provide them.  We first provide a formal definition of an important property for groups called {\it amenability}~\cite{paterson1988amenability}:
	\begin{definition}[Amenability]
		A discrete group $G$ is amenable if it admits a finitely additive probability measure $\mu$ (satisfying $\mu(A \cup B) = \mu(A) + \mu(B)$ for $A,B \subseteq G$ with $A \cap B = \varnothing$) which additionally obeys
		\begin{equation}
			\mu(g A) = \mu(A), \,\,\,\forall A \subseteq G, \,\,\, \forall g \in G.
		\end{equation}
	\end{definition}
	While this definition may seem rather formal, it turns out to be rather general and powerful condition.  The particular facts that we need are summarized below (without proof):
	\begin{proposition}
		The following statements about a discrete group $G$ are equivalent~\cite{woess2000random}:
		\begin{enumerate}
			\item $G$ is amenable
			\item The Cayley graph $\mathrm{Cay}(G, \mca)$ for finite generating set $\mca$ does not have vertex expansion: there is no $\epsilon > 0$ such that $\forall A \subset \mathrm{Cay}(G, \mca)$ with $A$ finite, $|\partial A| \geq \epsilon |A|$.
			\item Kesten's criterion~\cite{kesten1959symmetric, kesten1959full}: the probability that a length $L$ symmetric random walk on $\mathrm{Cay}(G, \mca)$ returns to its starting point, denoted $p_L(g,g)$, decays subexponentially in $L$.
		\end{enumerate}
	\end{proposition}

    Thus, amenable groups correspond to groups whose Cayley graphs have vanishing expansion.  Therefore, we expect the dynamics corresponding to amenable groups (with boundary depolarization) to be slow.  However, they are {\it not} necessarily exponentially slow.  For instance, an example of an amenable group is $\mathbb{Z}$, and the associated group dynamics is similar to that of the symmetric exclusion process, which has relaxation time $\sim 1/L^2$.  In fact, group dynamics when $G$ is amenable can often correspond to dynamics with global symmetries or polynomial fragmentation, in which case we cannot guarantee exponential thermalization time.  Instead, it is the {\it non-amenable groups} that will be our focus.  The reason is twofold.  First, we show that they correspond to group dynamics with exponentially strong fragmentation, which is relevant to the main conjecture of the paper.  Second, although we will show that the Cayley graphs of non-amenable groups are vertex expanders (which naively exhibit rapid mixing), we argue that an important subtlety due to the large but finite system size results in the dynamics being exponentially slow.  Unlike for amenable groups which can exhibit polynomially or exponentially slow dynamics, we provide evidence that the dynamics for non-amenable groups is generically exponentially slow. 
    
    To proceed, we first link the nature of the Hilbert space fragmentation with the expansion properties of $\mathrm{Cay}(G)$.  We require following simple result regarding the conductance of the Krylov graph:
    \begin{lemma}
    Consider maximally depolarizing group dynamics with group $G$ corresponding to a Markov process with transition matrix $\mcm = \Pi \, \mcn_L$, where $\Pi$ is the intra-sector part of the dynamics and $\mcn_L$ is the (symmetric) inter-sector part induced by the boundary noise. Then    \begin{equation}\label{eq:weighted_exp}
    \Phi(\mcg_G) \leq \min_{S \subset G: \mu(S) \leq 1/2}\frac{\sum_{g 
    \in S, g' \in S^c, g \sim g'} \nu(g')}{\sum_{g \in S} \nu(g)}
    \end{equation}
    where the equivalence relation $\sim$ is defined such that $g \sim g'$ if $\exists w \in \mck_g$ and $\exists w' \in \mck_{g'}$ such that $\mel{w}{\mcn_L}{w'} \neq 0$.  Furthermore, $\nu(g) = |\mck_g|/d^L$.
    \end{lemma}
    \begin{proof}
    We bound the quantity in Eq.~\ref{eq:group_conductance}.  Noting that the stationary distribution of $\mcm$ is uniform, we can write
    \begin{equation}
    \Phi(\mcg_G) = \min_{S \subset G: \mu(R_S) \leq 1/2}\frac{\sum_{g \in S, g' \in S^c} \sum_{\psi \sim g, \psi' \sim g'} \mel{\psi'}{\Pi \, \mcn_L}{\psi}}{\sum_{g \in S} |\mck_g|}
    \end{equation}
    Writing $\ket{\phi_g} = \sum_{\psi \sim g} \ket{\phi}$ for the uniform sum over states in sector $g$---with $|\phi_g\rangle$ a steady state of $\Pi$---and suppressing the argument of the minimum above with `$(\cdot)$', we get
    \begin{align}
    \Phi(\mcg_G) &= \min_{(\cdot)}\frac{\sum_{g \in S, g' \in S^c} \mel{\phi_{g'}}{\Pi\, \mcn_L}{\phi_g}}{\sum_{g \in S} |\mck_g|} \nonumber \\ 
    &= \min_{(\cdot)}\frac{\sum_{g \in S, g' \in S^c} \frac{1}{|\mck_{g'}|}\mel{\phi_{g'}}{\mcn_L}{\phi_g} \nu(g')}{\sum_{g \in S} \nu(g)} \nonumber \\
    &= \min_{(\cdot)}\frac{\sum_{g \in S, g' \in S^c} \frac{1}{|\mck_{g'}|}\mel{\phi_{g}}{\mcn_L}{\phi_{g'}} \nu(g')}{\sum_{g \in S} \nu(g)}.
    \end{align}
    where in the third line we used the fact that $\mcn_L$ is symmetric.  Note that $\frac{1}{|\mck_{g'}|}\mel{\phi_{g}}{\mcn_L}{\phi_{g'}}$ is the probability that a uniformly selected state $\psi \in \mck_{g'}$ transitions to $\mck_g$ under the stochastic process $\mcn_L$, so $\frac{1}{|\mck_{g'}|}\mel{\phi_{g}}{\mcn_L}{\phi_{g'}} \leq 1$.  Furthermore, $\mel{\phi_{g}}{\mcn_L}{\phi_{g'}} = 0$ if $g$ and $g'$ do not satisfy $g \sim g'$.  These two facts finish the proof.
    \end{proof}
    Next, we are going to define the graph ${\mcg}_G$ whose vertices are labelled by $g$, and where an edge connects $g$ and $g'$ if $g \sim g'$ (where `$\sim$' is defined in the Lemma above).  It is quasi-isometric to $\mathrm{Cay}(G)$, which it differs from only by the addition of self edges and certain second nearest-neighbor edges in $\mathrm{Cay}(G)$.  The following proposition then holds:
	\begin{proposition}\label{prop:return}
		If $G$ is non-amenable, then group dynamics on $G$ exhibits exponentially strong Hilbert space fragmentation and the unweighted infinite graph ${\mcg}_G$ is a vertex expander (i.e. calling $V$ the set of vertices of ${\mcg}_G$, for all subsets $A \subset V$ with $|A| \leq |V|/2$ there exists an $\epsilon > 0$ such that $|\partial A|/|A| \geq \epsilon$)
	\end{proposition}  
	\begin{proof}
		To show exponentially strong fragmentation, it is sufficient to show that for all group elements $g$, we have $|\mck_g|/d^L \leq \rho^L$ for some $\rho < 1$, with $\mck_g$ designating the set of length-$L$ words with $\varphi(w) = g$ and $d=2n+1$ the dimension of the onsite Hilbert space.  We first assume the case of no fragile fragmentation, i.e. each Krylov sector is labeled by a group element $g \in G$.  Then, one can reinterpret $|\mck_g|/d^L$ simply as the probability of a length-$L$ random walk on the Cayley graph starting at node $e$ and ending at node $g$.  Call this probability $p_L(e,g)$.  Then, we can write (using transitivity of the Cayley graph)
		\begin{align}
			p_{2L}(e,e) &= \sum_{g \in G} p_L(e,g)p_L(g,e) = \sum_{g \in G} p^2_L(e,g) \nonumber \\  
			&= \sqrt{\sum_{g \in G} p^2_L(e,g)}\sqrt{\sum_{h \in G} p^2_L(e,h)} \nonumber \\ 
			&\geq \sum_{g,h \in G} p_L(e,g) p_L(e,h) \geq p_{2L}(g,h) 
		\end{align}
		where Cauchy-Schwartz was used in the last line.  This means that $\mck_e$ is the largest sector.  Since $G$ is non-amenable these return probabilities decay exponentially, so that $|\mck_g|/d^L \leq \rho^L$, and group dynamics on $G$ is exponentially strongly fragmented.  Now we consider the case of fragile fragmentation.  Then, each sector $\mck_g$ may shatter into additional sectors, and the size of the largest sector $|\mck_{\mathrm{max}}| \leq |\mck_e|$, thus implying that the group dynamics is exponentially strongly fragmented.  
        
        We use the equivalent characterization of (non)-amenable groups to prove the second part about expansion, i.e. that $\mathrm{Cay}(G)$ is a vertex expander if $G$ is non-amenable.  For some generating set $\mca$, the graph ${\mcg}_G$ almost has the topology of the Cayley graph $\mathrm{Cay}(G, \mca)$ except that it has self loops (which do not affect the expansion) and next-nearest neighbor connections (which do).  Assuming that the Cayley graph is an $\epsilon$-expander (that is, $\forall A \subset \mathrm{Cay}(\mca,G)$, $|\partial A| \geq \epsilon |A|$), in ${\mcg}_G$ the same vertex set $A$ has $|\partial A|_{{\mcg}_G} \leq (d+1) |\partial A|_{\mathrm{Cay}(\mca,G)}$ where $d$ is the degree of $\mathrm{Cay}(\mca,G)$.  This means that ${\mcg}_G$ is a $\frac{\epsilon}{d+1}$-expander.
	\end{proof}

Is the converse of this statement true?  In App.~\ref{app:proofbonanza1}, we use some additional results from group theory to show the following statement: 
\begin{proposition}\label{prop:poly_charac}
$G$ has polynomial growth iff group dynamics on $G$ exhibits polynomially strong fragmentation.
\end{proposition}
The proof of this result uses a characterization of groups with polynomial growth rate (due to Gromov) in order to compute $|\mck_e|$.  However, due to fragile fragmentation, the size of the largest sector might be much smaller than $|\mck_e|$; showing that this does not happen then proves the above proposition.  In addition, there are examples of amenable groups $G$ for which the associated group dynamics is neither polynomially fragmented nor exponentially fragmented.  Making a very reasonable technical assumption (see App.~\ref{app:proofbonanza1}), we can address the nature of fragmentation in such groups and prove: 
\begin{proposition}\label{prop:all_charac}
If group dynamics on $G$ exhibits exponentially strong Hilbert space fragmentation, then $G$ is non-amenable.
\end{proposition}
This provides an exhaustive characterization of the nature of fragmentation for group-based dynamics.

As mentioned previously, from Prop.~\ref{prop:return}, since ${\mcg}_G$ is a vertex expander, we might expect that the dynamics thermalizes quickly.  However, an important subtlety is that ${\mcg}_G$ is only an expander because it is formally an {\it infinite} graph.  To better understand this, consider the example when the Cayley graph is a tree.  We know that an infinite $k$-ary tree $\mathbb{T}_k$ is an expander for $k \geq 3$.  However, for a finite system size $L$, we instead need to consider the infinite graph restricted to a ball of radius $L$, where distances are provided by the word metric.  The dynamics with a configuration graph formed by restricting ${\mcg}_G$ to a ball may no longer rapidly mix due to certain initial configurations localized near the boundary.  In the case of a tree, while the infinite tree is an expander, the finite tree is not an expander and has a relaxation time scaling exponentially in its diameter.  Is this phenomenon generic?  This question precisely turns out to be a conjecture in the expander literature known as {\it Benjamini's conjecture}:
	\begin{conjecture}[Benjamini, \cite{benjamini1998expanders}]
		A graph $\mathcal{G}$ is said to be {\it an expander at all scales} if there exists an $\epsilon > 0$ such that all balls $B \subset \mcg$ and subsets $A \subset \mcg$ with $|A \cap B| \leq |B|/2$ obey $|\partial A \cap B| \geq \epsilon |A \cap B|$.  An expander at all scales does not exist.
	\end{conjecture}
	
	Note that the conjecture does not require $\mcg$ to be a Cayley graph, and thus this phenomenon can generalize beyond group-based dynamics.  We also note the connection to another concept called the separation profile of a graph $\mathrm{sep}(n)$, which is defined to be the maximum over size $n$ subgraphs of the minimum number of nodes that need to be removed from the subgraph to shatter it into connected components, each with size $\leq n/2$ (see Ref.~\cite{benjamini2012separation}).  Benjamini's conjecture pertains to the separation profile of an infinite expander where the subgraphs are restricted to be balls.
    
    There are several subclasses of Cayley graphs where Benjamini's conjecture is proven to be true, see \cite{brodzki2013uniform, ozawa}.  However, the statement we are looking for is slightly different from Benjamini's conjecture.  In particular, Eq.~\ref{eq:weighted_exp} describes the expansion of ${\mcg}_G$ weighted by the distribution $\nu(g)$.  Since $\nu(g) = |\mathcal{K}_g|/d^L$ is the distribution of a length-$L$ random walk on $\mathrm{Cay}(G, \mca)$ starting at $e$, we would instead want an estimate of the conductance for graphs weighted by the heat kernel measure.  
    
    Fortuitously, Fracyzk and van Limbeek showed in 2019 that Benjamini's original conjecture can be proven with the heat kernel weighting; more precisely, they proved the following theorem \cite{fraczyk2019heat}:
	\begin{theorem}[Fracyzk and van Limbeek, \cite{fraczyk2019heat}]
		Let $\mathcal{G}$ be an infinite, irreducible, bounded degree graph.  Define $\nu_{L, o}$ to be the $L$-step heat kernel measure rooted at $o$.  We say that the heat kernel measure is $\epsilon$-expanding if $\nu_{L,o}(\partial A) \geq \epsilon \, \nu_{L,o}(A)$ for all $A \subset \mcg$ with $\nu_{L,o}(A) < 1/2$ and all choices of $o$.  The heat kernel measure on $\mcg$ is not $\epsilon$-expanding for any $\epsilon > 0$.
	\end{theorem}
	Surprisingly, this theorem suggests that such $\epsilon$-expanders do not exist for {\it any} bounded degree, not just Cayley graphs.  Unfortunately, it does not tell us the {\it rate} at which this decay occurs -- this is left as an open question in Ref.~\cite{fraczyk2019heat}.  For example, in the case where $\mcg = \mathbb{T}_k$, the restriction of $\mcg$ to a ball of diameter $L$ has conductance $\exp(-O(L))$ even with the heat kernel weighting.  However, it could be possible that the restriction of $\mcg$ has conductance $1/\poly(L)$.  Since expanders look locally tree-like, one may think it is unlikely to achieve such a slowly decaying conductance, but because the conductance can be dictated by `global' properties of the graph, we cannot rule out such a situation.  Instead, we will pose the following conjecture, which we call a quantitative version of Benjamini's conjecture:
	\begin{conjecture}\label{conj:weightedbenjamini}
		Let $\mathcal{G}$ be an infinite irreducible, bounded degree graph.  There exists an $A \subset \mcg$ and an $o$ such that $\nu_{L,o}(A) < 1/2$ and 
		\begin{equation}
			\nu_{L,o}(\partial A) \geq \phi^L \, \nu_{L,o}(A)
		\end{equation}
		for $\phi < 1$.
	\end{conjecture}
	Henceforth, we will focus our attention on graphs of the form ${\mcg}_G$ with $G$ a finitely generated and finitely presented group.  Demonstrating counterexamples to the above conjecture (in the context of Cayley graphs) would require classifying boundaries of rather exotic groups, which to the authors' knowledge is an open problem and thus beyond the scope of this paper.  Here, we will content ourselves with a proof for hyperbolic groups, which are rather generic based on Theorem~\ref{thm:random_hyperbolic}.
	
	\ss{A proof for hyperbolic groups} \label{ss:hyperbolic}
	
    We will outline the main theorems that we prove as well as some of the essential ideas, but will divert most of the details to Appendix~\ref{app:proofbonanza2}.  Let us first provide some crude intuition for why the thermalization time is exponential in $L$.  To understand this, first note that the Cayley graph of a hyperbolic group corresponds to a {\it hyperbolic metric space}.  Crudely (and technically incorrectly) approximating this metric space by $n$-dimensional hyperbolic space $\mathbb{H}^n$ for some $n$, we want to compute the expansion of $\mathbb{H}^n$ restricted to a ball of diameter $L$, denoted $B(L)$, which has volume $\sim e^{(n-1) L}$.    If we slice this ball in half and call one of the regions $R$, the boundary $\partial R$ is a disk, which has volume $\sim e^{(n-2) L}$.  Therefore, the expansion of $R$ is $\sim e^{-L}$.  Unfortunately, this heuristic argument is not enough because it assumes that Cayley graphs of hyperbolic groups are isometric to $\mathbb{H}^n$ for some $n$, which is not true.

    Instead, we need a more sophisticated notion of the dimension of the group.  The way we quantify this is by using the notion of the {\it boundary} of a hyperbolic group, also known as the Gromov boundary.  The Gromov boundary directly gives us information about the structure of the group in two ways.  First, the boundary of the group is a topological space, and we show that it has finite dimension.  Second, there is a correspondence between the boundary of a hyperbolic group $\partial G$ and the group itself via a certain map $f: \partial G \times (0,D] \to G$ where the interval $(0,D]$ corresponds to the `radial' direction in $G$.  Therefore, by constructing sets with small expansion in $\partial G$, we can construct sets with small expansion in $G$ using this map.  The finite dimension of the boundary allows us to construct such small expansion sets.  These two ideas allow us to prove the following theorem for hyperbolic groups:
    \begin{theorem}[Informal]
		Consider a hyperbolic metric space $X$ corresponding to the Cayley graph of a hyperbolic group with the word metric, where $\partial X$ has non-zero dimension.  The expansion of $X \cap B(L)$ is $\leq e^{-\xi L}$ for some $\xi > 0$.
	\end{theorem}
    The proof of this theorem is in App.~\ref{app:proofbonanza2}, and relies on many ideas in the Bonk-Schramm embedding theorem of Ref.~\cite{bonk2011embeddings}.  This theorem therefore proves a quantitative version of Benjamini's conjecture but does not yet address the heat kernel version of Benjamini's conjecture (Conjecture~\ref{conj:weightedbenjamini}).  Naively speaking, so long as the heat kernel measure ``smoothly varies'' over the Cayley graph, then one should still expect the weighted expansion to decay exponentially in $L$.  In particular, taking a large subset $R$ of the Cayley graph, we expect the heat kernel measure of $R$ is $O(1)$, and since $|\partial R|$ is exponentially smaller than $|R|$ and the heat kernel varies smoothly enough, we expect the measure of the boundary to be $O(\exp(-L))$.  Once again, this intuition is roughly correct, but one needs to more carefully quantify how smoothly the heat kernel measure varies over the group.  For this, we need several properties of random walks on hyperbolic groups, namely concentration inequalities~\cite{maher2018random} and the finiteness of the entropy.  Combined with geometric properties such as the finite dimension of the Gromov boundary as well as some probabilistic arguments for showing the existence of subsets with guarantees on their heat kernel measure, we prove the following theorem:
    \begin{theorem}[Informal]
    Suppose $X$ is the Cayley graph of a hyperbolic group $G$, weighted by a probability distribution $\nu$ corresponding to an $L$-step random walk measure with step distribution $\mu$.  Assume that $\mu$ has bounded support and is symmetric.  The weighted expansion of $X$ is $\leq e^{-\xi L}$ for some $\xi > 0$.
    \end{theorem}
    The proof idea is discussed before the theorem statement in App~\ref{app:proofbonanza2} and further details can be found there.  Some parts of our proof are a sketch but can be readily made rigorous. Once again, we emphasize that due to Theorem~\ref{thm:random_hyperbolic}, hyperbolic groups are rather generic, and this proves slow thermalization times for a {\it large} class of constrained dynamics. 
    
    We end by briefly noting that there is a relationship between the vanishing expansion of the heat kernel measure and the vanishing expansion of subsets of the boundary of the group.  In particular, equipping the boundary of a hyperbolic group with a ``hitting'' probability measure, one can show that the boundary has vanishing conductance (i.e. the boundary is amenable), and this can be mapped onto the existence of subsets of the Cayley graph that also have vanishing conductance.  This is the main idea that Fracyzk and van Limbeek use to prove their theorem, and it indicates that quantitatively characterizing the conductance of the boundary of Cayley graphs of finitely generated groups will provide an answer to our Conjecture~\ref{conj:weightedbenjamini}.  We leave this interesting question to the future.
	
	\ss{Non-hyperbolic groups and beyond}\label{sec:nongroups}

    We discussed the boundary depolarizing model for group dynamics, but a natural question is to ask whether this is a good model for subsystem dynamics.  For this, we need to know the ergodicity length for group dynamics.  For this, we will assume the initial state is a random product state, which with high probability has a number of $\tte$'s that is $O(L)$.  Call the number of $\tte's$ $\alpha L$; then, if the subsystem is of size $\alpha L/2$, then we can shuttle $\alpha L/2$ $\tte's$ in the subsystem and $\alpha L/2$ $\tte's$ next to the subsystem, and use the free relations to generate $w w^{-1}$ for any $w$ in the subsystem and $w^{-1}$ outside the subsystem.  This means that the ergodicity length of any groups dynamics is $O(L)$ with high probability.  In order to get a more interesting scaling we need to turn to semigroups (see Ref.~\cite{balasubramanian2023glassy}).
	
	Within group dynamics, our results mainly pertain to a subclass of non-amenable groups, namely hyperbolic groups.  There are two natural questions one could ask.  The first is whether the dynamics remain exponentially slow for non-amenable groups that are not hyperbolic.  Addressing this question will provide a complete proof of our Conjecture~\ref{conj:weightedbenjamini} for groups.  One simple example of such a group is $H \times \mathbb{Z}$, where $H$ is a hyperbolic group,.  However, the expansion of the Cayley graph of $H \times \mathbb{Z}$ restricted to a ball should still be dictated by that of $H$, and so the dynamics should remain exponentially slow.  
    
    A more complicated example of a non-hyperbolic and non-amenable group is $SL(3;\mathbb{Z})$, which is finitely generated and presented.  In particular, this group has an interesting property not shared by any of the models we have thus far studied.  Note that words which represent the identity element correspond to oriented and based loops on the Cayley graph which are homotopy equivalent to a point.  Therefore, group dynamics within the sector of length-$L$ words representing the identity element (without depolarizing noise) correspond to the dynamics of a fluctuating loop on the Cayley graph.  If the length of the loop is roughly maintained but it has large filling (or area), then it will take a long time to contract the loop to a point, see Ref.~\cite{balasubramanian2023glassy} for more discussion of this point.  This would imply that the diameter of the graph $\mcg_{\mch}$ (ignoring additional edges created by the boundary noise) would be large.  For the group $G \cong SL(3;\mathbb{Z})$, loops of length $L$ can have filling like $\sim \exp(L)$~\cite{young2013dehn} (i.e. following Ref.~\cite{balasubramanian2023glassy}, the Dehn function is $\sim \exp(L)$) and so the diameter of $\mcg_{\mch}$ can be exponentially large in $L$ without boundary depolarizing noise.  However, proving this estimate on the diameter requires that the loop length must remain $O(L)$ during its homotopy to a point, which is related to the linear scaling of a quantity called the filling length function.  It is believed that this function scales linearly in $L$ for $SL(3;\mathbb{Z})$, although a rigorous proof is not known to the authors' knowledge.  When one turns on boundary depolarization, we expect the dynamics to be slow, but it can be of type I or II.  We can rule out type I dynamics because the depolarizing noise can ``erase'' the current word in the system and rewrite a new word in $O(L^2)$ steps, which implies that the diameter of the Krlov graph is $\text{poly}(L)$ and not $\exp(L)$.  However, because $SL(3;\mathbb{Z})$ has a similar looking structure to hyperbolic groups, we still expect it to be type II.  There are certainly other exotic groups that are non-amenable but do not look like hyperbolic groups, and it would be interesting to characterize their boundaries and thus compute their expansion.

    The second question is whether there are amenable groups that still exhibit exponentially slow dynamics when coupled to a bath.  It may be possible that groups of polynomial growth rate have this property, but that our bounds on the conductance are not strong enough to prove this.  We believe this is unlikely and conjecture that the dynamics is polynomially slow even with a one-sided infinite temperature bath.  A more interesting class of groups are those which have return probabilities scaling superpolynomially but subexponentially in $L$.  One such example is the Baumslag-Solitar group $BS(1,2)$ which has the presentation
    \be BS(n,m) = \lan \tta,\, \ttb\, | \, \texttt{b} \texttt{a}^m \texttt{b}^{-1} = \texttt{a}^n\ran.\ee
    The Krylov sector corresponding to the identity sector has size $\sim \exp(- \alpha L^{1/3}) |\mch|$~\cite{woess2000random}, indicating that this group is neither polynomially nor exponentially fragmented (based on the results in App.~\ref{app:proofbonanza1}, fragile fragmentation is not expected to affect the asymptotic scaling of the size of the identity sector).  Furthermore, this group has as similar property to $SL(3;\mathbb{Z})$ in that the intra-sector dynamics (without the boundary depolarizing noise) are slow due to large diameters.  However, in this case, we can prove that loops in the Cayley graph have large fillings {\it and} do not change in length very much during the homotopy, providing a rigorous proof that the configuration graph in the identity sector has a diameter $\sim \exp(L)$.  The Cayley graph of $BS(1,2)$ is isomorphic to $\mathbb{T}_3 \times \mathbb{R}$, and it is possible that when restricted to a ball $B(L)$, the expansion still is exponentially small in $L$.  Thus, this would provide an example of an {\it amenable} group which still exhibits dynamics of type II (when boundary depolarizing noise is added, the diameter of the Krylov graph is the diameter of the Cayley graph which no longer scales like $\sim \exp(L)$, ruling out type I dynamics).  This would also be an interesting direction for future study.  In Sec.~\ref{sss:counterexample}, we also discuss an example not based on a group which exhibits polynomial fragmentation but exponentially slow dynamics.  This suggests a larger and more interesting landscape of systems that exhibit exponentially slow dynamics. 
        
    Finally, we ask whether our results can be extended beyond group-based constraints.  Ref.~\cite{balasubramanian2023glassy} showed that the constraints of any classically-fragmented model of dynamics in 1D can be recast as {\it semigroup dynamics}.  A semigroup has the same structure as a group except that it in general lacks inverses, and consequently the only difference between semigroup and group dynamics is that the former does not possess basis states associated with the inverses of generators. The Krylov graphs of generic classically fragmented models are thus obtained from the ``Cayley graphs'' of semigroups. These graphs lack some of the nice features enjoyed by the Cayley graphs of groups: their edges are in general directed (due to the lack of inverses) and are not transitive.  However, since our dynamics is reversible, the Krylov graph in this case would simply correspond to an undirected version of the semigroup Cayley graph.  Furthermore, Fraczyk and van Limbeek's result, as well as Benjamini's conjecture, does not require transitivity of the graph, and should still be applicable to semigroups.  One complication is that for semigroup dynamics, it is possible for the maximum degree of the Krylov graph to scale with $L$.  As an example, consider a model with onsite Hilbert space ${\rm span}(\k\upa,\k\doa,\k{\star})$, with the constraints allow for all 2-site transitions involving at least one $\star$: $\k{x \star} \lra \k{y \star}, \, \k{\star x}\lra \k{\star y}$, for all $x,y\in \{\upa,\doa,\star\}$. Then the sector where at least one $\star$ is present is connected to the $2^L$ frozen sectors that lack any $\star$s, so $\max_{v \in \mcg_\mck} {\rm deg}(v) = 2^L$.  In this example however, by additionally regrouping these $2^L$ frozen sectors into a single sector, one can construct a new Krylov graph where the degree remains bounded, and study Benjamini's conjecture for this.  We additionally note that semigroup dynamics is necessary in order to get type 0 dynamics; with a group constraint, a maximally depolarizing boundary can connect any two words together, thus implying that type 0 dynamics is not possible.
    
	Thus, our conjecture about heat kernel expanders (Conjecture~\ref{conj:weightedbenjamini}) holds equally well for semigroups, provided the maximal degree of the Cayley graph remains bounded (or a reasonable ``regrouping'' procedure exists if this is not the case). As such, upon accounting for semigroups this conjecture would also imply our exponential fragmentation conjecture.

	\section{Models with fragmentation transitions} \label{sec:fragtrans}

	So far, we have demonstrated in a set of concrete examples that an exponentially fragmented Hilbert space under $\dyn$ leads to exponentially slow thermalization. In most examples discussed above 
    % \ethan{edited since class 1 models don't have to have symmetries a priori}
    (except for spin-1 breakdown model), the origin of exponential fragmentation and hence slow dynamics can be traced back to strong HSF, where $\mathcal{K}_{\rm max}$ is an exponentially small fraction of the global symmetry sector to which it belongs. In this section, we discuss models in which the degree of HSF---now as measured {\it within a fixed symmetry sector}---transitions from strong to weak as the global symmetry charge is varied (e.g. by tuning the charge/particle density). According to our definition in Sec.~\ref{sec:setup}, these models are polynomially fragmented\footnote{Notice that in these models, the total number of Krylov sectors is $N_{\mathcal K} = O({\rm exp}(L))$, and yet $|\mathcal{K}_{\rm max}|/|\mathcal{H}| = \Omega(1/{\rm poly}(L))$. This is because the total number of global symmetry sectors is only $O({\rm poly}(L))$, and the fact that there exist weakly fragmented symmetry sectors necessarily implies that $\mathcal{K}_{\rm max}$ cannot be an exponentially small fraction of the entire Hilbert space.}, and as we will show explicitly, they thermalize on timescales that are only polynomial in system sizes when coupled to a bath.
	
	\subsection{Particle-conserving East model}

	The simplest example exhibiting a fragmentation transition is the particle-conserving East model~\cite{10.21468/SciPostPhys.15.3.093, PhysRevB.108.144308}. Consider a one-dimensional system of hard-core particles with occupation number $n_i=0, 1$ on each site. The dynamics conserve the total number of particles, yet particle hoppings are subject to an East-constraint, which, in its simplest version (range $r=1$), requires that a particle can hop to the right only when the site to its left is occupied, as illustrated in Fig.\ref{fig:East_Krylov}(a). It is easy to see that the Hilbert space of this model further fractures within each U(1) symmetry sector: configurations with different leftmost particle positions remain disconnected from one another under the dynamics. Intuitively, one may anticipate that the degree of fragmentation should be different for different charge sectors, since at high fillings the East constraint essentially becomes ineffective and the system should behave like an unconstrained charge-conserving system where the Hilbert space is fully connected within the symmetry sector. This turns out to be exactly the case, as shown in Ref.~\cite{PhysRevB.108.144308}, that the system transitions from being strongly fragmented and nonergodic at low fillings $n<n_c=0.5$, to weakly fragmented and thermalizing diffusively at $n>n_c$, with the two phases separated by a continuous phase transition at $n_c$.

	%---------------------------------------------------------------------------------------------------------------------------------------------------
	\begin{figure}[!t]
		\centering
		\includegraphics[width=0.5\textwidth]{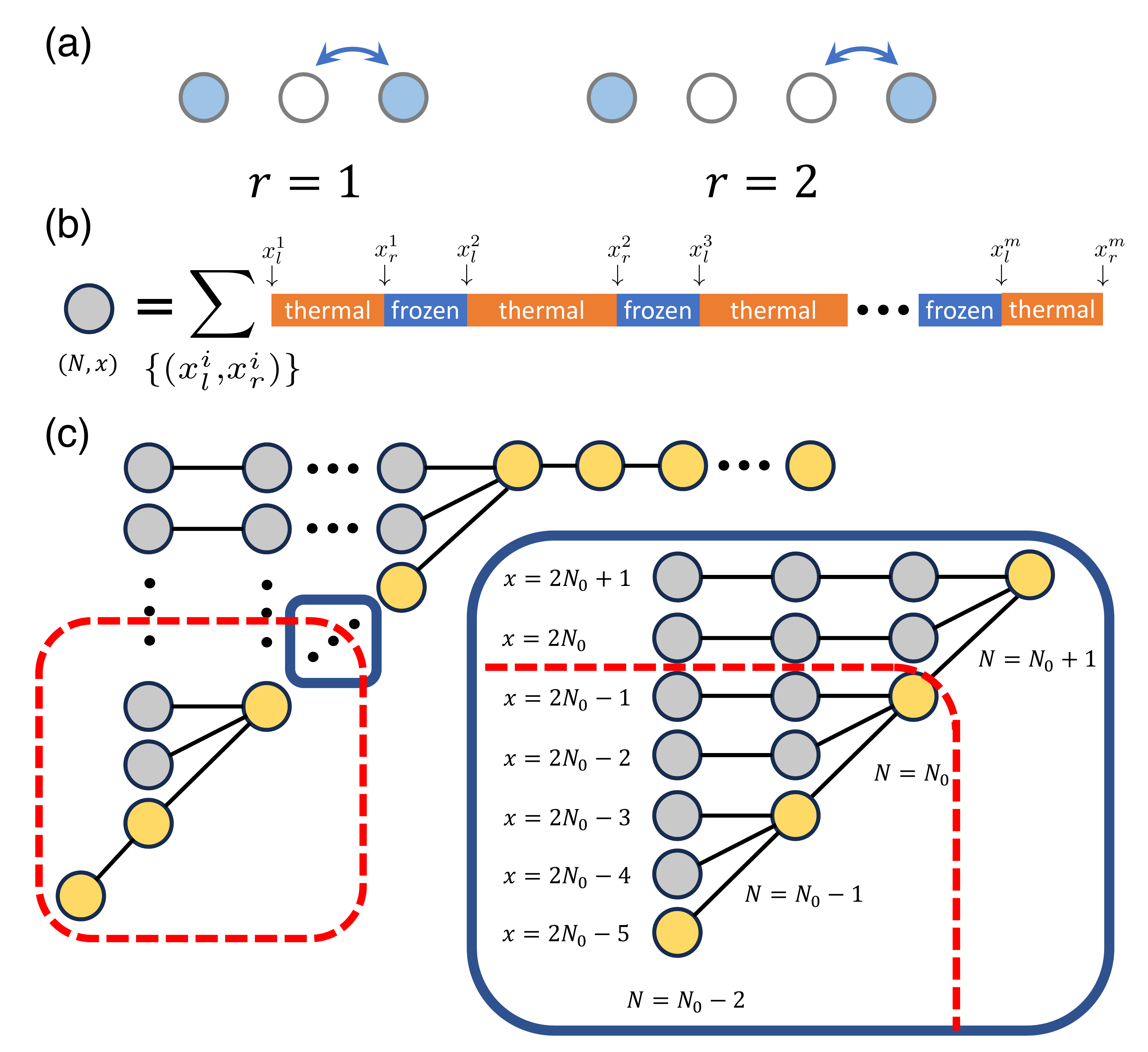}
		\caption{(a) Allowed local dynamical moves under particle-conserving East constraints. A particle is allowed to move only when there is at least 1 occupied site to its left within range $r$. In this work, we consider the simple case with $r=1$. (b) An illustration of the packed sector with a fixed total charge $N=N_0$ and rightmost point $x$ that the particles are able to reach.
			Each orange segment represents a thermal region extended to its maximal length with critical density $\frac{1}{2}$, and each blue segment represents an empty (frozen) region. The packed sector consists of all Krylov sectors (labelled by $\{(x^i_l, x^i_r)\}$) with the same $(N,x)$, satisfying $\sum_{i=1}^m (x_r^i-x_l^i+2)=2N,~x_r^m=x$. (c) The Krylov graph for the particle-conserving East model. Each circle represents a packed sector as illustrated in (b), and the circles are connected with bonds due to the thermal bath coupled to the boundary. The figure on the right is a zoom-in of a region on the left. The top right region of the graph where the vertices are connected horizontally corresponds to $N\geq \frac{L}{2}$, such that $x=L$ and cannot further increase. }
		\label{fig:East_Krylov}
	\end{figure}
	%---------------------------------------------------------------------------------------------------------------------------------------------------
	
	\sss{Real-space picture}
	
	The dynamics of this model also admits a simple real-space picture. For a given configuration with an average particle density below $n_c$, there exists regions whose local charge densities are above $n_c$ and hence locally thermal. Particles in such regions are able to spread out up to a maximal distance at which point the local density drops to $n_c$, and the thermal region cannot grow any further. Therefore, under $\dyn$ the system contains locally thermal regions separated by frozen regions in real space, and the frozen regions constitute a finite fraction of the full system when $n<n_c$. As $n$ increases past $n_c$, the thermal regions are able to merge and the fraction of the frozen sites vanishes.

	Now consider coupling the left endpoint to a bath (particle reservoir) that randomly adds or removes one particle at each time. More precisely, since the particle on the leftmost site is immobile, we instead couple the bath to the second leftmost site while keeping the leftmost site occupied at all times. With the real-space picture described above in mind, it is easy to anticipate the dynamics under a bath. Starting from a random initial particle configuration with $n<n_c$, a local thermal region is seeded in the vicinity of the bath, which expands to larger and larger distances as more particles are pumped into the system at the boundary. As the dynamics inside the thermal region is fast, one expects that the spreading of the thermal region is also fast, until the entire system thermalizes.
	
	%---------------------------------------------------------------------------------------------------------------------------------------------------
	\begin{figure}[!t]
		\centering
		\includegraphics[width=0.5\textwidth]{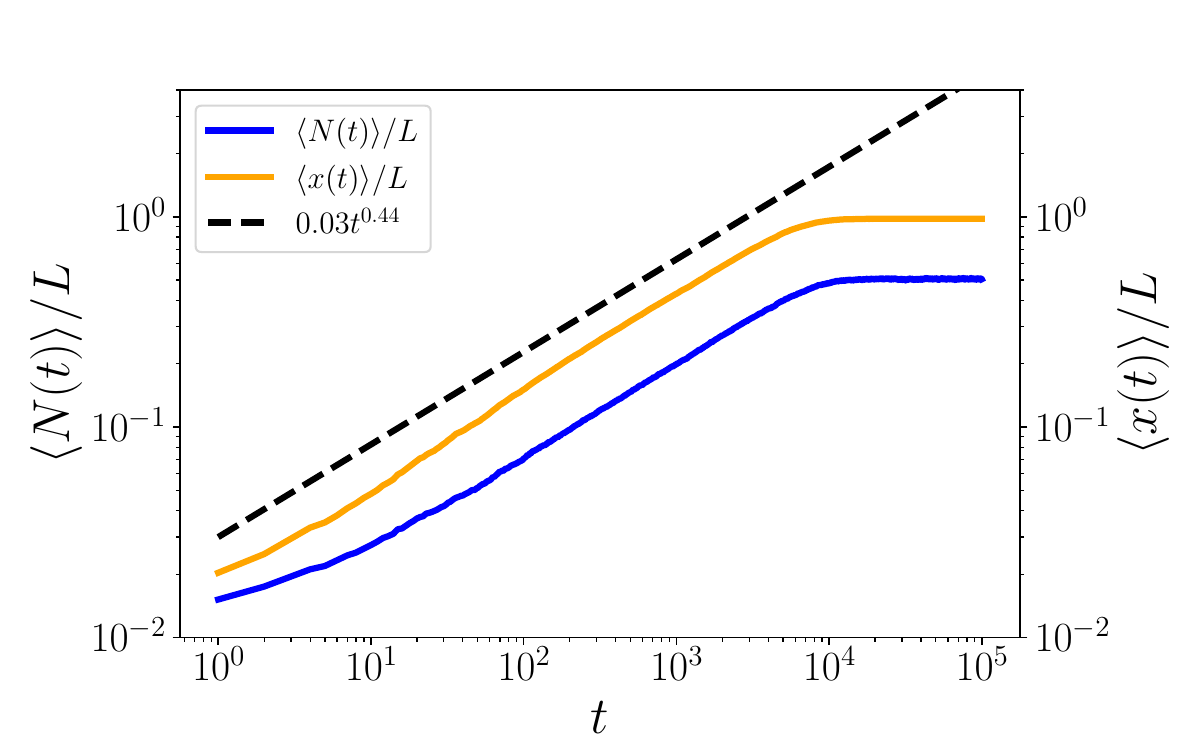}
		\caption{Numerical results for the particle-conserving East model. The blue line represents the average total particle number and the yellow line represents the size of the maximal region that the particles can spread out at a given time. We find that these two quantities behave in a similar fashion, both showing a polynomial-in-time growth, which implies fast relaxation.}
		\label{fig:East_Dynamics_numerics}
	\end{figure}
	
	%---------------------------------------------------------------------------------------------------------------------------------------------------

	To corroborate the above picture, we numerically study the relaxation of total charge starting from an initial configuration $|\bullet \circ \circ \circ \circ \cdots\rangle$, where $\bullet$ and $\circ$ denotes an occupied and empty site, respectively. In Fig.\ref{fig:East_Dynamics_numerics}, we show the average total particle number $\langle N(t) \rangle$, and the average size of the thermal region $\langle x(t) \rangle$ as a function of time. The results clearly demonstrate that both quantities increase polynomially in time with identical exponents, indicating fast relaxation.

	\sss{Krylov graph}
	
	While the above real-space picture is intuitive, it is possible to understand the thermalization process in terms of the Krylov graph. An interesting property of the particle-conserving East model is that the maximal range to which an active region can grow is completely determined by the number of particles in that region~\cite{PhysRevB.108.144308}. To see this, consider the following particle configuration
	\begin{equation}
		\underbrace{\bullet \bullet \bullet \cdots \bullet}_N \circ \circ \cdots \circ.
	\end{equation}
	The maximum distance that the particles can spread to the right is given by $x=2N-1$, which corresponds to the most dilute configuration allowed by the constraint
	\begin{equation}
		\underbrace{\bullet \circ \bullet \circ \bullet \circ \cdots \bullet}_{2N-1} \circ \circ \circ \circ. 
	\end{equation}
	At this point, the particle density of the thermal region self-tunes to the critical value $n_c=\frac{1}{2}$ and ceases to grow any further. In general, a particle configuration will contain spatially disconnected thermal regions separated by empty (frozen) regions. 
	Therefore, each Krylov sector can be uniquely labelled by a collection of coordinates $\{(x^i_l, x^i_r)\}$, which specifies the boundaries of each thermal region in the system when expanded to the maximal length [see Fig.~\ref{fig:East_Krylov}(b)].
	
	Upon coupling to a bath at the boundary, different Krylov sectors become connected. Rather than directly constructing the Krylov graph, it turns out to be useful to coarse-grain further and construct a packed sector consisting of all Krylov sectors with the same particle number $N$ and coordinate of the rightmost site $x$ that the particles can reach: $(N, x)$. The reason is that, while the total particle number obviously changes by $\pm 1$ per time step, the farthest distance that the particles can expand to the right serves as another useful ``quantum number" to organize various Krylov sectors with the same total charge. Consider a simple example of three particles: $\bullet \bullet \bullet \circ \circ \cdots \circ$. In this case, the rightmost position that the particles are able to reach is $x=5$. Suppose the bath removes one particle from the second site on the left, such that $(N=3, x=5) \rightarrow (N'=2, x')$. Now the value of $x'$ depends on the specific particle configuration within the Krylov sector when the system is put in contact with the bath. For instance,
	\begin{enumerate}
		\item $\bullet \bullet \bullet \circ \circ \cdots \circ \xrightarrow[]{{\rm bath}} \bullet \circ \bullet \circ \circ \cdots \circ, \quad (N'=2, x'=x-2)$
		
		\item $\bullet \bullet  \circ \bullet \circ \cdots \circ \xrightarrow[]{{\rm bath}} \bullet \circ \circ \bullet \circ \cdots \circ, \quad (N'=2, x'=x-1)$
		
		\item $\bullet \bullet \circ  \circ \bullet \circ \cdots \circ \xrightarrow[]{{\rm bath}} \bullet \circ \circ \circ \bullet \circ \cdots \circ, \quad (N'=2, x'=x)$.
	\end{enumerate}
	One can analyze the reverse process where the bath adds one particle in a similar way. We thus group the Krylov sectors according to $(N, x)$, and the coarse-grained Krylov graph is depicted in Fig.~\ref{fig:East_Krylov}(c). We find that the Krylov graph again has a self-similar tree-like structure. 
	In Appendix~\ref{app:east}, we calculate the expansion $\Phi(G_{\mathcal{K}})$ for the cut shown in Fig.~\ref{fig:East_Krylov}(c), which yields $\Phi(G_{\mathcal{K}}) = \frac{1}{2(2N-1)}$, only polynomially small in the total charge (and system size).

	\subsection{Range-four dipole-conserving model}
	\label{ss:dipole_four}
	
	Another class of models exhibiting a fragmentation transition are dipole-conserving (fracton) systems with local dynamical moves involving $r>3$ consecutive sites~\cite{PhysRevB.101.214205, PhysRevB.107.045137}. Unlike the range-three fracton model discussed in Sec.~\ref{ss:dipole} which does not have a weakly fragmented phase and suffers from real-space bottleneck, fractonic systems with longer-range dynamical moves do have a fragmentation transition as the charge density is varied, similarly to the U(1)-symmetric East model. We consider $r=4$ in this subsection. The scenario of the fragmentation transition very much resembles that of the East model, and previous studies have found that the transitions in these two classes of models seem to belong to the same universality class~\cite{PhysRevB.108.144308}. In particular, the existence of local thermal regions that are able to grow until the density self-tunes to the critical point also applies to fractonic systems~\cite{PhysRevB.101.214205, PhysRevB.107.045137}. We thus expect that charge relaxation is also fast in this case when coupled to a bath.
	
	%---------------------------------------------------------------------------------------------------------------------------------------------------
	\begin{figure}[!t]
		\centering
		\includegraphics[width=0.5\textwidth]{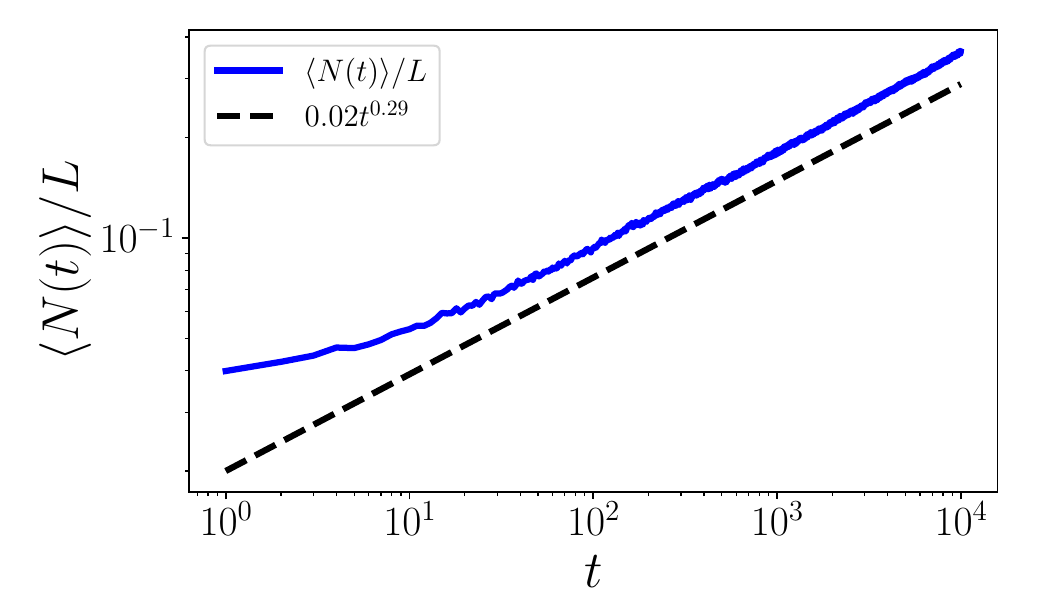}
		\caption{Numerical results for the range-four dipole-conserving model. The average total particle number increases polynomially in time, indicating fast relaxation. The behavior is to be contrasted with the range-three dipole-conserving model, where the total charge (particle number) relaxes logarithmally slowly (Fig.~\ref{fig:Dipole_conserving_r=3_Dynamics}).}
		\label{fig:Dipole_conserving_r=4_Dynamics}
	\end{figure}
	
	%---------------------------------------------------------------------------------------------------------------------------------------------------

	In Fig.~\ref{fig:Dipole_conserving_r=4_Dynamics}, we show numerical results for the average total particle number $\langle N(t)\rangle$ as a function of time, starting from an initial state with no particle, similarly to the setup considered in Sec.~\ref{ss:dipole}. In contrast to the logarithmically slow charge relaxation observed in range-three fracton systems, here the results clearly demonstrate a fast polynomial-in-time relaxation dynamics.

	\section{Discussion} \label{sec:disc}
	
	In this work we have considered the thermalization dynamics of constrained spin chains connected to a thermalizing bath at one end. Our study has focused on constraints which strongly fragment Hilbert space (in the absence of the bath coupling), with the largest sector occupying an exponentially small fraction of Hilbert space. We have provided strong evidence for the conjecture that in the presence of the thermalizing bath, the thermalization time $\tth$ is exponentially long in system size, regardless of the details of the dynamics or specific choice of constraint. This evidence came in the form of proving exponentially long bounds on $\tth$ in a large family of models, and connecting this conjecture to a related conjecture in the mathematics of expander graphs. 
	
	In what follows we briefly describe a few additional lines of inquiry that would be interesting to pursue in future work. 

	\sss{Slow thermalization without exponential fragmentation}\label{sss:counterexample}
	
	While our conjecture is that exponentially strong HSF is a sufficient condition for exponentially slow thermalization, it is by no means necessary, and the converse of our conjecture is not true. An explicit example illustrating this fact 
	 may be constructed by adding a certain type of `dynamical impurity' to the spin-1 breakdown model. 	
	One does this by augmenting the Hilbert space of the breakdown model by a state $\k{A}$, with the dynamics conserving the number of $A$ particles $\sum_i \proj{A}_i$. The $A$ particles can move to the left only by absorbing particles in the breakdown model, so that the only allowed matrix elements that move an $A$ from site $i$ to site $i+1$ are of the form $\kb{n,A}{A,0}_{i,i+1}$ for $n\in \{1,2\}$ (together with the Hermitian conjugate thereof). We furthermore allow $A$s to modify the breakdown-sector degrees of freedom to their right, so that they dynamics also includes matrix elements $\kb{A,m}{A,m'}_{i,i+1}$ with $m,m'\in\{0,1,2\}$. 
	
	With this construction, it is easy to see that for an initial state that includes a pattern like $\k{A0^lA0^{L-l-2}}$, it will take a time $\sim 2^l$ for the left $A$ to provide the right $A$ with the `fuel' it needs to move. This argument implies an exponentially long thermalization time reflected in the relaxation of the $A$ density $\r_A$, which in more generic states we expect to relax with a diffusion constant $D \sim 2^{-1/\rho_A}$. 

    Additionally, within the context of group-based dynamics, there could be other candidates for non-exponentially fragmented dynamics that are exponentially slow.  One such example could be polynomially fragmented groups due to intra-sector connectivity, though as remarked in Sec.~\ref{sec:nongroups} this is unlikely.  Another example is a certain kind of Baumslag-Solitar group which is neither polynomially nor exponentially fragmented but whose Cayley graph has a branching structure due to the exponential growth rate of the group.

	\sss{Other types of open system dynamics}
	
	When the bath couples to a spatially local collection of $O(1)$ sites, exponentially fragmented dynamics retains memory of its initial conditions for times exponetially long in $L$. If {\it all} sites couple to a depolarizing bath, or if at each time step a depolarizing bath couples to a random fraction of the sites, the story is drastically different: in these cases, if the strength of the coupling to the bath is $\l$ the thermalization time can be shown to satisfy $\tth = O(\frac1\l \log(L))$, which follows as a consequence of the results in \cite{aharonov1996limitations}. A bath coupled at all sites thus heats the system doubly-exponentially faster than a local bath. It could thus be interesting to explore couplings intermediate between strictly local and fully global. As an example, when the coupling occurs near the boundary and is taken to decay exponentially in space, we expect that models in all classes will have still $\tth = O(\exp(L))$, but this scaling will change when the coupling becomes more long range. Future work could also explore the influence of bath non-Markovianity on $\tth$.

	\sss{Towards quantum dynamics}
	
	This paper has mostly focused on classical dynamics or random unitary circuit dynamics (which reduces to the former for observables linear in the density matrix after circuit averaging).  An important question is whether Hamiltonian quantum dynamics with energy conservation can yield qualitatively different dynamical properties.  For instance, it could be possible that $\dyn$ thermalizes {\it faster} in the quantum setting a) due to constructive quantum interference (which is often difficult to achieve in a local system, see Refs.~\cite{childs2003exponential, balasubramanian2023exponential} for examples in non-local systems), or b) if $\dyn$ has additional time-dependence that steers the dynamics in $\mch$ (in which case the time dependence of $\dyn$ would need to depend on the initial state).  In the absence of these kinds of mechanisms, we conjecture that the dynamics in the Hamiltonian setting is similarly slow for strongly fragmented systems. In fact, for systems of type I, where the slowness originates from large Krylov graph diameters, this was explicitly shown in Ref.~\cite{balasubramanian2023glassy}.
	
	\sss{Higher dimensions}
	
	An obvious question is the status of our conjecture in higher dimensions. It is easy to construct examples of strongly fragmented models in $d>1$ which are in class 0 (persistently non-ergodic; e.g. $d>1$ PXP models) or class I (exponentially large Krylov-graph diameters; e.g. $d>1$ models with exponentially modulated symmetries). Understanding to what extent models in class II (Krylov graphs with small diameters but strong bottlenecks) exist in $d>1$ is an interesting question.\footnote{In this situation, we are looking for intrinsically 2D examples and not stacks of 1D models.} Since for $d>1$ the number of sites that couple to the bath is extensive, the nodes of the Krylov graph will generically have thermodynamically large degrees. New techniques may thus be needed to address this question in generality. Studying the models of \cite{lehmann2023fragmentation} could provide a fruitful starting point.  It would also be interesting to study relaxation in the presence of an infinite bath for the two dimensional group loop models introduced in Ref.~\cite{balasubramanian2023glassy}, of which several subclasses where studied in Refs.~\cite{stephen2022ergodicity, Balasubramanian2023, stahl2023topologically, khudorozhkov2024robust}.

	\section*{Acknowledgements} 
	We thank Ian Agol, Itai Benjamini, Mikolaj Fraczyk, Rahul Nandkishore, Elia Portnoy, Pablo Sala, Alistair Sinclair, Ewin Tang, and Wouter van Limbeek for helpful discussions and correspondence, Sajant Anand and Ruihua Fan for helpful discussions about the spin-1 breakdown model at the early stages of this work, and Alexey Khudorozhkov for discussions and collaboration on related work. 
	This research is supported in part by the National Science Foundation under Grant No. DMR-2219735 (Y. H. and X. C.), Grant No. PHY-2325080 and a grant from the Simons Foundation (MP-SIP-00001553, AWH) (S. B.), Grant No. 12375027 from the National Natural Science Foundation of China and a Peking University startup fund (Z.-C. Y.), and a Miller research fellowship (E.L.). Numerical simulations were performed on the High-performance Computing Platform of Peking University.

	\appendix
    
    \section{Local dynamics versus maximally depolarizing dynamics}\label{app:local}

    In the main text, we introduced and studied a maximally depolarizing model of dynamics.  In this appendix, we will show that the thermalization times of any generic local dynamics with constraints are strictly longer than that of the maximally depolarizing model.  As in the main text, we define the probability distribution
	\begin{equation}
		p_{\psi}(t) = \mathbb{E}_{U_t}[\mel{\psi}{\rho(t)}{\psi}]
	\end{equation}
	where $\ket{\psi}$ is a basis state and we assume that $\rho(t) = U_t^{\dagger} \rho(0) U_t$ where $U_t$ is a depth-$t$ circuit built from local gates (which act on $l$ qubits) satisfying the Hilbert space constraint.  Under the Haar average (assuming no correlations between the gates), $\rho(t)$ becomes diagonal and the probability vector $p(t)$ evolves under time as 
	\begin{equation}
		p(t+1) = \mathcal{M}_1 \mathcal{M}_2 \mathcal{M}_3 \cdots \mathcal{M}_{l}\mathcal{N}_L \, p(t).
	\end{equation}
	Here, $\mathcal{M}_i$ denotes the tensor product of a sequence of $l$-site transition matrices globally shifted by $i$ sites which respect the Hilbert space constraint, and
    $\mathcal{N}_L$ denotes the transition matrix corresponding to the depolarizing channel applied to the last site.  Therefore, $\mathcal{M}_1 \mathcal{M}_2 \mathcal{M}_3 \cdots \mathcal{M}_{l}$ can be considered to be part of a brickwork random unitary circuit which we call a ``single period'' of the circuit; under the assumption that $\mathcal{N}$ is applied after every $k$ periods, we may write
	\begin{equation}
		p(k, t+1) = (\mathcal{M}_1 \mathcal{M}_2 \mathcal{M}_3 \cdots \mathcal{M}_{l})^k\mathcal{N}_L \, p(k, t)
	\end{equation}
	and in the limit that $k \to \infty$, this gives us the maximally depolarizing dynamics $p(\infty, t+1) = \Pi\, \mathcal{N}_L \, p(\infty, t)$ where $\Pi$ is a projector onto the steady state in all Krylov sectors.  If the Markov chain is ergodic (which is true so long as the dynamics is not class 0), the stationary distribution of this Markov chain is uniform over all states, i.e. $p_{\psi}(k,\infty) = \frac{1}{d^L}$ regardless of $k$, where $d$ is the local Hilbert space dimension.  

    For group-based dynamics (see Ref.~\cite{Balasubramanian2023}), it is suggestive to label the Krylov sectors by an element of the group $g \in G$.  However, due to the phenomenon of fragile fragmentation, this is not true: it may not possible traverse between two words $w$ and $w'$ corresponding to the same group element because $w$ may need to first grow beyond the size of the system $L$ before it can shrink to become $w'$.  This kind of bottleneck causes the sector of words $\mck_g = \{w: |w| = L; \, w \sim g\}$ to further shatter into many sectors: see Ref.~\cite{Balasubramanian2023} for more discussion of this phenomenon.  Therefore, in addition to proving that local dynamics is strictly slower than the maximally depolarizing model, we also want to show that additionally coarse-graining the conductance to account for fragile fragmentation does not quantitatively affect our results.

    First we show that for any fragmented dynamics (which need not obey a group structure), local dynamics is strictly slower than the maximally depolarizing dynamics:

    \begin{lemma}
		Let $\mcg_{\mathrm{loc}}$ denote the configuration graph of local depolarizing dynamics, i.e. one for which the transition matrix is given by
        \begin{equation}
        \mcm_{\mathrm{loc}} = \mathcal{M}_1 \mathcal{M}_2 \mathcal{M}_3 \cdots \mathcal{M}_{l_R}\mathcal{N}_L
        \end{equation}
        for $\mathcal{M}_i$ a bistochastic matrix with the uniform distribution as a left and right eigenstate with eigenvalue $1$. Furthermore let $\mcg_{{\rm loc},\mck}$ denote the graph $\mcg_{\rm loc}$ which is then coarse-grained according to the same procedure that coarse-grains $\mcg$ to $\mcg_\mck$ for maximally depolarizing dynamics. Then, $\Phi(\mcg_{\mathrm{loc}, \mck})$ satisfies 
        \begin{equation}
        \Phi(\mcg_{\mathrm{loc}, \mck}) = \Phi(\mcg_{\mck})
        \end{equation}
        with $\Phi(\mcg_{\mck})$ denoting the coarse-grained graph conductance of the maximally depolarizing dynamics.  
    \end{lemma}
    Thus, if the thermalization time of the coarse-grained maximally depolarizing dynamics is $\geq \exp(\alpha L)$, then the thermalization time of the local depolarizing dynamics is also $\geq \exp(\alpha L)$, following from the fact that $\Phi(\mcg_{\text{loc}}) \leq \Phi(\mcg_{\mathrm{loc}, \mck}) = \Phi(\mcg_{\mck}) $.
    \begin{proof}
        
        Recall that the graph conductance $\Phi(\mcg)$ for weighted graph $\mcg$ with stationary distribution $\mu$ and transition matrix $\mcm$ is defined to be
		\begin{equation}
			\Phi(\mcg) = \min_{R \subset G:\, \mu(R) \leq 1/2} \frac{\sum_{\psi \in R, \psi' \in R^c} \mu(\psi) \mel{\psi}{\mcm}{\psi'}}{\mu(R)}.
		\end{equation}
        To prove the inequality, first note that the stationary distribution of both the local and non-local models are uniform.  Since the proposition concerns the coarse-grained conductance, we only minimize over sets $R_S$ which are expressible as $R_S = \bigcup_{s \in S} \mck_s$ for some subset $S$ of the set of all Krylov sectors $K$.  Denote by $\ket{\pi(R)}$ the uniform superposition over states in $R$.  Then, we can write
        \begin{equation}
			\Phi(\mcg_{\mck}) = \min_{S \subset K:\, \mu(R_S) \leq 1/2} \sqrt{\frac{\mu(R_S^c)}{\mu(R_S)}} \mel{\pi(R_S)}{\mcm}{\pi(R_S^c)}.
		\end{equation}
        For the maximally depolarizing dynamics, we have $\mcm = \Pi \, \mathcal{N}_L$ where $\Pi = \mcm_1 \cdots \mcm_{l_R}$; since $\bra{\pi(R_S)} \Pi = \bra{\pi(R_S)}$ for any $S$ (which follows from a uniform superposition of states in a Krylov sector being a left steady state of $\Pi$), $\mel{\pi(R_S)}{\mcm}{\pi(R_S^c)} = \mel{\pi(R_S)}{\mathcal{N}_L}{\pi(R_S^c)}$.  For the local depolarizing dynamics,
		\begin{align}
			\mel{\pi(R_S)}{\mcm_{\text{loc}}}{\pi(R_S)} &= \mel{\pi(R_S)}{\mathcal{M}_1 \mathcal{M}_2 \cdots \mathcal{M}_{l_R}\mathcal{N}_L}{\pi(R_S^c)} \nonumber\\ &= \mel{\pi(R_S)}{\mathcal{N}_L}{\pi(R_S^c)}
		\end{align}
		where in the second equality we use the fact that $\bra{\pi(R_S)}$ is a left stationary distribution for $\mathcal{M}_i$ and thus a stationary distribution for $\mathcal{M}_1 \mathcal{M}_2 \mathcal{M}_3 \cdots \mathcal{M}_{l_R}$.  Thus, it follows that $\Phi(\mcg_{\mathrm{loc}, \mck}) = \Phi(\mcg_{\mck})$.
	\end{proof}
    
    As previously mentioned, for group based dynamics which exhibits fragile fragmentation, the Krylov sectors cannot be labeled uniquely by group elements.  Let $\mcg_G$ denote the graph obtained by coarse-graining $\mcg_\mck$ by grouping all sectors associated with a given group element into a single node. Call $S$ a subset of group elements of $G$ and $R_S = \{w : \exists g \in S,\, w \sim g\}$, and consider then the conductance 
    \begin{equation}
    \Phi(\mcg_{G}) = \min_{S \subset G:\, \mu(R_S) \leq 1/2} \frac{\sum_{\psi \in R, \psi' \in R^c} \mu(\psi) \mel{\psi}{\mcm}{\psi'}}{\mu(R)}.
    \end{equation}
    Since $\mcg_G$ is obtained by coarse-graining $\mcg_\mck$, for every $S$ there exists some $S' \subset K$ constructed from $S$ such that $R_S = \bigcup_{s \in S'} \mck_s$.  Because of this, the following holds:    
	\begin{lemma}
		For maximally depolarizing dynamics for a local constraint based on a presentation of group $G$,
		\begin{equation}
			\Phi(\mcg_{\mck}) \leq \Phi(\mcg_{G}).
		\end{equation}
	\end{lemma}
	\begin{proof}
		Follows from the definitions.
	\end{proof}
    Therefore, for studying group dynamics in Section~\ref{sec:groups}, it suffices to compute $\Phi(\mcg_{G})$, which can be more readily related to properties of the Cayley graph of $G$.

    \section{Equivalence of exponentially fragmentation and non-amenability }\label{app:proofbonanza1}
	
    In this appendix we will prove Prop.~\ref{prop:poly_charac} and Prop.~\ref{prop:all_charac} from the main text.  Proofs of these statements will ultimately provide a complete characterization of all group dynamics which exhibit exponentially strong fragmentation, and implies the equivalence:
    \begin{equation*}
    G\text{ non-amenable} \leftrightarrow \dyn_G\text{ exponentially fragmented}
    \end{equation*}
    
    As we noted in the main text, while the non-amenability of a group implies that the corresponding fragmentation of the group dynamics is exponentially strong, the converse is not a priori clear.  We will show that if the group dynamics exhibits polynomially strong fragmentation or weak fragmentation, then the group has to be amenable.  The most subtle part of the proof is dealing with fragile fragmentation, which can further fracture $\mck_g$.  Thus it is not immediately clear how to estimate the size of the largest component even when accounting for fragile fragmentation.  After proving the result for group dynamics with polynomially strong fragmentation, we show that all remaining amenable groups cannot exhibit exponentially strong fragmentation, even when accounting for fragile fragmentation.  For this result, we make a mild assumption which is very likely true.

    To prove the first statement about polynomial fragmentation, we note a few important definitions and results: 
    \begin{definition}[Nilpotent and virtually nilpotent groups]
    A group $G$ is class $k$ nilpotent iff it admits a central series with length $k$:
    \begin{equation}
    \{1\} = G_1 \triangleleft G_2 \triangleleft G_3 \triangleleft \cdots \triangleleft G_k = G
    \end{equation}
    where $G_i \triangleleft G_{i+1}$ means that $G_i$ is a normal subgroup of $G_{i+1}$.
    A group $G$ is virtually nilpotent if there exists a subgroup $H$ of finite index which is nilpotent.
    \end{definition}
    \begin{definition}[Growth rate]
    The growth rate of a group $G$, denoted $N(r)$, is the number of group elements that are expressible as words of length $\leq r$.  A polynomial growth rate of degree $n$ is one where $N(r) = O(r^{n})$.
    \end{definition}
    A seminal result by Gromov characterizes all groups with polynomial growth rate:
    \begin{theorem}[Gromov~\cite{gromov1981groups}]
    A finitely generated and presented group $G$ has polynomial growth rate iff it is virtually nilpotent.
    \end{theorem}
    The following is a consequence of seminal results, many of which can be found in~\cite{woess2000random}:
    \begin{theorem}\label{thm:poly_rw}
    A finitely generated and presented group $G$ has polynomial growth rate of degree $n$ iff a symmetric random walk asymptotically has $p_{2L}(e,e) \sim L^{-n/2}$.
    \end{theorem}
    The above theorems, along with the result from the main text showing that $\mck_e$ is the largest sector, would imply that group dynamics which is polynomially fragmented is equivalent to the group having polynomial growth and thus being amenable.  However, due to fragile fragmentation we need to further characterize the size largest sector {\it within} $\mck_e$ in order to determine whether the group dynamics still exhibits polynomial fragmentation.  For this, we will need to introduce the notion of a {\it filling length} (a similar quantity called the expansion length was discussed in Ref.~\cite{balasubramanian2023glassy}).  A good reference for these concepts is Ref.~\cite{riley2006filling}.
    \begin{definition}[Filling length]
    Consider a group $G$, and a length-$L$ word representing the identity element which does not contain any $\tte$'s.  Call $w \to w_1 \to w_2 \to \cdots w_m \to \varnothing$ a derivation of $w$ if each pair of consecutive words differs by the application of a single relation and $\varnothing$ denotes the empty word.  Call $\mathrm{FL}(w, \varnothing)$ the minimum over derivations from $w$ to $\varnothing$ of the maximum value of the length of an intermediate word in the derivation.  Then,
    \begin{equation}
    \mathrm{FL}(L) = \max_{w \sim \varnothing: |w| = L} \mathrm{FL}(w, \varnothing).
    \end{equation}
    \end{definition}
    Furthermore, we have the following result regarding the filling length for nilpotent groups
    \begin{proposition}[Gersten, Holt, Riley,~\cite{gersten2003isoperimetric}]
    For nilpotent and virtually nilpotent groups, $\mathrm{FL}(L) \sim L$.
    \end{proposition}
    Next, we prove that a randomly drawn word representing the identity element of a group with polynomial growth rate will have a large fraction of $\tte$'s:
    \begin{lemma}
    Consider a group $G$ with polynomial growth rate.  If the number of length-$L$ words representing the identity is $|\mck_e|$, with probability $1-\epsilon$, the number of $\tte$'s in a randomly chosen word $w \in \mck_e$ is $\geq \eta(\epsilon) L$, where $\eta$ and $\epsilon$ are positive constants.
    \end{lemma}
    \begin{proof}
    First, we note that the symmetric random walk generating a length-$L$ word is lazy due to the `$\tte$' generator explicitly present in the group relations.  Call $W_n$ the number of words of length $n$ that do not have any $\tte$'s.  Then, the expected number of $\tte$'s in a length-$L$ word uniformly chosen in $\mck_e$ is
    \begin{equation}
    \langle N_{\tte} \rangle = \frac{1}{|\mck_e|}\sum_{n=0}^L \binom{L}{n} W_n (L-n)
    \end{equation}
    We know that due to Theorem~\ref{thm:poly_rw}, for $n > 0$, $\frac{C}{n^{k/2}} (2d)^n \leq W_n \leq \frac{C'}{n^{k/2}} (2d)^n$ where $k$ is the growth rate degree, $d$ is the number of non-identity generators, and $C$ and $C'$ are constants.  Furthermore, $\frac{C}{L^{k/2}} (2d+1)^L \leq |\mck_e| \leq \frac{C'}{L^{k/2}} (2d+1)^L$. Thus, we can write
    \begin{equation}
    \langle N_{\tte} \rangle \gtrsim L\sum_{n=1}^L \binom{L}{n} \frac{(2d)^n}{(2d+1)^L} (x_n^{k/2} - x_n^{k/2-1})
    \end{equation}
    where $\gtrsim$ indicates greater than or equal to up to multiplicative constants, and $x_n = L/n$.  Next, we note that $f_k(x) = x^{k/2}(1-1/x)$ is convex in the interval $[1,\infty)$ for $k > 2$.  Noting that $p(n) = \binom{L}{n}\frac{(2d)^n}{(2d+1)^L}$ is the PDF of a binomial random variable with parameter $2d/(2d+1)$, we can write
    \begin{align}
    \langle N_{\tte} \rangle &\gtrsim \zeta L \sum_{n=1}^L \zeta^{-1} p(n) f_k(x_n) \gtrsim \zeta L f_k\left(\sum_{n=1}^L \zeta^{-1} p(n) x_n\right)  
    \end{align}
    where we reweighted the distribution by $\zeta = 1-\frac{1}{(2d+1)^L}$ and used Jensen's inequality.  What remains is to compute the expectation value of $x_n$ under the reweighted binomial distribution.  First define $x_n' = L/(n+1)$, and note that $f_k(\cdot)$ is increasing in the interval $[L/(L+1), \infty)$ for $L$ large enough.  Thus,
    \begin{align}
    \langle N_{\tte} \rangle \gtrsim \zeta L f_k\left(\sum_{n=1}^L \zeta^{-1} p(n) x_n'\right). 
    \end{align}
    Then, use the identity 
    \begin{equation}
    \sum_{n=1}^L \binom{L}{n} \frac{(2d)^n}{(2d+1)^L} \frac{L}{n+1} = \frac{(2d+1)L}{2d(L+1)} - \frac{L}{(2d+1)^L}
    \end{equation}
    which is $(2d+1)/(2d + \epsilon)$ for $\epsilon \to 0$ as $L \to \infty$.  Then, $\zeta f_k((2d+1)/(2d + \epsilon))$ is a constant, and thus $\langle N_{\tte} \rangle \sim L$.

    Next, suppose that $k = 2$.  Then, 
    \begin{equation}
    \langle N_{\tte} \rangle \gtrsim L\sum_{n=1}^L \binom{L}{n} \frac{(2d)^n}{(2d+1)^L} (x'_n - 1)
    \end{equation}
    which evaluates to $L(1/2d + \epsilon)$ where $\epsilon \to 0$ as $L \to \infty$.  Thus, $\langle N_{\tte} \rangle \sim L$.  Finally, suppose that $k = 1$.  Define $y_n = 1/x_n = n/L$ so that we may write
    \begin{equation}
    \langle N_{\tte} \rangle \gtrsim L\sum_{n=0}^L \binom{L}{n} \frac{(2d)^n}{(2d+1)^L} (y_n^{-1/2} - y_n^{1/2})
    \end{equation}
    Note that $g(x) = x^{-1/2} - x^{1/2}$ is convex in the interval $(0,1]$, and thus by Jensen's inequality
    \begin{equation}
    \langle N_{\tte} \rangle \gtrsim L g\left(\sum_{n=0}^L p(n) y_n\right)
    \end{equation}
    Then, we use
    \begin{equation}
    \sum_{n=0}^L \binom{L}{n} \frac{(2d)^n}{(2d+1)^L} \frac{n}{L} = \frac{1}{2d+1}
    \end{equation}
    and thus $\langle N_{\tte} \rangle \gtrsim L g((2d+1)^{-1}) \sim L$.  Having shown that $\langle N_{\tte} \rangle \sim L$ for all $k$, call $\langle N_{\tte} \rangle = (1-\delta) L$.  Applying Markov's inequality to the random variable $L - N_{\tte}$, we have
    \begin{equation}
    \mathbb{P}\left[ L - N_{\tte} \geq (1-\eta) L\right] \leq \frac{\delta}{1-\eta},
    \end{equation}
    and choosing $\eta$ small enough proves the Lemma.
    \end{proof}

    \begin{proposition}\label{prop:fragFL}
    Consider a group $G$ with $\mathrm{FL}(L) \sim L$.  If the number of length-$L$ words representing the identity is $|\mck_e|$, under group dynamics on $G$, there exists a connected component of words within $\mck_e$ of size $\sim |\mck_e|/L^{\xi}$ for $\xi > 0$.
    \end{proposition}
    \begin{proof}
    Suppose $\mathrm{FL}(L) = \alpha L$, $\alpha > 1$.  Consider a length-$L$ word $w = \tte^L$.  Divide this word into $L/\ell$ segments of length $\ell$.  In each length $\ell$ segment, we will create a word representing the identity such that the number of $\tte$'s in the word is at least $\eta \ell$.  Selecting $\eta$ small enough, there is an $\epsilon(\eta)$ such that the number of such words is $(1-\epsilon) (2d+1)^{\ell}/\ell^{k/2}$ due to the previous Lemma and Theorem~\ref{thm:poly_rw}.  For each length $\ell$ interval, all such words can be connected to one another.  This is because the number of $\tte$'s needed for this in each interval is $\mathrm{FL}(\ell) = \alpha \ell$ and since the number of $\tte$'s in the entire word is $ \geq (\eta \ell)L/\ell = \eta L$, choosing $\ell = \eta L / \alpha$ ensures the connectivity property.  Therefore, the number of words that one can construct is at least
    \begin{equation}
    W \geq \left((1-\epsilon) \frac{(2d+1)^{\ell}}{\ell^{k/2}}\right)^{L/\ell}
    \end{equation}
    and using $L/\ell = \alpha/\eta$, we find
    \begin{equation}
    W \gtrsim  \frac{(2d+1)^{L}}{L^{\alpha k/(2\eta)}} \sim \frac{|\mck_e|}{L^{k/2(\alpha/\eta - 1)}},
    \end{equation}
    thus proving the Proposition.
    \end{proof}
    This implies the following Proposition:
    \begin{proposition}
    Group dynamics for group $G$ has polynomial fragmentation iff $G$ has polynomial growth rate.
    \end{proposition}
    \begin{proof}
    If $G$ has polynomial growth rate, then the size of the identity sector is $|\mck_e|/|\mch| \sim L^{-2k}$ where $L$ is the system size.  Furthermore, we know that all nilpotent groups have $\mathrm{FL}(L) \sim L$, and thus all virtually nilpotent groups have $\mathrm{FL}(L) \sim L$ because the asymptotic scaling of the filling length is a quasi-isometry invariant (see Ref.~\cite{riley2006filling}).  Therefore, the above Proposition can be applied, showing that there exists a connected subset $\mck \subset \mck_e$ such that $|\mck|/|\mch| \geq L^{-\xi}$.  Thus, the largest sector must have size $L^{-\xi} \leq |\mck_{\mathrm{max}}|/|\mch| \leq L^{-2k}$ thus proving polynomial fragmentation.

    If $G$ has polynomial fragmentation, then there exists a Krylov sector $\mck \subset \mck_g$ for some $g$ such that $|\mck|/|\mch| \sim L^{-\zeta}$.  Then, 
    \begin{equation}
    \frac{|\mck_e|}{|\mch|} \geq \frac{|\mck_g|}{|\mch|} \geq \frac{|\mck|}{|\mch|} \sim L^{-\zeta}
    \end{equation}
    where the first inequality follows from Prop.~\ref{prop:return} in the main text.  Therefore, the return probability decays with $L$ at least as slow as inverse polynomially.  From the results in this appendix, this implies that $G$ is virtually nilpotent and thus has polynomial growth rate.
    \end{proof}
    Next, we will study the case where the dynamics is neither polynomially fragmented nor exponentially fragmented.  Such groups exist where return probabilities scale neither inverse polynomially nor exponentially; the most notable example is the Baumslag-Solitar group, where the return probability of a length $L$ random walk scales like $\sim \exp(-L^{1/3})$.  In general, call the scaling of the return probability $\exp(-f(L))$ where $f(L) \ll O(L)$.  Furthermore, call the inverse filling length function $\mathrm{FL}^{-1}(L)$ (i.e. satisfying $\mathrm{FL}^{-1} \circ \mathrm{FL}(L) = L$).  Then the following is true:
    \begin{proposition}
    Given a finitely generated and presented group $G$, suppose that with probability $1-\epsilon$, a randomly chosen length-$L$ word from $\mck_e$ has a number of $\tte$'s which is $\geq \eta L$ for some $\eta(\epsilon)$.  Then, if $G$ is amenable, then group dynamics on $G$ is not exponentially strongly fragmented. 
    \end{proposition}
    \begin{proof}
    The proof is similar to that of Prop.~\ref{prop:fragFL}.  If $G$ is amenable, then the return probability decays subexponentially with $L$, i.e. $\exp(-f(L))$.  We consider a word $\tte^L$ and split it into $L/\ell$ segments of length $\ell$.  The length of $\ell$ has to be $\ell \sim \mathrm{FL}^{-1}(\eta L)$  to access all words of length $\ell$ representing the identity.  Since this can be done on each segment of length $\ell$ the number of such words which can be formed is
    \begin{equation}
    W \gtrsim \left((1-\epsilon)\frac{(2d+1)^{\ell}}{\exp(f(\ell))}\right)^{L/\ell} \sim |\mch| \frac{(1-\epsilon)^{L/\ell}}{\exp(L f(\ell)/\ell)}.
    \end{equation}
    Since $\ell \sim \mathrm{FL}^{-1}(\eta L)$ which is asymptotically scaling with $L$, $W/|\mch|$ decays subexponentially in $L$ and thus the dynamics cannot exhibit exponential fragmentation.
    \end{proof}

    The above proofs heavily relied on the existence of an ``identity'' character $\tte$.  Suppose one constructs a Hamiltonian implementing rewriting dynamics that does not feature such a character in the local Hilbert space.  An example is the pair-flip model, which is a group model for $G = \mathbb{Z}_2 \ast \mathbb{Z}_2$.  Then, one simply needs to replace $\tte$ with a sequence of characters of length $O(1)$ which represents the identity.  So long as the expected number of such sequences is $O(L)$ for a randomly chosen length-$L$ word representing the identity element, then the same results will hold.  Since the sequence is of length $O(1)$, we expect this to be true, although we leave a proof to future work.

    \section{Gromov boundaries and quantitative Benjamini conjecture for hyperbolic groups}\label{app:proofbonanza2}
    
        In this appendix, we provide the reader with a brief introduction to boundaries of hyperbolic groups, before proving the main theorems in Sec.~\ref{sec:groups}.  This appendix will be divided into four parts.  First, we introduce Gromov boundaries of hyperbolic groups and prove some of their known properties.  Then, we use these properties to prove a quantitative version of Benjamini's conjecture for hyperbolic groups.  We then prove a quantitative version of Benjamini's conjecture for the heat kernel measure in hyperbolic groups.  We finally provide a discussion about the relationship to the amenability of the Poisson boundary, which we believe is ultimately the route towards a more general proof for all groups (and beyond). 

    	First, we will provide a discussion of some important results pertaining to hyperbolic metric spaces and their boundaries.  Understanding properties of the boundaries of these spaces is important because the bottlenecks are produced from finite system size effects, which mean that they have to originate near the boundary of the metric space.  Denote $X$ as a hyperbolic metric space (i.e. a metric space satisfying the $\delta$-thin triangle property for some $\delta > 0$) and call $\partial X$ the {\it Gromov boundary} of that space.  Excellent surveys which will contain most of the results we will use in this section are Refs.~\cite{hullhyperbolic, kapovich2002boundaries}.  We will also be using $d(a,b)$ and $|a-b|$ interchangeably to denote the word distance. The formal definition of the Gromov boundary is given below:
	\begin{definition}[Gromov boundary]
		Given a $\delta$-hyperbolic metric space $X$, define a geodesic ray to be an isometry $\gamma: [0,\infty) \to X$ such that the set of points $\gamma(0),\gamma(1),\cdots, \gamma(t)$ is a geodesic for all $t$.  We define the equivalence relation $\sim$ to mean that $\gamma_1 \sim \gamma_2$ if there exists a $K > 0$ such that $d(\gamma_1(t), \gamma_2(t)) \leq K$ for all $t > 0$.  Call $\Gamma$ the space of geodesic rays.  The {\it Gromov boundary} $\p X$ is then defined as $\partial X = \Gamma / \sim$, i.e. the set of equivalence classes of geodesic rays.
	\end{definition}
	To build intuition, we list some examples of Gromov boundaries:
	\begin{enumerate}
		\item The Gromov boundary of $\mathbb{Z}$ is $\partial \mathbb{Z} \cong \{-\infty, \infty\}$
		\item The Gromov boundary of the free group is $\partial F_n \cong \mathcal{C}$ where $\mathcal{C}$ is a Cantor set
		\item The Gromov boundary of Fuchsian groups $\Gamma$ (certain discrete subgroups of $PSL(2,\mathbb{R})$) is $\partial \Gamma \cong S^1$.  This follows from the fact that Fuchsian groups are surface groups of the hyperbolic plane $\mathbb{H}^2$ and $\partial \mathbb{H}^2 \cong S^1$.
		\item The Gromov boundary of most hyperbolic groups is homeomorphic to a Menger sponge~\cite{champetier1995proprietes, kapovich2002boundaries}.
	\end{enumerate}
	One nice property of the Gromov boundary is that it is a quasi-isometry invariant.  Roughly speaking, this means that two geodesic metric spaces $X$ and $Y$ are not `equivalent' to one another (i.e. mappable via a quasi-isometry) if their Gromov boundaries are not equivalent.  Since hyperbolic groups have a variety of different kinds of Gromov boundaries (including examples with non-integer Hausdorff dimensions), the classification of these groups is rather rich.  We should {\it not} simply expect to quasi-isometrically embed any hyperbolic group in a simple space like $\mathbb{H}^n$.  Another nice fact is that the Gromov boundary is itself a topological space, and a notion of distance can be defined -- we elaborate on this below.  
	
	One important property of the Gromov boundary is that it has finite Hausdorff dimension.  For our purposes it is more convenient to use a different measure of dimension called the Assouad dimension~\cite{assouad1979etude}:
	\begin{definition}[Assouad dimension]
		Let $X$ be a metric space with distance metric $d$.  Let $S_{\alpha, \beta}$ denote the largest set such that all points $x,y$ in the set satisfy $\alpha \leq d(x,y) \leq \beta$.  
    The Assouad dimension of $X$ is then defined by 
	\be d_a(X) = \min \{ \, t\, : \, |S_{\a,\b}| \leq K (\b/\a)^t \, \, \, \,\forall \alpha < \beta \}.\ee
    for some constant $K$.
	\end{definition}
    
	We will also need to define the Gromov product, an important notion of distance:
	\begin{definition}[Gromov product]
		Given $x,z,w \in X$, the Gromov product is
		\begin{equation}
			(x|z)_w = \frac{1}{2}\left(d(x,w) + d(z,w) - d(x,z)\right),
		\end{equation}
		which roughly measures how long geodesics remain close to one another.  When $x, w \in X$ but $\gamma \in \partial X$, the Gromov product is
		\begin{equation}
			(x|\gamma)_w = \sup\{ \liminf_{i \to \infty}(x|y_i)_w : y_i \in [\gamma]\}.
		\end{equation}
	\end{definition}
	Using the Gromov product, we can formulate an equivalent definition of a $\delta$-hyperbolic metric space which follows from the thin triangle property (see standard texts for a proof):
    \begin{lemma}[$\delta$-hyperbolic metric space]
    $X$ is a $\delta$-hyperbolic metric space iff for all $x,y,z,w \in X$,
	\begin{equation}
		(x|z)_w \geq (x|y)_w \wedge (y|z)_w - \delta.
	\end{equation}
    where $a \wedge b = \min(a,b)$.  
    \end{lemma}
    Finally, since the Gromov boundary is a topological space, we additionally would like to endow a metric on it, called the visual metric:
	\begin{definition}[Visual metric on Gromov boundary]
		If $x,y \in \partial X$ and $o \in X$, let $\epsilon > 0$: the {\it visual metric} on the Gromov boundary is given by
		\begin{equation}
			d_{o,\epsilon}(x,y) = \inf_{\{x_i\}: x_0 = x, x_n = y}\left(\sum_{i=1}^n e^{-\epsilon (x_{i-1}|x_i)_o}\right)
		\end{equation}
	\end{definition}
    The following lemma is a well-known result that we present without proof~\cite{hullhyperbolic, kapovich2002boundaries}:
    \begin{lemma}
	If $X$ is $\delta$-hyperbolic, then for $\epsilon$ small enough,
	\begin{equation}
		\frac{1}{2} e^{-\epsilon (x|y)_o} \leq d_{o,\epsilon}(x,y) \leq e^{-\epsilon(x|y)_o}.
	\end{equation}
    $\forall x,y \in \partial X$ and some $o \in X$.
    \end{lemma}
	For the rest of this subsection, we will use many ideas from the proof of the Bonk-Schramm embedding theorem, see Ref.~\cite{bonk2011embeddings}:  
    \begin{theorem}[Bonk-Schramm~\cite{bonk2011embeddings}]
    Any Gromov hyperbolic geodesic metric space with bounded growth at some scale is roughly similar to some convex subset of $\mathbb{H}^d$ for some $d$.  By roughly similar, we mean there exists a map $f:X \to Y \subset \mathbb{H}^d$ and constants $\alpha$ and $\beta$ such that
    \begin{equation}
    \left|\alpha d_X(x,y) - d_Y(f(x),f(y))\right| \leq \beta
    \end{equation}
    $\forall x,y \in X$ and $\sup_{y \in Y}d_Y(y, f(X)) \leq \beta$.  Here, $d_{X,Y}$ denotes the geodesic distance in $X$ and $Y$.
    \end{theorem}
    We may expect to compute the expansion of the Cayley graph restricted to a ball by computing the expansion of the associated convex set obtained upon embedding.  Unfortunately characterizing such convex sets in $\mathbb{H}^d$ is not straightforward as these sets can reside in lower-dimensional subspaces or exhibit fractal-like properties.  Instead, a more convenient proof involves working directly with $X$ rather than its embedding, but still relying on ideas from Bonk and Schramm's proof.  In particular, they first prove that the Gromov boundary has finite Assouad dimension.  Then they use a result due to Assouad to embed the boundary into $\mathbb{R}^d$ for some $d$.  Then they argue that the rest of the space can be embedded into a space with $\mathbb{R}^d$ as its boundary, namely $\mathbb{H}^d$.  The first idea in the proof is particularly important to us:
\begin{proposition}\label{prop:assouad}
		The Assouad dimension of the Gromov boundary of a Gromov hyperbolic geodesic metric space with bounded degree is finite.
	\end{proposition}
    For the reader's convenience, we will reproduce the proof from Ref.~\cite{bonk2011embeddings} with minor pedagogical differences. We first need the following fact:
    \begin{fact}\label{fact:shortfact}
       Given $X$ a $\delta$-hyperbolic metric space, for $z \in \partial X$ and every $t > 0$, there is an $x \in X$ such that $|x| = t$ (where $|x| \triangleq d(x,o)$ for some specified origin $o$) and $(x|z)_o - t = C(\delta)$. 
    \end{fact}
    \begin{proof}
        To show the fact, consider $w \in X$ with $|w|\geq t$ such that $(w|z) \geq t+\delta$, and call $x$ a point on $[o,w]$ satisfying $|x| = t$.  Using the simple identity $(x|w)_o \leq  d(x,o) \wedge d(w,o)$ and the hyperbolicity condition, we obtain
		\begin{equation}
			t = d(x,o)\wedge d(w,o) \geq (x|z)_o \wedge (w|z)_o - \delta
		\end{equation}
		Since $(w|z) \geq t+\delta$ we must have $(x|z)_o \wedge (w|z)_o = (x|z)_o$ therefore giving the desired result.  We will now spend a bit of time proving the following simple fact.  Suppose that $a,b \in X \cup \partial X$ and $x,y \in X$.  Suppose that $(a|x)_o \geq |x| - \epsilon$ and $(b|y)_o \geq |y| - \epsilon$ for some $\epsilon \geq 0$.  Then,
		\begin{equation}
			d(x,y) = |x| + |y| - 2\left( (a|b)_o \wedge |x| \wedge |y|\right) + C(\delta, \epsilon).
		\end{equation}
		The proof follows from using the hyperbolicity condition many times (we will drop the subscript $o$ in what follows):
		\begin{align}
			(a|b) &\geq (a|x) \wedge (b|x) - \delta  \nonumber \\
			&\geq (a|x) \wedge (b|y) \wedge (x|y) - 2\delta\nonumber \\
			&\geq |x| \wedge |y| \wedge (x|y) - 2\delta - \epsilon.
		\end{align}
		Using the Gromov product identity $(x|y) \leq |x| \wedge |y|$, we have $(a|b) \geq (x|y) - 2\delta - \epsilon$ and therefore $(a|b) \wedge |x| \wedge |y| \geq (x|y) - 2\delta - \epsilon$.  Next, we can apply the hyperbolicity condition multiple times for $(x|y)$ instead:
		\begin{align}
			(x|y) &\geq (x|a) \wedge (y|a) - \delta  \nonumber \\
			&\geq (a|x) \wedge (y|b) \wedge (a|b) - 2\delta\nonumber \\
			&\geq |x| \wedge |y| \wedge (a|b) - 2\delta - \epsilon.
		\end{align}
		Combining these two results gives:
		\begin{equation}
			|x| \wedge |y| \wedge (a|b) - C(\delta, \epsilon) \leq (x|y) \leq |x| \wedge |y| \wedge (a|b) + C(\delta, \epsilon)
		\end{equation}
		proving the fact.
		\end{proof}
        We now prove Prop.~\ref{prop:assouad}:
        \begin{proof}
        {\it of Prop.~\ref{prop:assouad}}. 
		Now, consider points $z_1, z_2, \cdots, z_n \in \partial X$ satisfying $\alpha \leq d_{o,\nu}(z_i, z_j) \leq \beta$.  Rescaling the metric on $X$ so that we can set $\nu = 1$ (see Ref.~\cite{bonk2011embeddings} for further discussion of this point; without the rescaling we would have some additional unimportant dependence on $\nu$), this expression is equivalent to
		\begin{equation}
			-\log \beta - C(\delta) \leq (z_i|z_j) \leq  -\log \alpha + C(\delta)
		\end{equation}
		for $0<\alpha \leq \beta \leq 1$.  We also choose points $x_1, \cdots, x_n \in X$ and $y_1, \cdots, y_n \in X$ which satisfy $|x_j| = -\log \beta, |-\log \beta - (x_j|z_j)| \leq C(\delta)$ and $|y_j| = -\log \alpha, |-\log \alpha - (y_j|z_j)| \leq C(\delta)$ as per the claim made towards the beginning of the proof.  Next, we use Fact~\ref{fact:shortfact} (with $\epsilon$ set to a function of $\delta$) to show that (abusing notation with use of $C(\cdot)$ to indicate different constants)
		\begin{align}
			d(x_i, x_j) &= |x_i| + |x_j| -2((z_i|z_j)\wedge |x_i| \wedge |x_j|) + C(\delta) \nonumber \\
			& \leq -2 \log \beta - 2 (-\log \beta - C(\delta)) + C(\delta) \nonumber \\
			& \leq C(\delta).
		\end{align}
		Similarly, 
		\begin{align}
			d(x_i, y_j) &= |x_i| + |y_j| -2((z_i|z_j)\wedge |x_i| \wedge |y_j|) + C(\delta) \nonumber \\
			&= -\log \beta - \log \alpha - 2 (-\log \beta) + C(\delta) \nonumber \\
			&= \log \frac{\beta}{\alpha} + C(\delta).
		\end{align}
		and, 
		\begin{align}
			d(y_i, y_j) &= |y_i| + |y_j| -2((z_i|z_j)\wedge |y_i| \wedge |y_j|) + C(\delta) \nonumber \\
			&\leq -2\log \alpha -2 (-\log \beta - C(\delta) \wedge -\log \alpha) + C(\delta) \nonumber \\
			&= 2 \log \frac{\beta}{\alpha} + C(\delta).
		\end{align}
		This implies that a ball of radius $R_{\ast} = \log(\beta/\alpha) + C(\delta)$ centered on one of the $x_i$ will contain all of the $y_i$.  
		
		The rest proceeds using standard topological arguments.  If $X$ is a bounded degree graph, then it satisfies the property that there exist constants $R>r>0$ and $N > 0$ such that every open ball of radius $R$ can be covered with $N$ open balls of radius $r$ in $X$.  Furthermore, a ball of radius $2R-r$, denoted $B(2R-r)$ can be covered by $N^2$ open balls of radius $r$ \footnote{The argument is standard: consider $B(R)$ centered at the same point, which is covered by $N$ balls $B_1(r), B_2(r), \cdots, B_N(r)$.  Consider the union $\bigcup_i B_i(R)$; this covers $B(2R-r)$ following from $X$ being a geodesic metric space.  Finally, each one of these balls is covered by $N$ balls of radius $r$ and so $B(2R-r)$ is covered by $N^2$ balls of radius $r$.}.  The generalization to this is $B(nR - (n-1)r)$ being covered by $N^n$ balls of radius $r$.  Call $n_{\ast}$ the smallest integer for which $n_{\ast}(R-r) \geq R_{\ast}$.  Then, the points $\{y_i\}$ are covered by $N^{n_{\ast}}$ radius-$r$ balls.  If $r \ll \log \frac{\beta}{\alpha}$, then balls of radius $r$ cover at most one of the $y_i$'s.  Therefore,
		\begin{equation}
			S_{\alpha, \beta} \leq N^{n_{\ast}} \leq N^{1 + R_{\ast}/(R-r)} \leq C(R,r,N,\delta) \left(\frac{\beta}{\alpha}\right)^{C(R,r,N,\delta)},
		\end{equation}
		which implies that the Assouad dimension of the Gromov boundary is finite.
	\end{proof}
	The finiteness of the Assouad dimension means that the Gromov boundary exhibits an isoperimentric inequality of the form $|\partial A|/|A| \sim |A|^{r}$ for $A \subset \partial X$.  To show a similar bound for the space $X$, we will need a `bulk-boundary' dictionary, which maps points on the Gromov boundary to points in $X$ at some fixed distance to the specified origin $o$.  This concept is also introduced an heavily utilized in Ref.~\cite{bonk2011embeddings}. 
	\begin{theorem}\label{thm:quant_hyperbolic}
		Consider a $\delta$-hyperbolic visual\footnote{A metric space $X$ is visual with respect to $o$ if every point $x \in X$ lies on a roughly-geodesic ray originating from $o$.  Cayley graphs of hyperbolic groups with the word distance are visual.} metric space $X$ where $\partial X$ has non-zero Assouad dimension.  The expansion of a ball $B_o(L)$ centered at $o$ is $\leq e^{-\xi L}$ for some $\xi > 0$.
	\end{theorem}
	\begin{proof}
		The Gromov boundary $\partial X$ of hyperbolic metric space $X$ has $\text{diam}(\partial X) = D$ for $D = O(1)$.  Define the ``cone'' space 
		\begin{equation}
			\text{Con}(\partial X, d) = \partial X \times (0, D].
		\end{equation}
		with $d$ the word metric on $X$.  We can define the metric on this space $\rho:\text{Con}(\partial X,d) \times \text{Con}(\partial X,d) \to \mathbb{R}^+$ satisfying
		\begin{equation}
			\rho((z,h), (z',h')) = 2 \log \left(\frac{d_{o,\epsilon}(z,z') + h \wedge h'}{\sqrt{hh'}}\right).
		\end{equation}
		where $d_{o,\epsilon}$ is the visual metric on $\partial X$ with respect to some basepoint $o$.  Bonk and Schramm also show (see Theorem 7.2 of Ref.~\cite{bonk2011embeddings}) that this space is $\delta$-hyperbolic and visual. Additionally, if $X$ is visual, they show that $X \sim \text{Con}(\partial X, d)$ \footnote{The assumption of the metric space being visual is not needed for the main Bonk-Schramm embedding theorem.}.  In this context $\sim$ means that there exists a quasi-isometry $f: \text{Con}(\partial X, d) \to X$, i.e. satisfying the property that there exist constants $\lambda, K$ such that for all $\zeta = (z,h)$ and $\zeta' = (z',h')$
		\begin{equation}
			\rho(\zeta, \zeta') - K \leq \lambda^{-1}d_X(f(\zeta) - f(\zeta')) \leq \rho(\zeta, \zeta') + K.
		\end{equation}
		 The proof of this property is the proof of Theorem 8.2 in Ref.~\cite{bonk2011embeddings}.  For our purposes, we need the explicit map $f$.  Define $\gamma_z$ to be a geodesic ray with $\gamma_z(0) = o$ which is in the equivalence class $z \in \partial X$.  Then,
		\begin{equation}
			f(z,h) = \gamma_z(\epsilon^{-1} \log(D/h))
		\end{equation}
		obeys the quasi-isometry property.  Suppose we consider a ball on $X$ with radius $L$, denoted as $B(L)$.  The range of $h$ which will contain all points in $B(L)$ is 
		\begin{equation}
			0 \leq \epsilon^{-1} \log(D/h) \leq L
		\end{equation}
		or $h \in [D e^{-\epsilon L}, D]$, which is discretized into $L$ intervals.  We also show the following property of the quasi-isometry.  Consider the points in $X$ corresponding to the image of points in $\mathrm{Con}(X,d)$ with $h = D e^{-\epsilon n}$ for integer $n \in [0,L]$.  Since the distance to neighbors in $X$ is $O(1)$, the corresponding points in $\partial X$ must satisfy $\rho((z,h), (z',h)) = 2 \log (D^{-1} e^{\epsilon n} d_{o,\epsilon}(z,z') + 1) = O(1)$, which is not possible to satisfy if $d_{o,\epsilon}(z,z') \ll D e^{-\epsilon n}$.  Then, at the surface of a ball of radius $n$ in $X$, associated points in $\partial X$ are at least a distance $\sim O(D e^{-\epsilon n})$ from each other (up to multiplicative constants determined by parameters of the quasi-isometry).

    Suppose that we pick $\beta < D$.  The points in $S_{\a,\b}$ can by definition be enclosed in a ball of diameter $\beta$.  Add additional points to this set so that $\alpha \leq d_{o,\epsilon}(z_i,z_j) \leq \beta + r \alpha$ for constant $r$ (related to the parameters of the quasi-isometry $f$): this set, denoted by $S_{\alpha, \beta+k\alpha}$, is enclosed in a ball of diameter $\beta + r \alpha$.  Choose $r$ such that $r \alpha/\beta \ll 1$.  Since $\partial X$ has finite Assouad dimension denoted $d_a > 0$, there exists an $S_{\alpha, \beta}$ satisfying
		\begin{align}
			|S_{\alpha, \beta+r\alpha}| - |S_{\alpha, \beta}| &= C \left(\frac{\beta}{\alpha}\right)^{d_a}\left(\left(1 + \frac{r \alpha}{\beta}\right)^{d_a} - 1\right) \nonumber \\
            &\leq C\left(\frac{\beta}{\alpha}\right)^{d_a}\left(\left(1 + \frac{r \alpha}{\beta}\right)^{\ceil{d_a}} - 1\right) \nonumber \\
			&\leq  C' 2^{d_a} r\left(\frac{\beta}{\alpha}\right)^{d_a - 1}
		\end{align}
		where in the second inequality we used $(1+x)^{\ceil k} - 1 \leq 2^{\ceil{k}} x \leq 2^{k+1} x$ if $x < 1$ and $C' = 2C$.  This provides an upper bound on the number of points adjacent to the boundary of the ball of diameter $\beta$.  Define the set $A_n \subset X$ to be
		\begin{equation}
			A_n = \{x: |x|_o = n, \exists z \in S_{\alpha, \beta}, |x-f(z,D e^{-\epsilon n})| \leq K\}.
		\end{equation}
        for some constant $K$, which is the image of the diameter $\beta$ ball in $\partial X$ at scale $h = D e^{-\epsilon n}$.  Then define 
        \begin{equation}
        A = \bigcup_{n=0}^L A_n.
        \end{equation}
	    To compute $|A_n|$, we define $Z_n$ to be a set of maximal size which can be expressed as $\{\hat{z}_i \in S_{\alpha, \beta}: \forall i \neq j, \, K' e^{-\epsilon n} \leq d_{o,\epsilon}(\hat{z}_i, \hat{z}_j) \leq \beta\}$, where $K'$ is a constant.

       Due to the property of the quasi-isometry we proved earlier, for large enough $K'$ the size of $|A_n|$ is the size of $|Z_n|$, up to an $O(1)$ constant depending on $K, K'$, and the degree of the Cayley graph.  Because of the finite Assouad dimension of $\partial X$, we can choose the set $Z_n$ such that $|Z_n| = C(\beta/\alpha_n)^{d_a} = C (\beta e^{\epsilon n}/D)^{d_a}$.  Therefore,
        \begin{equation}
        |A| \sim \sum_{n=0}^L |Z_n| \sim C \sum_{n=0}^L \left(\frac{\beta e^{\epsilon n}}{D}\right)^{d_a} \sim C' e^{\epsilon d_a L}.
        \end{equation}
        Next, we compute $|\partial A|$.  We construct the set $\partial S = S_{\alpha, \beta+r\alpha} \setminus S_{\alpha, \beta}$.  We then define
        \begin{equation}
			\partial A_n = \{x: |x|_o = n, \exists z \in \partial S, |x-f(z,D e^{-\epsilon n})| \leq K\}.
		\end{equation}
        and similarly define $\partial Z_n$ to be a set of maximal size which is expressed as $\{\hat{z}_i \in \partial S: \forall i \neq j, K' e^{-\epsilon n} \leq d_{o,\epsilon}(\hat{z}_i, \hat{z}_j) \leq \beta+r\alpha\}$.  We then use the bound on $|S_{\alpha, \beta+r\alpha}| - |S_{\alpha, \beta}|$ to show that $|\partial Z_n| \leq C' 2^{d_a} r (\beta/\alpha_n)^{d_a-1} = C' 2^{d_a} r (\beta e^{\epsilon n}/D)^{d_a-1}$.  Therefore,
        \begin{equation}
        |\partial A| \sim \sum_{n=0}^L |\partial Z_n| \lesssim C' 2^{d_a} r \sum_{n=0}^L \left(\frac{\beta e^{\epsilon n}}{D}\right)^{d_a-1} \lesssim C'' e^{\epsilon (d_a-1) L}.
        \end{equation}
        Therefore, there exists an $A \subset X$ such that $|\partial A|/|A| \lesssim e^{-\epsilon L}$, thus proving the upper bound for the expansion.
        
	\end{proof}
    The proof idea crucially only relies on the Assouad dimension of the boundary being finite.  In addition, an important assumption we made for the proof was that the Assouad dimension is strictly positive.  Having zero Assouad dimension would mean that the hyperbolic group has trivial Gromov boundary (i.e. it is isomorphic to a finite set of points) which implies that such a group would be amenable. 

    \subsection{Hyperbolic groups with heat kernel measure}

    Having proved that the unweighted Cayley graph has expansion exponentially small in $L$, we would like to prove a similar statement for the $\nu$-weighted Cayley graph of hyperbolic groups where $\nu$ is the heat kernel measure.  Naively, one can use the same sets $R$ and $\partial R$ constructed in the proof of Theorem~\ref{thm:quant_hyperbolic}, but unfortunately we cannot guarantee that $\nu(\partial R)/\nu(R)$ scales in a similar way as $|\partial R|/|R|$ without requiring more information about $\nu$.  Instead, we resort to a slightly different reasoning, which uses a probabilistic argument along with certain properties of random walks on hyperbolic groups.  In particular, we need the following result, which was proved in~\cite{aoun2022random} and the references therein (see also~\cite{maher2018random}):
    \begin{theorem}\label{thm:random_walk_conc}
    Consider a symmetric (possibly lazy) random walk on hyperbolic group $G$ with step distribution $\mu$.  Construct sample path $w_n = \ttg_1 \ttg_2 \cdots \ttg_n$ where $\ttg_i \sim \mu$.  Suppose $\mu$ is a non-elementary probability distribution on $G$ with bounded support.  Then, for each sample path $w_n$, the limit
    \begin{equation}
    \lim_{n \to \infty} \frac{|w_n|}{n} = v > 0
    \end{equation}
    exists almost surely.  Furthermore, there exists a constant $K$ such that
    \begin{equation}
    \mathbb{P}[||w_n| - v n| \geq n t] \leq 2\,\exp\left(-\frac{n t^2}{K}\right).
    \end{equation}
    \end{theorem}
    This theorem shows that the length-$n$ random walk measure concentrates on the surface of a ball of radius $v n$, with Hoeffding-like guarantees on the tails.  Additionally, we will be making use of other properties of random walks on groups, in particular relating to their entropy:
    \begin{definition}[Entropy]
    The entropy of the step distribution $\mu$ of a random walk on $G$ is $H(\mu) = -\sum_{g \in G} \mu(g) \log \mu(g)$.  Denoting $\nu_n = \delta_o \ast \mu^{\ast n}$ where $\ast$ denotes convolution and $\delta_o$ is a point mass distribution at $o$, the entropy density of a random walk on $G$ is
    \begin{equation}
    s = - \lim_{n \to \infty} \frac{1}{n} \sum_{g \in G} \nu_n(g) \log \nu_n(g).
    \end{equation}
    \end{definition}
    The following result holds for Cayley graphs of non-amenable groups:
    \begin{theorem}
    A symmetric random walk (which can be lazy) on a non-amenable group has $s > 0$.
    \end{theorem}
    This follows from the inequality $s \geq -2 \log \rho$, where the return probability is $p_{2n}(e,e) \sim \rho^{2n}$, as well as Kesten's criterion.  We will not provide a proof of the inequality as it can be found in many standard texts.  Next, we prove some simple technical Lemmas before attempting to prove our main result:

    \begin{lemma}\label{lem:technical1}
    Consider the Cayley graph $X$ of hyperbolic group $G$ with Gromov boundary $\partial X$ of diameter $D$ equipped with the map $f:\mathrm{Con}(\partial X, d) \to X$.  Cover $\partial X$ with open balls of diameter $R$ denoted $\mcb_z(R)$, where $z \in \partial X$ is the center of the ball.  Call $C_z = f(\mcb_z(R) \times (0,D])$ the image of the ball under the map $f$.  Consider a symmetric random walk on $G$ with finite entropy density $h$. Defining 
    \begin{equation}
    s_{\max} = \frac{1}{n} \max_{z} \sum_{x \in C_z} -\nu_L(x) \log \nu_L(x),
    \end{equation}
    if $s_{\max} = s/K$ for constant $K$, then $R \geq D \cdot \exp\left(-(\epsilon v - s/(d_a K))L\right)$, where $\epsilon$ is the parameter of the visual metric, $d_a$ is the Assouad dimension of $\partial X$, and  $v$ is defined in Theorem~\ref{thm:random_walk_conc}.
    \end{lemma}
    \begin{proof}
     Due to Theorem~\ref{thm:random_walk_conc}, the probability of a random walk $w_L$ being distance $\gg \sqrt{L} \log L$ from the surface of a ball of radius $v L$ scales inverse polynomially in $L$.  Therefore, the contribution to the entropy density from any of these points is $\leq -(\log \nu_{\min})/L^{C+1}$ where $\nu_{\min} \geq 1/(2d+1)^L$ is the smallest probability of any $x \in X$ and $d$ is the number of generators of $G$.  Thus, the contribution to the entropy density from these points is $\leq (\log (2d+1))/L^{C}$.  Since the maximum contribution to the entropy density from one of the regions is $s/K$, this places a lower bound on the size of the region which is $\geq \exp(s L/K)$ nodes.  At least $\geq \frac{\exp(s L/K)}{\sqrt{L} \log L}$ of these nodes reside on the surface of a ball of some radius in the interval $[v L - O(\sqrt{L} \log L), v L + O(\sqrt{L} \log L)]$.  We will simply call this value $v L$ because the subleading contributions do not affect the analysis.

    Next, from the explicit form of the quasi-isometry $f$, we know that points $(z, D e^{-\epsilon v L}) \in \mathrm{Con}(\partial X, d)$ correspond to being close to the surface of a radius-$v L$ ball in $X$.  Furthermore, each point on $X$ near this ball corresponds to a radius $\alpha \sim D e^{-\epsilon v L}$ ball in $\partial X$.  If the Assouad dimension of $\partial X$ is $d_a > 0$, then $(R/\alpha)^{d_a} \gtrsim |S_{\alpha,R}| \geq \exp(s L/K)$, which means that $R \gtrsim D \cdot \exp\left(-(\epsilon v - s/(d_a K))L\right)$.
    \end{proof}
    \begin{lemma}\label{lem:technical2}
    Suppose a subset $P \subset X$ of Cayley graph $X$ contributes $h/K$ to the entropy density with respect to the $L$-step heat kernel measure $\nu_L$.  Then, there exists constants $a < 1$ and $a' < 1$ with $a' > a$ such that
    \begin{equation}
    \frac{s}{K \log(1/a)} \leq \nu_L(P) \leq \frac{s}{K \log(1/a')}
    \end{equation}
    \end{lemma}
    \begin{proof}
    We suppress the subscript $L$ in $\nu_L$ for convenience.  Define $\nu_{\min} = \min_x \nu(x)$ and $\nu_{\max} = \max_x \nu(x)$.  We note the simple inequality
    \begin{equation}
    -\frac{1}{L}\sum_{x \in P} \nu(x) \log \nu_{\max} \leq \frac{s}{K} \leq -\frac{1}{L}\sum_{x \in P} \nu(x) \log \nu_{\min}.
    \end{equation}
    Note that $\nu(x) > 0$ for all $x \in X$.  It is immediate that $\nu_{\min} = (1/(2d+1))^L$ where $d$ is the number of generators, and $\nu_{\max} = \rho^L$ since the identity sector is the largest sector along with Kesten's criterion.  Therefore, writing $\sum_{x \in P} \nu(x) = \nu(P)$,
    \begin{equation}
    \nu(P) \log(1/\rho) \leq \frac{s}{K} \leq \nu(P) \log(2d+1),
    \end{equation}
    and thus $s/(K \log(2d+1)) \leq \nu(P) \leq s/(K \log(1/\rho))$.
    \end{proof}
    
    We now sketch the idea behind the proof of the main result of this subsection.  The main goal is to construct a set $R$ with not too large of measure such that the $\nu$-weighted expansion of $R$ is exponentially small in $L$.  The proof strategy first constructs a region $R'$ which has large probability mass.  Then, a number of ``shells'' $S_m$ are added to $R'$ such that $S_m$ ``wraps around'' $R' \cup S_1 \cup S_2 \cup \cdots \cup S_{m-1}$.  By a probabilistic argument, we show that there is an $m$ where $\nu(S_m)$ is small.  We then call $R$ the space enclosed by the separators $S_m$ and $G$, where $G$ is the boundary of a ball of a certain diameter.  We argue that $\nu(G)$ is small, and $\nu(R)$ is large, therefore showing that the expansion of $\nu(R)$ is small.
    \begin{theorem}\label{thm:quant_hyperbolic_rw}
    Suppose $X$ is the Cayley graph of a hyperbolic group $G$, weighted by the $L$-step heat kernel measure $\nu_L$.  Then, the weighted expansion of $X$ satisfies $\Phi(X) \leq \exp(-\xi L)$ for some $\xi > 0$.
    \end{theorem}
    \begin{proof}
    We will be denoting subsets of boundaries with calligraphic letters, and subsets of hyperbolic space with latin letters.  First, consider the Gromov boundary $\partial X$ of the hyperbolic metric space $X$ corresponding to the Cayley graph of $G$.  Suppose that $\partial X$ has diameter $D$: then, cover $\partial X$ by open balls of diameter $\beta < D$.  For each such ball $\mathcal{B}_z(\beta/2)$ with $z$ labeling the center of the ball, define the cone $C_z$ to be the image of $\mathcal{B}_z(\beta/2) \times (0,D]$ under the quasi-isometry $f:\text{Con}(\partial X, d) \to X$.  Then,
    \begin{equation}
    \bigcup_{z} \, C_z = X.
    \end{equation}
    Call $s(P)$ the contribution to the entropy density due to set $P$, so that $s(X) = s > 0$.  We define $z_{m} = \mathrm{argmax}_{z} s(C_{\mathcal{B}_z(\beta/2)}) $.  Choose $\beta$ so that $s(C_{z_m}) = s/K$ for some large constant $K$.  Construct $\mathcal{B}_{z_m}(3\beta/2)$, which is a ball of diameter $3\beta$ centered around $z_m$.  The number of balls of diameter $\beta$ that fit in this ball is $\lesssim 3^{d_a}$ with $d_a$ the Assouad dimension of $\partial X$.  Call $C'_{z_m}$ the image of $\mathcal{B}_{z_m}(3\beta/2) \times (0,D]$ under $f$.  Therefore, 
    \begin{equation}\label{eq:bound1}
    \nu(C'_{z_m} \setminus C_{z_m}) \lesssim \frac{3^{d_a}}{\log(1/\rho)} \frac{s}{K} \leq 3^{d_a} \frac{\log(2d+1)}{\log(1/\rho)} \nu(C_{z_m})
    \end{equation}
    where we have used Lemma~\ref{lem:technical2}.  For simplicity, call $S = C'_{z_m} \setminus C_{z_m}$.  We now will subdivide $S$ into many thin ``shells''.  To construct these shells, subdivide $\mathcal{B}_{z_m}(3\beta/2) \setminus \mathcal{B}_{z_m}(\beta/2)$ into concentric annuli of thickness $r e^{-\epsilon n}$ for constant $r$ and integer $n$, with $\epsilon$ is the parameter of the visual metric.  Define $\mca_{n}(\Delta)$ to be the annulus with inner diameter $\Delta$, so that
    \begin{equation}
    \mathcal{B}_{z_m}(3\beta/2) \setminus \mathcal{B}_{z_m}(\beta/2) = \bigcup_{m = 0}^{2\beta e^{\epsilon n}/r} \mca_{n}(\beta + m r e^{-\epsilon n}).
    \end{equation}
    Call $Q_n(R) = f(\mca_{n}(R), D e^{-\epsilon n})$ and $P_n = f(\mathcal{B}_{z_m}(3\beta/2), D e^{-\epsilon n})$.  We then define the ``shells''
    \begin{equation}
    S_m = \bigcup_{n=0}^L Q_n(\beta + m r e^{-\epsilon n}) \cap P_n
    \end{equation}
    which satisfy $S = \bigcup_{m = 0}^{2\beta e^{\epsilon L}/r} S_m$.  To provide some intuition for these shells, note that for large enough $m$ and small $n$, the argument of $Q_n(\beta + m r e^{-\epsilon n})$ can become larger than $3 \beta$.  Thus, the intersection with $P_n$ implies that no points in $Q_n(\beta + m r e^{-\epsilon n})$ are included in $S_m$ for $n$ small enough.  Thus, one should think of these shells as having two open ends whose radii increase as $m$ increases.  Next, we know from Eqn.~\ref{eq:bound1} that
    \begin{equation}
    \sum_{m=0}^{2\beta e^{\epsilon L}/r} \nu(S_m) \lesssim 3^{d_a} \frac{\log(2d+1)}{\log(1/\rho)} \nu(C_{z_m}),
    \end{equation}
    and therefore treating $m$ as a random variable uniformly sampled from $0$ to $2\beta e^{\epsilon L}/r$, we have from Markov's inequality
    \begin{equation}
    \mathbb{P}\left[\nu(S_m) \geq \frac{F}{\beta e^{(\epsilon - \eta)L}} \nu(C_{z_m})\right] \leq \frac{1}{e^{\eta L}}
    \end{equation}
    where all constants $d_a, \rho, d, r$ have been absorbed into a single constant $F$.  This means that there exists an $m^{\ast} \in [0, 2 \beta e^{(\epsilon - \eta) L}/r]$ such that $\nu(S_{m^\ast}) \leq \frac{F}{\beta e^{(\epsilon - \eta)L}} \nu(C_{z_m})$.  As mentioned previously, $S_{m^{\ast}}$ is a shell with two open ends; we now want to compute the radial location of one of the open ends.  For this, we find the value of $n^{\ast}$ for which $\beta + m^{\ast} r e^{-\epsilon n^{\ast}} = 3\beta$.  Since $m^{\ast} \leq 2 \beta e^{(\epsilon - \eta) L}/r$, we have $2 \beta e^{(\epsilon - \eta) L - \epsilon n^{\ast}} \geq 2 \beta$, or $n^{\ast} \leq L(1 - \eta/\epsilon)$.  Then, define $G_{n^{\ast}} = \{x: |x| = n^{\ast}, x \in C'_{z_m}\}$.  The region $S_{m^{\ast}} \cup G_{n^{\ast}}$ naturally separates $X \cap B(L)$ into two pieces, the smaller of which we call $R$.  In particular, $R$ has the property that 
    \begin{equation}
    C_{z_m} \cap (B(L) \setminus B(n^{\ast})) \subset R
    \end{equation}
    Note that $C_{z_m} \cap B(n^{\ast})$ has measure
    \begin{equation}
    \nu(C_{z_m} \cap B(n^{\ast})) \leq \mathbb{P}\left[|w_L| \leq L\left(1 - \frac{\eta}{\epsilon}\right)\right]
    \end{equation}
    where the probability is taken over length-$L$ random walks on the Cayley graph with step distribution $\mu$.  Since $\mu$ has bounded support, we can use Theorem~\ref{thm:random_walk_conc} after choosing $\eta$ so that $1 - \frac{\eta}{\epsilon} = (1-\zeta) v$ and find
    \begin{equation}
    \nu(C_{z_m} \cap B(n^{\ast})) \leq \exp(-\zeta' L)
    \end{equation}
    for constant $\zeta' \propto \zeta^2$.  Then, we use
    \begin{align}
    \nu(\partial R) &\leq \nu(S_{m^\ast}) + \nu(G_{n^{\ast}}) \nonumber \\
    &\leq \frac{F}{\beta e^{(\epsilon - \eta)L}} \nu(C_{z_m})+ \nu(C_{z_m} \cap B(n^{\ast})) \nonumber \\
    &\leq F \beta^{-1} e^{-\epsilon (1-\zeta) v L} \nu(C_{z_m})+ e^{-\zeta' L} \nonumber \\ 
    &\leq F' e^{(\epsilon \zeta v - s/(d_a K)) L} \nu(C_{z_m})+ e^{-\zeta' L} \nonumber \\ &\leq (1+F') e^{-((s/(d_a K)-\epsilon \zeta v) \wedge \zeta') L} \nonumber  
    \end{align}
    where in the second to last line we use Lemma~\ref{lem:technical1}.  Next, since $s(C_{z_m}) = s/K$, we know that $\nu(C_{z_m}) \geq s/(K \log(2d+1))$ but that $\nu(C_{z_m}) \leq s/(K \log(1/\rho)) = 1/100$ for sufficiently large $K$.  Therefore,
    \begin{align}
    \nu(R) &\geq \nu(C_{z_m} \cap (B(L) \setminus B(n^{\ast}))) \nonumber \\
    &= \nu(C_{z_m}) - \nu(C_{z_m} \cap B(n^{\ast})) \nonumber \\
    &\geq \frac{\log(1/\rho)}{100 \log(2d+1)} - e^{-\zeta' L} \geq \frac{\log(1/\rho)}{101 \log(2d+1)} 
    \end{align}
    for sufficiently large $L$.  However, $\nu(R) \leq \nu(C'_{z_m} \cap (B(L) \setminus B(n^{\ast})))$, which is upper bounded by using Eqn.~\ref{eq:bound1}:
    \begin{align}
    \nu(R) &\leq \nu(C'_{z_m} \cap (B(L) \setminus B(n^{\ast}))) \nonumber \\
    &\leq \nu(C'_{z_m}) \nonumber \\
    &\leq \nu(C_{z_m}) + \frac{3^{d_a}}{\log(1/\rho)} \frac{s}{K} \nonumber \\
    &\leq \frac{1}{100} + \frac{3^{d_a}}{\log(1/\rho)} \frac{s}{K}
    \end{align}
    and can be chosen to be less than $1/2$ for sufficiently large $K$.  Therefore, $\nu(\partial R)/\nu(R) \lesssim e^{-((s/(d_a K)-\epsilon \zeta v) \wedge \zeta') L}$, and with $\zeta$ small enough such that $s/(d_a K)-\epsilon \zeta v > 0$, this proves the theorem.
    \end{proof}

    We note that this result also extends to certain generalizations of hyperbolic groups, such as weakly hyperbolic groups.  This is because Theorem~\ref{thm:random_walk_conc} holds for more general hyperbolic metric spaces than Cayley graphs of word hyperbolic groups.
    
    \subsection{Connection to amenability of Poisson boundary}
    The proof above is rather general, and crucially relied on the Assouad dimension of the boundary being finite in order to construct the sets $S_m$.  Therefore, Benjamini's conjecture could be reformulated as a conjecture suggesting that boundaries of infinite bounded degree graphs are amenable.  If they are, then one could construct subsets of the boundary which have vanishing expansion, which can be mapped to a subset of vertices of the graph which also must have vanishing conductance.  In particular, the analysis of Fraczyk and van Limbeek~\cite{fraczyk2019heat} essentially reduces to showing that the boundary of the measured space corresponding to a random walk on group $G$ is amenable, though they do not provide a quantitative estimate of the conductance.  In this subsection, we summarize their results and provide a more heuristic but direct argument that the expansion decays exponentially in $L$ for hyperbolic groups.
    
    To proceed we need to know some properties of random walks on Cayley graphs.  The first result pertains to the asymptotic behavior of such a walk at infinity.   More precisely, we define the {\it Poisson boundary} of a random walk
    \begin{definition}[Poisson boundary]
    Consider a random walk on $G$ with step distribution $\mu$. Assume that for any initial distribution $m$, the support of the measure $m \ast \mu^{\ast n}$ is $G$ in the limit $n \to \infty$, where $\ast$ denotes convolution.  For two sample paths $w = g_0 g_1 g_2\cdots$ and $w' = g_0' g_1' g_2'\cdots$, define the equivalence relation $w \sim w'$ if there exists $k$ and $k'$ such that $w_{m+k} = w_{m+k'}'$ for all $m \geq 0$.  Call $(\Gamma, \nu)$ the measured space of samples paths corresponding to random walks on $G$.  The Poisson boundary of $(G, \mu)$ is the measured space $(\Gamma, \nu)/\sim$.  We denote the Poisson boundary by $(P, \tau)$, with $\tau$ called the hitting measure.
    \end{definition}
    For hyperbolic groups, there are a number of known results that characterize Poisson boundaries.  We will rely on a result due to Kaimanovich~\cite{kaimanovich1994poisson, kaimanovich2000poisson}, who requires an assumption that the step distribution of the random walk has finite entropy.  This is true for our considerations, since $\mu$ has bounded support.
    \begin{theorem}[Kaimonovich]
    Consider a random walk on $G$ with step distribution $\mu$.  If $H(\mu)$ is finite and $\sum_{g \in G} \mu(g) \log |g|$ is finite, then the Poisson boundary of $(G,\mu)$ is $(\partial G, \tau)$, where $\partial G$ is the Gromov boundary of $G$ and $\tau$ is the hitting measure on the Gromov boundary and is non-atomic. 
    \end{theorem}
    Therefore, properties of the Poisson boundary can be inherited from properties of the Gromov boundary and the hitting measure.  Now, we will follow the analysis of Fracyzk and van Limbeek, who first proved the heat kernel version of Benjamini's conjecture for Cayley graphs of groups before proving the more general case of bounded-degree graphs~\cite{fraczyk2019heat}.  Call $(P,\tau)$ the Poisson boundary of $(\Gamma, \mu)$.  Construct a sequence of subsets $\{S_n\}$ where $S_n \in P$.  Then for a fixed $n$ and $p \in P$, they construct the set
    \begin{equation}
    A_n(p) = \{g \in G: g p \in S_n\}
    \end{equation}
    Under a $k$-step random walk measure for $k \to \infty$:
    \begin{equation}
    \lim_{k \to \infty} \mu^k(A_n(p)) = \tau(S_n).
    \end{equation}
    Thus, choosing $S_n \subset P$ and $\partial S_n \subset P$, we have that there exist sets $A_n(p) \subset G$ and $\partial A_n(p) \subset G$ such that 
    \begin{equation}
    \frac{\mu^k(\partial A_n(p))}{\mu^k(A_n(p))} \xrightarrow{k \to \infty} \frac{\tau(\partial S_n)}{\tau(S_n)}.
    \end{equation}
    Since $P = \partial G$ due to the result of Kaimonovich the expansion of $A_n(p)$ weighted by a $k$-step random walk converges to the expansion of a subset $S_n \subset \partial G$ on the Gromov boundary weighted by $\tau$.  Next, we use a property of $\tau$ that it has finite Haussdorff dimension (see Theorem 1.3 of Ref.~\cite{blachere2011harmonic})\footnote{We note that this Theorem may be potentially used to provide a more direct proof of our Theorem~\ref{thm:quant_hyperbolic_rw}, which we leave to future work.}:
    \begin{theorem}[Blach\'ere, Ha\"issinsky, Mathieu~\cite{blachere2011harmonic}]
    Suppose $X$ is the Cayley graph of a (non-elementary) hyperbolic group with Gromov boundary $\partial X$ admitting a visual metric $d_{o, \epsilon}$.  Let $\tau$ be the hitting measure of a heat kernel measure with step distribution $\mu$ that has bounded support.  Denote by $\mathcal{B}_z(r)$ a ball in $\partial X$ of radius $r$ centered at $z \in \partial X$.  Then,
    \begin{equation}
    \lim_{r \to 0}\frac{\log \tau(\mathcal{B}_z(r))}{\log r} = \frac{s}{\epsilon v}
    \end{equation}
    for $\tau$-almost every $z \in \partial X$.
    \end{theorem}
    
    Our argument will be heuristic from here on.  Using the above result, for a ball of radius $\alpha$ which is small enough, the measure $\tau(\mathcal{B}_z(\alpha)) \sim \alpha^{s/(\epsilon v)}$.  Choosing $S_\beta$ to be a ball $\mathcal{B}_z(\beta)$ for some $z \in P$ we have 
    \begin{equation}
    \tau(S_\beta) \sim \alpha^{s/(\epsilon v)} \left(\frac{\beta}{\alpha}\right)^{d_a}
    \end{equation}
    and choosing $\partial S_{\beta} = S_{\beta + r \alpha} \setminus S_{\beta}$, we have
    \begin{equation}
    \tau(\partial S_{\beta}) = \tau(\partial S_{\beta + r \alpha}) - \tau(\partial S_{\beta}) \lesssim \alpha^{s/(\epsilon v)} \left(\frac{\beta}{\alpha}\right)^{d_a - 1},
    \end{equation}
    the weighted expansion of this set scales like $\tau(\partial S_{\beta})/\tau(S_{\beta}) \lesssim \alpha/\beta$.  We expect that for hyperbolic groups $\alpha$ is exponentially small in $L$ while $\beta$ is $O(1)$, which would imply exponentially small in $L$ expansion for the Cayley graph restricted to a ball of diameter $L$.  However, this argument is clearly a sketch and a more rigorous proof may also provide an avenue towards a general proof for all non-amenable groups (since we would only need to prove that the hitting measure for the boundaries of such groups have finite Hausdorff dimension).
	\section{Details about the spin-1 breakdown model}\label{app:breakdown}
	
	In this appendix, we give a detailed analysis of the Hilbert space structure of the spin-1 breakdown model.
	
	%\subsection{Ergodicity within Krylov sectors}
	
	%Firstly, we will prove that $\dyn$ is fully ergodic within each charge sector. 
	
	%Assume that there are two sectors with the same amount of symmetric charge but are dynamically disconnected. Consider the states in each sector with 
	%the least real particles, which means moves of canceling 2 particles on site $i$ and creating 1 particle on site $i+1$ are not allowed for all sites. There must be a site $j$ with the smallest charge weight $2^{j-1}$ that two states differ from each other. If the difference is 1 particle, since other sites having differences all have charge weight of integral multiple of $2^j$, the requirements that 2 sectors having the same number of charges can't be satisfied. So the difference of site $j$ must be 2 which means one state has 2 particles and the other state has 0 particles on site $i$. We have preassumed that moves of canceling 2 particles on site $i$ and creating 1 particle on site $i+1$ are forbidden for all sites, so the site $j+1$, $j+2$... to the end of the chain should all have 2 particles for the state having 2 particles on site $j$, making the charge number always larger than the other state with 0 particles on site $j$ since they have exactly the same configurations on site 1 to $j-1$. This also contradicts the assumptions that they have the same amount of symmetric charges. We then show that the dynamic is indeed ergodic within each symmetric sector.

	\subsection{Recurrence relation for $|\mathcal{K}_Q|$}
	
	Denote by $|\mathcal{K}_Q(L)|$ the dimension of the Krylov sector with charge $Q$ and total length $L$. Since $n_i=S^z_i+1$ could be 0, 1, or 2, the maximum value of $Q$ is $Q_{\rm max}=2\sum_{i=1}^{L}2^{i-1}=2^{L+1}-2$ and the total Hilbert space dimension satisfies $\sum_{Q=0}^{Q_{\rm max}}|\mathcal{K}_Q(L)|=3^{L}$. 
	
	Due to the particle-hole symmetry about half filling, the sizes of the Krylov sectors are also symmetric with respect to $Q_{\rm max}/2$:
	\begin{equation}
		|\mathcal{K}_{Q}(L)|=|\mathcal{K}_{Q_{\rm max}-Q}(L)|.
	\end{equation}
	Therefore, below we only need to consider sectors with $Q \leq Q_{\rm max}/2 = 2^L-1$. Notice that for $Q \leq 2^L-1$, the rightmost site can have at most one particle, and hence particles cannot propagate further to the right even if we were to append more empty sites to its right. We thus conclude that $|\mathcal{K}_{Q}(L+1)|=|\mathcal{K}_{Q}(L)|$ for $Q\leq2^{L}-1$.
	
	%Under this consideration, we can see that the furthest site $l$ a state of charge $Q$ can reach during its dynamics satisfies $2^{l-1}\leq Q <2^l$, which means for a chain with length longer than $l$, the finite length of the chain will not constraint the dynamics of $Q$ sector. $|\mathcal{K}_{Q}(L)|=|\mathcal{K}_{Q}(L+1)|$ holds for $Q\leq2^{L}-1$. 
	
	Now consider some state with charge $Q \leq 2^L-1$. If $Q$ is odd, the first site has exactly 1 particle and is frozen. Thus $|\mathcal{K}_Q(L)|$ is completely determined by charges of the system excluding the first site. Relabelling each site $ 2 \leq i \leq L$ as $i-1$ amounts to lowering the weight of each site in computing $Q$ by a factor of 2, and hence 
	\begin{equation}
		|\mathcal{K}_{Q}(L)|=|\mathcal{K}_{(Q-1)/2}(L-1)|, \quad {\rm for} \ Q \ {\rm odd}. 
	\end{equation}
	As explained above, the sizes of the Krylov sectors for the values of $Q$ we focus on do not change if we append more empty sites to the right of the system, we can safely write
	\begin{equation}
		|\mathcal{K}_{Q}(L)|=|\mathcal{K}_{(Q-1)/2}(L-1)| = |\mathcal{K}_{(Q-1)/2}(L)|, \quad {\rm for} \ Q \ {\rm odd}.
	\end{equation}
	
	%then each state can be labeled by the configuration of $Q-1$ particles on the rest of the chains on its right with the length shortened by 1. So we only need to consider the rest of the chain. We can lower the weight of each site by 2, and the dimension of this sector can be written as $|\mathcal{K}_{Q}(L)|=|\mathcal{K}_{(Q-1)/2}(L-1)|$. We are considering only the left patch of the whole Hilbert space, so the length of the chain doesn't have constraints on the dynamics, $|\mathcal{K}_{Q}(L)|=|\mathcal{K}_{(Q-1)/2}(L-1)|=|\mathcal{K}_{(Q-1)/2}(L)|$  
	
	For $Q$ even, the first site has either 0 or 2 particles, and the total size of $\mathcal{K}_Q$ is the sum of these two types of configurations. Using the same line of reasoning as above, we have 
	\begin{eqnarray}
		|\mathcal{K}_{Q}(L)|&&=|\mathcal{K}_{Q/2}(L-1)|+|\mathcal{K}_{Q/2-1}(L-1)|  \nonumber \\
		&&=|\mathcal{K}_{Q/2}(L)|+|\mathcal{K}_{Q/2-1}(L)|, \quad {\rm for} \ Q \ {\rm even}.  \nonumber \\
	\end{eqnarray}
	In the above equations, the first term in the sum represents contributions from configurations with $n_{i=1}=0$, and the second term for $n_{i=1}=2$.
	
	To summarize, we have the following recurrence relation for $|\mathcal{K}_Q(L)|$ with $Q\leq 2^L-1$:
	\begin{equation}
		\label{recurrence_breakdown}
		|\mathcal{K}_{Q}(L)|=
		\left\{
		\begin{aligned}
			&\ |\mathcal{K}_{(Q-1)/2}(L)|,\quad {\rm for} \ Q \ {\rm odd;}\ \\
			&\ |\mathcal{K}_{Q/2}(L)|+|\mathcal{K}_{Q/2-1}(L)|,\quad {\rm for} \ Q \ {\rm even.}
		\end{aligned}
		\right.
	\end{equation}
	
	%With this recurrence relation and particle-hole symmetry, we can efficiently calculate the dimension of each Krylov sector numerically, which is drawn in fig.\ref{fig:breakdown_Krylov_dimension} The Krylov graph shows an interesting self-similar structure, which leads to the golden ratio exponential growth of the largest Krylov sector and the diffusive nature of its dynamics, as discussed in the following sections.
	
	%---------------------------------------------------------------------------------
	%    \begin{figure}
		%        \centering
		%        \includegraphics[width=0.5\textwidth]{Breakdown_Krylov_Dimension.pdf}
		%        \caption{Dimension of each Krylov sector with conserved charge $Q$. }
		%        \label{fig:breakdown_Krylov_dimension}
		%    \end{figure}
	%---------------------------------------------------------------------------------
	
	\subsection{Size of the largest Krylov sector}
	
	We show that the dimension of the largest Krylov sector grows as $|\mathcal{K}_{\rm max}(L)| \sim \phi^L$, where $\phi=(1+\sqrt{5})/2$ is the golden ratio. The first question we need to address is: which charge sectors have the largest size? We observe from Fig.~\ref{fig:breakdown_Krylov_dimension} that there are four charge sectors (two from each particle-hole sector) having the largest size. Below we will first work out the exact $Q$ of these four largest sectors. Then we derive the asymptotic size of these largest sectors. Again, we restrict ourselves to $Q\leq 2^L-1$ and the other two sectors are related by particle-hole symmetry.
	
	To gain some intuition, we start by listing $|\mathcal{K}_Q(L)|$ for $2 \leq L \leq 4$:
	\begin{itemize}
		\item $L=2$:  
		
		$|\mathcal{K}_0(2)|=1,\ |\mathcal{K}_1(2)|=1,\ {\color{magenta}{|\mathcal{K}_2(2)|=2}}, \ |\mathcal{K}_3(2)|=1;$
		
		\item $L=3$:
		
		${\color{magenta}{|\mathcal{K}_4(3)|=3}},\ |\mathcal{K}_5(3)|=2,\ {\color{magenta}{|\mathcal{K}_6(3)|=3}}, \ |\mathcal{K}_7(3)|=1;$
		
		\item $L=4$:
		
		$|\mathcal{K}_8(4)|=4,\ |\mathcal{K}_9(4)|=3,\ {\color{magenta}{|\mathcal{K}_{10}(4)|=5}}, \\ |\mathcal{K}_{11}(4)|=2, \
		{\color{magenta}{|\mathcal{K}_{12}(4)|=5}}, \ |\mathcal{K}_{13}(4)|=3, \\ |\mathcal{K}_{14}(4)|=4,
		\ |\mathcal{K}_{15}(4)|=1.$        
	\end{itemize}
	These numbers can be easily obtained by using the recurrence relation Eq.~(\ref{recurrence_breakdown}). For each $L$, the largest sectors are highlighted. We make the following simple but important observations from the above list:
	\begin{enumerate}
		\item The largest sectors all have $Q$ even;
		
		\item There are two largest sectors for all $L\geq 3$.
	\end{enumerate}
	The first observation can be easily justified by inspecting the recurrence relation~(\ref{recurrence_breakdown}). For any pair of consecutive charges $(Q, Q+1)$ with $Q$ odd, their Krylov sectors have dimensions $(|\mathcal{K}_{(Q-1)/2}(L)|, |\mathcal{K}_{(Q-1)/2}(L)| + |\mathcal{K}_{(Q+1)/2}(L)|)$. Hence even charge sectors are always bigger than the adjacent odd sectors. 
	
	The second observation requires more care, but is also straightforward to see. Let us write explicitly how we obtain $|\mathcal{K}_Q(L)|$ of the two largest sectors for $L=3$, using data from $L=2$:
	\begin{eqnarray}
		|\mathcal{K}_4(3)| &=& |\mathcal{K}_2(2)| + |\mathcal{K}_1(2)| = 2 + 1, \nonumber \\
		|\mathcal{K}_6(3)| &=& |\mathcal{K}_3(2)| + |\mathcal{K}_2(2)| = 1 + 2. 
	\end{eqnarray}
	Notice that both expressions not only involve the size of the largest even charge sector at system size $L-1$, but also the largest {\it odd} charge sector at $L-1$. This is also a direct consequence of Eq.~(\ref{recurrence_breakdown}), since when $Q$ even, one of the two contributions in the sum $(Q/2, Q/2-1)$ must be odd. Therefore, the appearance of two equally largest charge sectors can be traced back to the fact that there are two odd charge sectors with equal size at the very beginning of the recurrence relation $L=2$: $|\mathcal{K}_1(2)|=|\mathcal{K}_3(2)|=1$, which then propagate to larger $L$'s and generate two sequences following the recurrence relation. 
	
	Now we can put the simple observations made above into a more mathematical form. Denote by $Q^\star(L)$ the charge of the largest Krylov sector for system size $L$. We know from observation 1 that $Q^\star$ must be even. From observation 2 we know that there are two $Q^\star$'s with equal size, which we denote as $Q^\star_1(L)$ and $Q^\star_2(L)$ with $Q^\star_1(L)<Q^\star_2(L)$.
	For convenience, we also denote by $Q^\star_{\rm odd}(L)$ the charge of the largest {\it odd} charge sector for system size $L$. From the first line of Eq.~(\ref{recurrence_breakdown}), we have 
	\begin{equation}
		Q^\star_{\rm odd}(L+1) = 2 Q^\star(L) + 1.
	\end{equation}
	Likewise, $Q^\star_{1,2}$ will also give rise to two sequences of $Q^\star_{\rm odd, 1/2}$. We claim that the following relation holds for all $L$:
	\begin{itemize}
		\item for $L$ even,
		\begin{eqnarray}
			Q^\star_{\rm odd,1}(L) &=& Q^\star_1(L)-1  \\
			Q^\star_{\rm odd,2}(L) &=& Q^\star_2(L)+1
		\end{eqnarray}
		
		\item for $L$ odd,
		\begin{eqnarray}
			Q^\star_{\rm odd,1}(L) &=& Q^\star_1(L)+1  \\
			Q^\star_{\rm odd,2}(L) &=& Q^\star_2(L)-1.
		\end{eqnarray}       
	\end{itemize}
	We can directly check the above relations using the explicit values for $ 2 \leq L \leq 4$ listed above. In general, one can prove this by induction. Suppose $L$ is even, and $Q^\star_{\rm odd,1}(L) = Q^\star_1(L)-1$ holds. For $L+1$, we have
	\begin{equation}
		Q^\star_{\rm odd,1}(L+1) = 2 Q^\star_1(L)+1.
		\label{eq:a12}
	\end{equation}
	Since $|\mathcal{K}_{Q^\star_1}(L+1)| =|\mathcal{K}_{Q^\star_1}(L)| + |\mathcal{K}_{Q^\star_{\rm odd,1}}(L)|$, we must have
	\begin{equation}
		Q^\star_1(L+1) = 2 Q^\star_1(L).
		\label{eq:a13}
	\end{equation}
	By comparing Eqs.~(\ref{eq:a12}) and (\ref{eq:a13}), we see that $Q^\star_{\rm odd,1}(L+1) = Q^\star_1(L+1)+1$. One can similarly prove the rest of the relations by induction. In the process of the proof, we also obtain the following set of recurrence relations for $Q^\star$ alone:
	\begin{equation}
		%\label{recurrence_breakdown}
		Q^\star_1(L+1) =
		\left\{
		\begin{aligned}
			&\ 2Q^\star_1(L),\quad {\rm for} \ L \ {\rm even;}\ \\
			&\ 2(Q^\star_1(L)+1),\quad {\rm for} \ L \ {\rm odd,}
		\end{aligned}
		\right.
	\end{equation}
	
	\begin{equation}
		%\label{recurrence_breakdown}
		Q^\star_2(L+1) =
		\left\{
		\begin{aligned}
			&\ 2(Q^\star_2(L)+1),\quad {\rm for} \ L \ {\rm even;}\ \\
			&\ 2Q^\star_2(L),\quad {\rm for} \ L \ {\rm odd.}
		\end{aligned}
		\right.
	\end{equation}

	The solution for the sequence $Q^\star_1$ is
	\begin{equation}
		\left\{
		\begin{aligned}
			&Q^\star_1(2l)=\frac{1}{3}2^{2l+1}-\frac{2}{3}\\
			&Q^\star_1(2l+1)=\frac{1}{3}2^{2l+2}-\frac{4}{3}
		\end{aligned}
		\right.
	\end{equation}
	We find that this sequence approaches $\frac{1}{3}Q_{\rm max}(L)$ for both odd and even system sizes as $L$ becomes large. Similarly, the solution for the sequence $Q^\star_2$ is
	% \begin{equation}
		%     \left\{
		%     \begin{aligned}
			%         &q_{max}(2l+1)=2q_{max}(2l)+2\\
			%         &q_{max}(2l+2)=2q_{max}(2l+1)
			%     \end{aligned}
		%     \right.
		% \end{equation}
	\begin{equation}
		\left\{
		\begin{aligned}
			&Q^\star_2(2l)=\frac{5}{12}2^{2l+1}-\frac{4}{3}\\
			&Q^\star_2(2l+1)=\frac{5}{12}2^{2l+2}-\frac{2}{3}
		\end{aligned}
		\right.
	\end{equation}
	We find that this sequence approaches $\frac{5}{12}Q_{\rm max}(L)$ for both odd and even system sizes as $L$ becomes large. Including their particle-hole partners, we conclude that the four largest Krylov sectors have charges
	\begin{equation}
		\left\{ \frac{1}{3}Q_{\rm max}, \quad \frac{5}{12}Q_{\rm max}, \quad \frac{7}{12}Q_{\rm max}, \quad  \frac{2}{3}Q_{\rm max} \right\}. \nonumber
	\end{equation}
	
	%In this way, we can locate the largest sector at $\frac{1}{3}Q_{max}(L)$, $\frac{5}{12}Q_{max}(L)$, $\frac{7}{12}Q_{max}(L)$ and $\frac{2}{3}Q_{max}(L)$ considering the particle-hole symmetry. 
	
	%At last, we will prove the golden ratio scaling of the size of the largest Krylov sector $|\mathcal{K}_{max}$. Consider the $\frac{1}{3}Q_{max}(L)$ point for example, we can easily derive the recurrence relation for this sector size using Eq. \ref{recurrence_breakdown}. 
	%\zhicheng{To be continued....}
	
	Now it is straightforward to obtain a recurrence relation for the sizes of the largest Krylov sector $|\mathcal{K}_{\rm max}|$. Let us take the sector $Q= Q_{\rm max}/3$ as an example.
	For $L=2l$ even, we have
	\begin{equation}
		\begin{aligned}
			|\mathcal{K}_{\rm max}(2l)|&=|\mathcal{K}_{\frac{1}{3}2^{2l+1}-\frac{2}{3}}(2l)|\\
			&=|\mathcal{K}_{\frac{1}{3}2^{2l}-\frac{1}{3}}(2l)|+|\mathcal{K}_{\frac{1}{3}2^{2l}-\frac{4}{3}}(2l)|\\
			&=|\mathcal{K}_{\frac{1}{3}2^{2l-1}-\frac{2}{3}}(2l)|+|\mathcal{K}_{\frac{1}{3}2^{2l}-\frac{4}{3}}(2l)|\\
			&=|\mathcal{K}_{\frac{1}{3}2^{2l-1}-\frac{2}{3}}(2l-2)|+|\mathcal{K}_{\frac{1}{3}2^{2l}-\frac{4}{3}}(2l-1)|\\
			&=|\mathcal{K}_{\rm max}(2l-2)|+|\mathcal{K}_{\rm max}(2l-1)|.
		\end{aligned}
	\end{equation}
	In going from the first line to the second line, we used Eq. (\ref{recurrence_breakdown})
	
	For $L=2l+1$ odd, we have
	\begin{equation}
		\begin{aligned}
			|\mathcal{K}_{\rm max}(2l+1)|&=|\mathcal{K}_{\frac{1}{3}2^{2l+2}-\frac{4}{3}}(2l+1)|\\
			&=|\mathcal{K}_{\frac{1}{3}2^{2l+1}-\frac{2}{3}}(2l+1)|+|\mathcal{K}_{\frac{1}{3}2^{2l+1}-\frac{5}{3}}(2l+1)|\\
			&=|\mathcal{K}_{\frac{1}{3}2^{2l+1}-\frac{2}{3}}(2l+1)|+|\mathcal{K}_{\frac{1}{3}2^{2l}-\frac{4}{3}}(2l+1)|\\
			&=|\mathcal{K}_{\frac{1}{3}2^{2l-1}-\frac{2}{3}}(2l)|+|\mathcal{K}_{\frac{1}{3}2^{2l}-\frac{4}{3}}(2l-1)|\\
			&=|\mathcal{K}_{\rm max}(2l)|+|\mathcal{K}_{\rm max}(2l-1)|
		\end{aligned}
	\end{equation}
	This gives a Fibonacci sequence for the size of the largest sector, which directly leads to the asymptotic scaling $|\mathcal{K}_{\rm max}(L)| \sim \phi^L$.

	\section{Details of the $tJ_z$ model}\label{app:tJz}
	In this section, we derive rigorous bounds on the thermalization time of the $tJ_z$ model. As stated in the main text, we focus on the stochastic dynamics associated with the circuit-averaged evolution of the RU dynamics $\dyn$. A single timestep of evolution corresponds to the channel
	\be \mcc(\r) = \mcu^\dag \(\Tr_L[\r] \tp \frac{\unit}N\) \mcu,\ee
	where $\Tr_L[\r] \tp \frac{\unit}N$ represents the effect of depolarizing noise coupled to the boundary: the last site of the system is traced out and replaced by a maximally mixed state. $\mcu$ is the depth-2 brickwork circuit consisting of constraint-preserving Haar-random unitary $U_{i,i+1}$:
	\be \mcu = \( \bigotimes_{i=1}^{L/2-1}  U_{2i,2i+1}\) \(\bigotimes_{i=1}^{L/2}  U_{2i-1,2i}\) .\ee 
	For the $tJ_z$ dynamics, we consider unitary gates of the form
	\bea U_{i,i+1}&=\sum_{a\in\{\uparrow,\downarrow\}}C_a\kb{0a}{a0}+h.c.+\sum_{\substack{a,b\neq 0,\\ a=b=0}}e^{i\phi_{ab}}\kb{ab}{ab} \\
	&=U_{tJ_z}^\uparrow\oplus U_{tJ_z}^\downarrow\oplus(\bigoplus_{\substack{a,b\neq 0,\\ a=b=0}}e^{i\phi_{ab}}), 
	\eea
	where the matrices $U_{tJ_z}^a$ are drawn from the Haar random ensemble on $U(2)$, and the rest of the terms in the direct sum act as a diagonal matrix of random phases on the states that are frozen under the dynamics.
	We are interested in the evolution of the state $\rho(t)$ under the circuit-averaged dynamics, i.e.,
	\be \overline{\rho}(t)\equiv \oEE_{\mcc_t}[\mcc_t(\rho)], \ee
	where $\oEE_{\mcc_t}$ denotes averaging over the unitaries constituting $\mcu$. We divide each unit time interval into three steps of length $t=1/3$, where the depolarizing noise is applied at $t\in\nn$, the first layer of $\mcu$ is applied at $t\in\nn+1/3$, and the second layer of $\mcu$ is applied at $t\in\nn+2/3$. Therefore,
	\be \overline{\rho}(t+1/3)=\Tr_L[\overline{\rho}(t)]\otimes\unit. \ee
	Decomposing the state as $\overline{\rho}=\otimes_{i=1}^{L/2-1}\overline{\rho}_{2i,2i+1}$, then averaging over the Haar ensemble gives
	\bea \overline{\rho}_{2i,2i+1}(t+\frac{2}{3})&=\frac{1}{2}\sum_{a\in\{\uparrow,\downarrow\}}\Tr[\overline{\rho}_{2i,2i+1}(t+\frac{1}{3})\Pi_a^{tJ_z}]\Pi_a^{tJ_z} \\ &+ \Tr_{\substack{a,b\neq 0,\\ a=b=0}}[\overline{\rho}_{2i,2i+1}(t+\frac{1}{3})\Pi_{ab}]\Pi_{ab}, \eea
	where the projectors are defined as
	\bea \Pi_a^{tJ_z}&=\kb{0a}{0a}+\kb{a0}{a0}, \\
	\Pi_{ab}&=\kb{ab}{ab}. \eea
	The action of the second layer of the unitary is similar.
	Under time evolution, $\overline{\rho}_{2i,2i+1}$ becomes diagonal in the computational basis, resulting in an effective Markov generator on the probability distribution of the basis states $\{|\psi\rangle\}$,
	\bea \mathbf{p}(\{|\psi\rangle\},t+1)&=\mcm\mathbf{p}(\{|\psi\rangle\},t)\\
	&=\mcm_o\mcm_e\mcm_L \mathbf{p}(\{|\psi\rangle\},t), \eea
	where
	\bea \mcm_L&=\unit_{L-1}\otimes\frac{1}{3}\sum_{a,b}\kb{a}{b} \\
	\mcm_e&=\bigotimes^{L/2-1}_{i=1}M^{tJ_z}_{2i,2i+1}\\
	\mcm_e&=\bigotimes^{L/2-1}_{i=1}M^{tJ_z}_{2i-1,2i} \eea
	and the 2-site stochastic matrix is defined as
	\be M^{tJ_z}\equiv \frac{1}{2}\sum_{a\in\{\uparrow,\downarrow\}}(|0a\rangle+|a0\rangle)(\langle 0a|+\langle a0|)+\sum_{\substack{a,b\neq 0,\\ a=b=0}}\kb{ab}{ab}. \ee
	Therefore, the steady state of this Markov process is the uniform distribution over the Hilbert space, which coincides with the fact that $\overline{\rho}$ will eventually thermalize to a maximally mixed state under $\dyn$.
	
	As is also explained in the main text, it is convenient to consider a non-local Markov process defined as
	\be \mcm_{\text{nonloc}}\equiv (\mcm_o\mcm_e)^{\infty}\mcm_L, \ee
	such that a state thermalizes within a Krylov sector as soon as it reaches that sector. Physically one can think of this as applying a sufficiently deep random unitary circuit between consecutive applications of the boundary noise. By doing so, we essentially focus solely on the inter-sector dynamics due to the boundary depolarizing noise and ignore the intra-sector thermalization process.  The lack of locality in the dynamics leads to a thermalization time that is faster by a factor of $\mco(1/L)$, since the number of unitaries needed to spread the effect of the bath to the bulk of the system is $\mco(L)$. However, since we are after a lower bound on the thermalization time, there is no loss of generality in considering non-local dynamics.
	
	\subsection{Expansion and thermalization time}
	Following the discussion in \ref{sec: graph expansion}, we would like to find the expansion $\Phi(\mcg_\mck)$ of the Markov process $\mcm_{\text{nonloc}}$, defined as 
	\bea \Phi(\mcg_\mck)&=\min_{R\in \mcg_\mck:|R|\leq|\mch|/2}\Phi(R), \\
	\Phi(R)&\equiv\frac{1}{|R|}\sum_{\psi\in R,\psi'\in R^c}\langle\psi'|\mcm_{\text{nonloc}}|\psi\rangle. \eea
	
	For a Krylov sector $\mck_{\mathbf{s}_d}$ labeled by the spin pattern $\mathbf{s}_d=(s_1s_2\dots s_d)$, we define the cone $C_{\mathbf{s}_d}$ to be a subset of the Hilbert space that contains all the states with the same spin pattern up to the $d$-th particle, i.e.,
	\be C_{\mathbf{s}_d}\equiv \bigoplus_{\mathbf{s}'_l=(s_1s_2\dots s_d)\times(s'_{d+1}\dots s'_l), l\geq d}\mck_{\mathbf{s}'_l}. \ee
	When the exact spin pattern is not important, we can represent $\mck_{\mathbf{s}_d}$ and $C_{\mathbf{s}_d}$ as $\mck_d$ and $C_d$ (see Fig.~\ref{fig:tJz} for an illustration of $C_1$). As illustrated in Fig.~\ref{fig:tJz}, $C_d$ can be visualized as a branch of the binary tree $\mcg_\mck$ at depth $d$.
	
	To compute $\Phi(C_d)$, it is easy to see that the probability for the state $|\psi\rangle$ in $C_d$ to move to the complement $C_d^c$ under one application of the depolarizing noise is only nonzero when $|\psi\rangle$ has $d$ particles with the last site occupied, i.e.,
	\bea \sum_{\substack{\psi\in C_d,\\ \psi'\in C_d^c}}\langle\psi'|\mcm_{\text{nonloc}}|\psi\rangle &=\frac{2}{3}|\{\psi\in\mck_d:\sum_\sigma c^\dagger_{L,\sigma} c_{L,\sigma}|\psi\rangle\neq 0\}|\\
	& =\frac{2}{3}{L-1 \choose d-1}. \eea
	The denominator is simply the dimension of $C_d$,
	\be |C_d|=\sum_{l=0}^{L-d} 2^l|\mck_{d+l}|=\sum_{l=0}^{L-d} 2^l {L\choose d+l}. \ee
	Since $|C_d|$ increases as $d$ decreases, and the transition probability decreases as $d$ decreases for $d<L/2$, the expansion of the Krylov graph is lower bounded by the expansion of the cone at $d=1$, 
	\be \Phi(\mcg_\mck)\geq\Phi(C_1)=\frac{2}{3}\frac{1}{\sum_{l=0}^{L-1}2^l{L \choose 1+l}}=\frac{4}{3(3^L-1)}. \ee
	Hence, from Cheeger's inequality in Eq.~(\ref{cheeger}), the thermalization time is lower bounded by
	\be \tth\geq\frac{1}{2\Phi(\mcg_\mck)}\geq\frac{3}{8}(3^L-1), \ee
    which scales exponentially as the system size.
	%which has the same scaling of the Hilbert space dimension up to some constant coefficient.
	
	\subsection{Magnetization relaxation time}
	Now we examine the relaxation time of the expectation value of local operators, namely, the average magnetization
	\be m\equiv\frac{1}{L}\sum_{i=1}^L(\kb{\uparrow}{\uparrow}_i-\kb{\downarrow}{\downarrow}_i). \ee
	We define the relaxation time $t_m(\gamma;\psi)$ as the time needed for the expectation value of $m$ in the circuit-averaged state $\overline{\rho}(t,\psi)\equiv \oEE_{\mcc_t}[\mcc_t(\kb{\psi}{\psi})]$ to relax below $\gamma$ of its equilibrium value 0:
	\be t_m(\gamma;\psi)\equiv \min\{t:|\lan m\ran_{\overline{\rho}(t,\psi)}|\leq \gamma\}. \ee
	Then the thermalization time is lower bounded by the magnetization relaxation time $t_m(\gamma, \psi_{\max})$ of the state with the maximal magnetization $\lan m\ran=1$, i.e., $|\psi_{\max}\ran\equiv |\uparrow\ran^{\otimes L}$. The inhomogeneity of the distribution of charges in the Krylov graph allows us to select a subspace $A$ of $\mch$ with nonzero expectation value of $m$. We will show that when $A$ is large enough, it takes exponentially long time for a state initialized in $A$ to move to the complement $A^c$, which sets a lower bound for $t_m(\gamma; \psi_{\max})$. Define $A$ to be the cone $C_{\mathbf{s}_{\eta L}}$ that contains all the states with the same spin pattern $\mathbf{s}_{\eta L}=(\uparrow\uparrow\dots\uparrow)$ for the first $\eta L\in \mathbb{N}$ particles, we are interested in the expectation value of $m$ at time $t$ over initial states in $A$,
	\bea \lan m\ran_A(t)&\equiv \oEE_{\phi\in A}\lan m\ran_{\overline{\rho}(t;\phi)} \\
	&=\oEE_{\phi\in A}\biggl(\sum_{\phi'\in A} P(\psi(t)=\phi'|\psi(0)=\phi)\lan m\ran_{\phi'}\\
	& \quad +\sum_{\phi'\in A^c} P(\psi(t)=\phi'|\psi(0)=\phi)\lan m\ran_{\phi'} \biggr) \\
	&\geq\oEE_{\phi\in A}\left(\sum_{\phi'\in A} P(\psi(t)=\phi'|\psi(0)=\phi)\lan m\ran_{\phi'}\right)\\
	& \quad -P(\psi(t)\in A^c|\psi(0)=\phi), \eea
	where the last inequality holds since $\lan m\ran_\phi\geq -1$ for all $\phi$. In order to compute the first term, we define $\chi_d^A$ to be the set of states in $A$ with the same number of particles $d$ with $d\geq\eta L$. Then it is easy to see that the average transition probability $\oEE_{\phi\in A} P(\psi(t)=\phi'|\psi(0)=\phi)$ is the same for all $\phi'\in\chi_d^A$, and since $\sum_{\phi'\in\chi_d^A}\lan m\ran_{\phi'}=\eta$, we have
	\bea \quad &\oEE_{\phi\in A}\biggl(\sum_{\phi'\in A}P(\psi(t)=\phi'|\psi(0)=\phi)\lan m\ran_{\phi'}\biggr) \\ 
	&=\oEE_{\phi\in A}\biggl(\sum_{d=\eta L}^L \eta P(\psi(t)\in \chi_d^A|\psi(0)=\phi)\biggr) \\
	&=\eta P(\psi(t)\in A|\psi(0)\in A). \eea
	Therefore,
	\bea \lan m\ran_A(t)&\geq\eta-(1+\eta)P(\psi(t)\in A^c|\psi(0)\in A) \\
	&\geq \eta-(1+\eta)t\Phi(A), \eea
	where the transition probability is bounded from above by $t\Phi(A)$, 
    % from Eq.\ref{eq: transition prob}, 
    and the expansion of $A$ is given by 
	\be \Phi(A)=\Phi(C_{\eta L})=\frac{2}{3}{L-1 \choose \eta L-1}\Bigg/ \sum_{l=0}^{(1-\eta)L} 2^l{L\choose \eta L+l}. \ee
	To get a tight bound of $\lan m\ran_A$, we need to take $\frac{1}{L}\geq \eta\ll 1$ such that the expansion of the cone is exponentially small. Using Stirling's approximation for $\eta\ll 1$, we have
	\be {L-1 \choose \eta L-1}\approx\frac{(L-1)^{\eta L-1}}{(\eta L-1)!}, \ee
	The denominator becomes,
	\bea \sum_{l=0}^{(1-\eta)L}& 2^l{L\choose \eta L+l}=2^{-\eta L}\Bigg[3^L-\sum_{d=0}^{\eta L-1}2^d{L\choose d}\Bigg] \\
	& > 2^{-\eta L}\Bigg[3^L-(2L)^{\eta L-1}/(\eta L-2)!\Bigg],\\
	\eea
	where the second term is exactly $1$ when $\eta L=1$. Therefore, the inverse of the expansion gives
	\be \Phi(A)^{-1}\gtrapprox\frac{3}{2} 3^{(1-\eta \log_3 2)L}\bigg(\frac{\eta L-1}{e(L-1)}\bigg)^{\eta L-1}+C, \ee
	where the second term becomes a constant $C$ after taking the limit $L\to\infty$ while keeping $\eta L$ finite, and the first term is obtained using Stirling's approximation $x!\approx(x/e)^x$.
	
	The lower bound on $t_m(\gamma;\psi_{\max})$ is given by taking $\eta=2\gamma$, i.e.,
	\bea t_m(\gamma;\psi_{\max})&\gtrapprox 3^{(1-2\gamma\log_3[\frac{e(L-1)}{\gamma L-1}])L}\times\frac{e(L-1)}{2\gamma L-1}\times\frac{3}{2(1+2\gamma)}\\
	& \approx 3^{(1-2\gamma\log_3[e/\gamma])L}\times\frac{3e}{4\gamma(1+2\gamma)}.
	\eea
	When $\gamma\to \frac{1}{L}$, we have $\lan m\ran_A(t)\geq 2\gamma-t\Phi(C_1)$, so that we recover the bound of the thermalization time 
	\be t_m(\gamma\to \frac{1}{L};\psi_{\max})\approx t_{\text{th}}\geq \frac{3}{8}(3^L-1). \ee

	\section{The expansion $\Phi(\mcg_{\mathcal{K}})$ of the particle-conserving East model}
	\label{app:east}

	\begin{figure}
		\centering
		\includegraphics[width=0.45\textwidth]{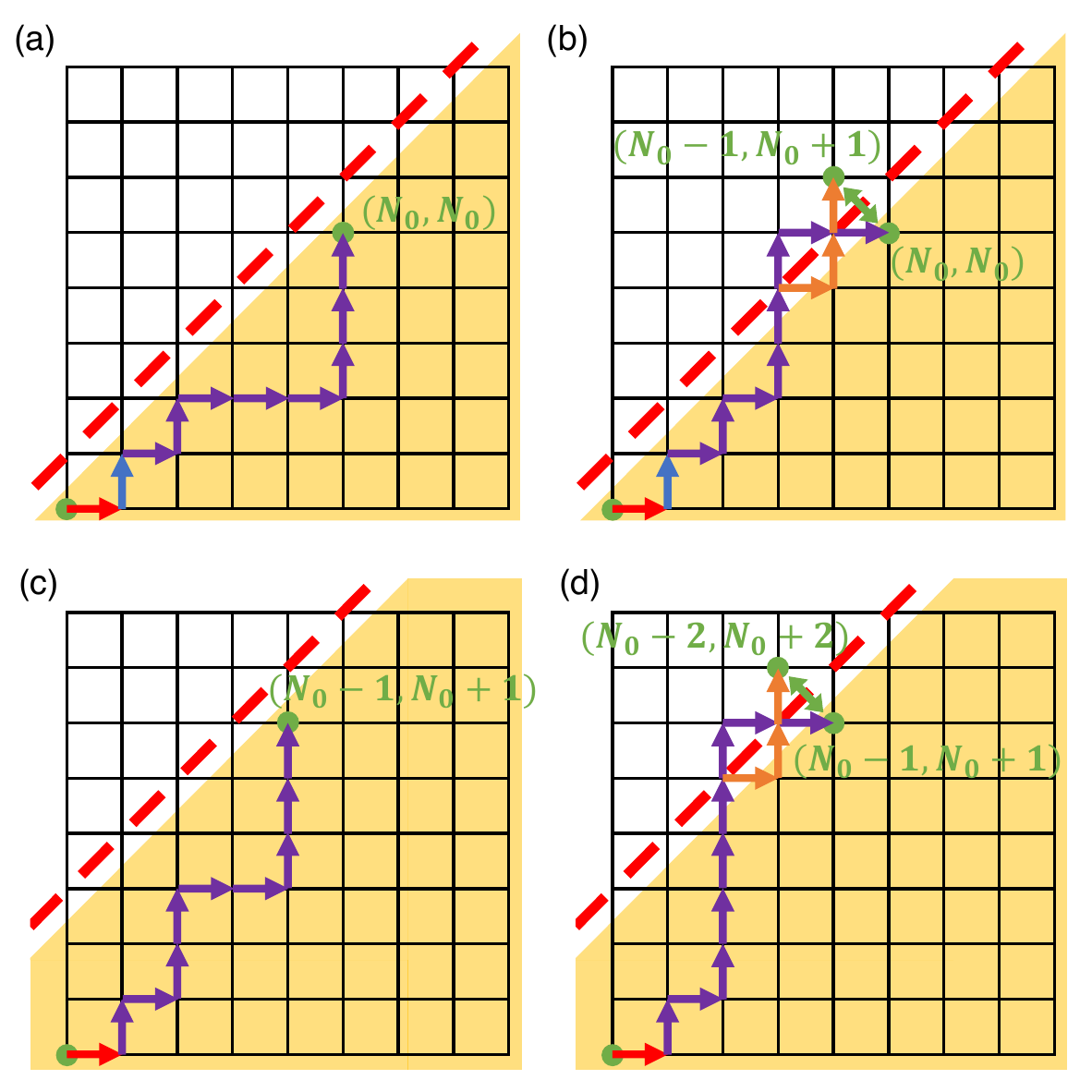}
		\caption{Computing the size of the thermalized Krylov sectors by mapping to the combinatorial problem of counting the number of allowed monotonic paths on a lattice. (a) and (b): Allowed and disallowed monotonic paths in the numerator of Eq. (\ref{eq:east_expansion}). The red arrow represents the first site fixed to be occupied while the blue arrow represents the second site fixed to be empty in the numerator. (c) and (d): Allowed and disallowed monotonic paths in the denominator. The 2 figures are examples of the sector $\mathcal{K}(N_0-1, 2N_0-1)$, where the red line is moved up to $y=x+3$ and the end point of the path is moved to $(N_0-1, N_0+1)$, according to the constraints. The first site is still fixed to be occupied while the second site has no constraints in the denominator.}
		\label{fig:East_calculation}
	\end{figure}
	
	%--------------------------------------------------------------------------------------------------------
	
	%\subsection{Krylov graph expansion and thermalization time}
	In this section, we derive the expansion $\Phi(\mcg_{\mathcal{K}})$ of the particle-conserving East model, corresponding to the cut depicted in Fig.~\ref{fig:East_Krylov}(c). Physically, this corresponds to choosing an initial state where all but the first site are empty [the bottom left vertex in Fig.~\ref{fig:East_Krylov}(c)]. Recall the definition of the expansion for a subregion $R$:
	\bea 
     \label{eq:east_expansion}
	\Phi(R)&\equiv\frac{1}{|R|}\sum_{\psi\in R,\psi'\in R^c} \mcm_{\psi, \psi'}. \eea
	As detailed in Ref.~\cite{PhysRevB.108.144308} and summarized in Fig.~\ref{fig:East_calculation}, the allowed configurations in a given Krylov sector can be mapped to the combinatorial problem of counting the number of allowed monotonic paths on a 2D square lattice. As a simple example, consider the size of the Krylov sector generated by
	\begin{equation}
		\underbrace{\bullet \bullet \bullet \cdots \bullet}_N \underbrace{\circ \circ \cdots \circ}_{L-N},
	\end{equation}
	where $L\leq 2N$. Due to the kinetic constraint, for any subregion $A=[1,k]$, there cannot be more empty sites than occupied sites. The problem is thus equivalent to counting the number of $L$-step walks on a 2D square lattice, where the path is not allowed to touch the line $y=x+1$. This number is easily seen to be given by
	\begin{equation}
		{L \choose N} - {L \choose N+1}.
	\end{equation}
	First consider the numerator of Eq.~(\ref{eq:east_expansion}). It is easy to see that, as long as the site coupled to the bath is empty, the state will be connected to $R^c$. Hence, the number contributing to the numerator is obtained by counting the number of allowed paths \textit{with the first two steps fixed}: 
	\begin{equation}
		\begin{aligned}
			\sum_{\psi\in R,\psi'\in R^c}\mcm_{\psi, \psi'} &={2N_0-2 \choose N_0-1}-{2N_0-2 \choose N_0-2}\\
			&=\frac{(2N_0-2)!}{N_0!(N_0-1)!}.
		\end{aligned}
	\end{equation}
	For the denominator, one needs to sum over all sectors belonging to $R$. Let us first sum over a column of vertices in Fig.~\ref{fig:East_Krylov}(c) with a fixed total charge $N_0-i$ ($0\leq i \leq N_0-1$) and varying $x$ [$2(N_0-i)-1 \leq x \leq 2N_0-1$]. We claim that the sum of the sizes of the Krylov sectors in one column $\sum_{x=2(N_0-i)-1}^{2N_0-1}|\mathcal{K}(N_0-i,x)|$ is given by the number of $2N_0$-step walks from the origin $(0,0)$ to $(N_0-i, N_0+i)$ on the 2D square lattice that do not cross the line $y=x+1+2i$, as shown in Fig.~\ref{fig:East_calculation}(c) and (d). This can be justified as follows. First, for any configuration in this column where the particles can reach a maximal distance of $x$, one can append $2N_0-x$ empty sites to the right, so that all walks have the same length $2N_0$, which is not constrained by $x$. Then, it is clear that each length-$2N_0$ walks from the origin to $(N_0-i, N_0+i)$ is in one-to-one correspondence with a configuration belonging to this column, provided that the following constraint is satisfied:
	\begin{equation}
		({\rm \# \ of \ holes }) - ({\rm \# \ of \ particles }) \leq 2i.
	\end{equation}
	As a result, the allowed paths cannot touch the line $y=x+1+2i$, which can be enforced by the same path reflection trick in Ref.~\cite{PhysRevB.108.144308}. Hence, the summation over a column of vertices in Fig.~\ref{fig:East_Krylov}(c) yields
	\begin{equation}
		\sum_{x=2(N_0-i)-1}^{2N_0-1}|\mathcal{K}(N_0-i,x)|={2N_0 \choose N_0-i}-{2N_0 \choose N_0-i-1}.
	\end{equation}
	Therefore, we have the denominator
	\begin{equation}
		\begin{aligned}
			|R|&=\sum_{i=0}^{N_0-1}\sum_{x=2(N_0-i)-1}^{2N_0-1}|\mathcal{K}(N_0-i,x)|\\
			&= {2N_0 \choose N_0}-1\approx\frac{2N_0!}{N_0!N_0!}.
		\end{aligned}
	\end{equation}
	The expansion is thus given by
	\begin{equation}
		\Phi(R)=\frac{1}{2(2N_0-1)}.
	\end{equation}

	\bibliography{refs}
\end{document}

%% file: ethans_preamble.tex
\usepackage{graphicx}
\usepackage[nointegrals]{wasysym} %
\usepackage[export]{adjustbox}
\usepackage{amsmath,amsfonts,amssymb,latexsym}
\usepackage{hhline}
\usepackage{bm}
\usepackage{verbatim}
\usepackage{enumitem}
\usepackage{mathrsfs}
\usepackage{slashed}
\usepackage{empheq}

\newcommand{\Tr}{\text{Tr}}

\newcommand{\p}{\partial}

\newcommand{\lan}{\langle}
\newcommand{\ran}{\rangle}

 \newcommand{\unit}{\mathbf{1}}

\newcommand{\doa}{\downarrow}
\newcommand{\upa}{\uparrow}

\newcommand{\ra}{\rightarrow}
\newcommand{\lra}{\leftrightarrow}

\newcommand{\<}{\langle}
\renewcommand{\>}{\rangle}
\renewcommand{\(}{\left(}
\renewcommand{\)}{\right)}

\newcommand{\tp}{\otimes}

\newcommand\bpm            {\begin{pmatrix}}
	\newcommand\epm           {\end{pmatrix}}

\def\app#1#2{%
	\mathrel{%
		\setbox0=\hbox{$#1\sim$}%
		\setbox2=\hbox{%
			\rlap{\hbox{$#1\propto$}}%
			\lower1.1\ht0\box0%
		}%
		\raise0.25\ht2\box2%
	}%
}

\newcommand{\tw}{\textwidth}

\newcommand{\vp}{\varphi}

\renewcommand{\cp}{\Phi}
\newcommand{\ct}{\Theta}

\newcommand{\inv}{^{-1}}

\newcommand{\ope}\odot

\usepackage{manfnt}

\newcommand{\bi}{\begin{itemize}}
	\newcommand{\ei}{\end{itemize}}

\usepackage{amsthm}
\newtheorem{theorem}{Theorem}

\newtheorem{proposition}{Proposition}
\newtheorem{fact}{Fact}

\newtheorem{lemma}{Lemma}
\newtheorem{conjecture}{Conjecture}

\theoremstyle{definition}

\newcommand\bpro		  {\begin{proposition}}
\newcommand\epro 		  {\end{proposition}}
\newcommand\bproof			  {\begin{proof}}
\newcommand\eproof 		  {\end{proof}}

	\newcommand\ed            {\end{definition}}

\newcommand\be            {\begin{equation}}
	\newcommand\ee            {\end{equation}}
\newcommand\ba            {\begin{aligned}}
	\newcommand\ea            {\end{aligned}}
\newcommand\bea{\begin{equation}\begin{aligned}}
		\newcommand\eea{\end{aligned}\end{equation}}
\usepackage[dvipsnames]{xcolor}

\usepackage{hyperref} 
\hypersetup{final}
\hypersetup{colorlinks, citecolor=blue, linkcolor=black, urlcolor=blue} 

\newcommand{\sss}{\subsubsection}
\renewcommand{\ss}{\subsection}

\renewcommand{\a}{\alpha}
\renewcommand{\b}{\beta}
\renewcommand{\d}{\delta}

\newcommand{\s}{\sigma}

\newcommand{\ep}{\varepsilon} %
\renewcommand{\l}{\lambda}

\renewcommand{\O}{\Omega}

\renewcommand{\r}{\rho}

\newcommand{\nn}{\mathbb{N}}

\newcommand{\mcc}{\mathcal{C}}
\newcommand{\mck}{\mathcal{K}}
\newcommand{\mcb}{\mathcal{B}}

\newcommand{\mco}{\mathcal{O}}
\newcommand{\mcu}{\mathcal{U}}

\newcommand{\mcd}{\mathcal{D}}

\newcommand{\mcg}{\mathcal{G}}

\newcommand{\mch}{\mathcal{H}}
\newcommand{\mca}{\mathcal{A}}

\newcommand{\mcm}{\mathcal{M}}
\newcommand{\mcn}{\mathcal{N}}

\usepackage[mathscr]{eucal} %

\newcommand{\tta}{\mathtt{a}}
\newcommand{\ttb}{\mathtt{b}}

\newcommand{\tte}{\mathtt{e}}

\newcommand{\ttg}{\mathtt{g}}
\newcommand{\tth}{\mathtt{h}}

\newcommand{\kb}[2]{|{#1}\rangle\langle{#2}|}
\renewcommand{\k}[1]{|#1\rangle}
\newcommand{\ket}[1]{|#1\rangle}
\newcommand{\proj}[1]{|#1\rangle\langle#1|}
\newcommand{\bra}[1]{{\langle #1}|}

\usepackage{dcolumn}
\usepackage{pifont}
\usepackage{ulem}